\documentclass[11pt,reqno]{amsart}
\usepackage{dsfont, amssymb,amsmath,amscd,latexsym, amsthm, amsxtra,amsfonts,enumerate}

\usepackage{caption,subcaption}
\usepackage[all]{xy}
\usepackage[active]{srcltx}
\usepackage{mathrsfs}
\usepackage[all]{xy}
\usepackage[active]{srcltx}
\usepackage{graphicx}
\usepackage{bm}
\usepackage{bbm}
\usepackage{tabularx}
\usepackage{caption}
\usepackage{color}
\usepackage{algorithmicx,algorithm}
\usepackage{longtable}
\usepackage[round]{natbib}
\usepackage{ulem}
\usepackage{verbatim}
\usepackage{appendix}
\usepackage{color}
\usepackage[pdfstartview=FitH, bookmarksnumbered=true,bookmarksopen=true, colorlinks=true, pdfborder=001, citecolor=blue, linkcolor=blue,urlcolor=blue]{hyperref}
\usepackage{CJK}
\newtheorem{theorem}{Theorem}[section]

\newtheorem{corollary}[theorem]{Corollary}
\newtheorem{proposition}[theorem]{Proposition}

\newtheorem{remark}[theorem]{Remark}

\renewcommand{\theequation}{\arabic{section}.\arabic{equation}}

\newtheorem{Them}{Theorem}[section]

\newtheorem{Remark}{Remark}[section]

\newcommand{\E}{\mathrm{E}}
\renewcommand{\theequation}{\arabic{section}.\arabic {equation}}

\numberwithin{equation}{section}
\usepackage{bm}
\usepackage{epstopdf}
\begin{document}
 \makeatletter
 \def\@setauthors{%
    \begingroup
    \def\thanks{\protect\thanks@warning}%
    \trivlist \centering\footnotesize \@topsep30\p@\relax
    \advance\@topsep by -\baselineskip
    \item\relax
    \author@andify\authors
    \def\\{\protect\linebreak}%
    {\authors}%
    \ifx\@empty\contribs \else ,\penalty-3 \space \@setcontribs
    \@closetoccontribs \fi
    \endtrivlist
    \endgroup } \makeatother
 \baselineskip 16pt
 \title[{{\tiny Optimal mix among PAYGO, EET and individual savings }}]
 {{\tiny Optimal mix among PAYGO, EET and individual savings }}
 \vskip 10pt\noindent
 \author[{Lin He, Zongxia Liang, Zhaojie Ren, Yilun Song}]
 {\tiny { \tiny  Lin He$^{a,\dag}$, Zongxia Liang$^{b,\ddag}$, Zhaojie Ren$^{b,*}$, Yilun Song$^{b,\S}$,}
    \vskip 10pt\noindent
    {\tiny ${}^a$School of Finance, Renmin University of China, Beijing
        100872, China \vskip 10pt\noindent\tiny ${}^b$Department of
        Mathematical Sciences, Tsinghua University, Beijing 100084, China }
    \footnote{\\
    	$ \dag$ email: helin@ruc.edu.cn\\
    	$ \ddag$ email: liangzongxia@mail.tsinghua.edu.cn\\
    	$*$ Corresponding author, email:  rzj20@mails.tsinghua.edu.cn\\
    	$\S$   email: songyl18@mails.tsinghua.edu.cn
    	   }}
 \maketitle
 \noindent
\begin{abstract}
   In order to deal with the aging problem, pension system is actively transformed into the funded scheme. However, the funded scheme does not completely replace PAYGO (Pay as You Go) scheme and there exist heterogeneous mixes among PAYGO, EET (Exempt, Exempt, Taxed) and individual savings in different countries. In this paper, we establish the optimal mix by solving a Nash equilibrium between the pension participants and the government. Given the obligatory PAYGO and EET contribution rates, the participants choose the optimal asset allocation of the individual savings and the consumption policies to achieve the objective. The results extend the ``Samuelson-Aaron" criterion to age-dependent preference orderings. And we identify three critical ages to distinguish the multiple outcomes of preference orderings based on heterogeneous characteristic parameters. The government is fully aware of the optimal feedback of the participants. It chooses the optimal PAYGO and EET contribution rates to maximize the overall utility of the participants weighted by each cohort's population. As such, the negative population growth rate leads to the decline of the PAYGO attractiveness as well as the increase of the older cohorts' weight in the government decision-making. The optimal mix is the comprehensive result of the two effects.
%In this paper, we study the optimal mix among PAYGO (Pay as You Go), EET (Exempt, Exempt, Taxed) and fully
%    funded pension. The three pension have heterogeneous
%    characteristics in longevity risk sharing, return efficiency and
%    preferential taxation. These characteristic parameters and the age
%    of the participant jointly determine the preference ordering of each
%    cohort. In fact, the participant's value function is monotonous to
%    the PAYGO and EET contribution rates, which are  given obligatorily
%    by the government. The preference ordering among the three pension
%    is determined by the signs of three closed-form coefficients, and we
%    establish three critical ages as the boundaries to distinguish the
%    preference accordingly. Meanwhile, the optimal PAYGO and EET
%    contribution rates are obtained by the government and the
%    participants of different cohorts based on optimizing a Nash equilibrium.
%    Particularly, the objective of the government is to maximize the
%    overall utility of the participants weighted by each cohort's
%    population. As such, the negative population growth rate leads to
%    the decline of the PAYGO attractiveness as well as the increase of
%    the older cohorts' weight in the government decision-making. Thus,
%    taking PAYGO as the main obligatory pension is still an optimal
%    choice, even considering the trend of population shrinking.
\end{abstract}
\vskip 15pt
\setcounter{equation}{0}
\section{{ {\bf Introduction}}}
Along with the decline of
fertility rate and the enlarge of life expectancy, the
sustainability and the return efficiency of PAYGO pension are facing
great challenges. Thus, the government starts to provide the funded
pension as an alternative option  (obligatorily or voluntarily). In order to improve the participation rate of the funded pension and guarantee the old-age welfare, preferential taxation policy is usually provided. EET pension is the one that contributions are exempt from tax, investment returns are exempt from tax, but the proceeds of pension savings are taxable. Accordingly, it exhibits tax saving properties. Usually, the pension system is composed of three pillars. The first pillar is typically PAYGO pension. The second pillar and part of the third pillar are EET pension. And the rest part is individual savings. Interestingly, we observe that the funded pension (even EET) does not completely replace PAYGO pension under the scenario of shrinking population and serious aging problem. And there exist
heterogeneous mixes of the pension schemes in different countries. For example, EET constitutes the majority of the pension in the U.S., the U.K. and the Netherlands. Meanwhile, PAYGO constitutes the
majority in China, Germany and France (cf. \cite{PC2003} and \cite{BCLM2017}).
\vskip 5pt
The optimal mix between the funded and unfunded pensions has been an
important and controversial topic for decades. According to the
classical ``Samuelson-Aaron" criterion in \citet{SAM1958} and
\citet{AAR1966}, there exists an exclusive optimal pension scheme
determined by the salary growth rate, population growth rate and
investment return. \citet{Merton1983} originally explores PAYGO
pension as a vehicle to make the labor capital exchangeable. Thus,
it contributes to improve the market completeness and the
individual's overall utility. Later, PAYGO pension is  widely
recognized as a ``quasi-asset". The mix of PAYGO pension and funded
pension helps to improve the participants' welfare due to risk
diversification and longevity risk sharing under the mean-variance
objective and the random death settings (cf.
\citet{HL1997}, \citet{DKO2000}, \cite{DMS2006},  \citet{BB2009}, \citet{CDP2011},
\citet{GLS2012} and \citet{BRV2012}).
In this paper, we assume that the major risks are perfectly positively correlated. Thus, we can eliminate the risk diversification effect and extend the ``Samuelson-Aaron" criterion accordingly. That is, each cohort has an age-dependent preference ordering among PAYGO, EET and individual savings. Based on this assumption, the existence of the optimal mix is formed by optimizing the government's weighted objectives of different cohorts. Thus, we establish a new explanation for the existence of the optimal mix.
\vskip 5pt
%Particularly, the self-managed investment on the assets could be regarded as the
%share of fully funded pension.
%The results in \citet{HL1997}, \citet{BB2009},
%\citet{CDP2011} and \citet{BRV2012} confirm that the access of PAYGO
%pension can be welfare improving under some circumstances, which are
%determined by the exogenous parameters, such as the characteristics
%of the demography, investment and salary processes. In addition to these
%characteristic parameters, we believe that the age of each cohort
%can also affect the preference ordering among the pensions.
Inspired by \citet{PC2003} and \cite{HLSY2021}, we model the optimal
mix as the optimum of the Nash equilibrium between the participants
and the government. The Nash equilibrium is originally established
by \citet{NASH1950} and later studied by \cite{SEL1965}. %\citet{HAU1976} and\citet{VanDAM1986}.
%that proposed subgame perfect Nash equilibrium.
The Nash equilibrium is used to depict the  %simultaneous
non-cooperative game between two sides.  And %Given one side's policy fixed, the other side's policy is optimal with respect to this, and vice verse.
the subgame perfect Nash equilibrium is the Nash equilibrium that does not involve any non-credible threat. The settings of this paper perfectly fit the subgame perfect Nash equilibrium model, and we solve it by backward induction, starting at the end of the dynamic game and reasoning backwards step by step. Given the PAYGO and EET contribution
rates, the participants dynamically choose the optimal asset
allocation of the individual savings and the consumption policies to maximize the overall utility
of the future consumption. Moreover, %Meanwhile,
the government is fully aware of the participants' optimal feedback functions with respect to any given contribution rates. The government chooses optimal PAYGO and
EET contribution rates to maximize the weighted utility of the
participants based on the optimal feedback functions. As such, the shrinking population has two effects on the optimal mix. One is that it reduces the attractiveness of PAYGO pension. The other is that it increases the older cohorts' weight in the government decision-making. The comprehensive impacts help to explain the coexistence of multiple pension mixes in different countries.
\vskip 5pt
The characteristic parameters are decisive in determining the
participants' preference ordering among the pension schemes. We first
establish the stochastic differential equations to depict the
characteristics of PAYGO, EET and individual savings in longevity
risk sharing, return efficiency and preferential taxation. For PAYGO
pension, it is collectively managed and we explore the two-stage
process as in \citet{BDOW2004}, \citet{JIN2010}, \citet{WLS2018}
and \cite{CHX2018} to depict the contribution before retirement and
the benefit after retirement. Obviously, the survival participants
could continuously receive the benefit and thus the longevity risk
is managed by intergenerational transfer. Because of the increasingly shrinking
population and serious aging problem, the comparative efficiency of this
``quasi-asset" declines. For EET pension, it is accumulated
in the personal account. Meanwhile, the investment is operated by
the professional institutions and the comparative efficiency is
relatively higher than the one of the individual savings. Particularly, EET
pension is only charged by a relatively lower tax rate after
retirement and thus it has preferential taxation properties. Because
EET pension is obligatorily provided by the government, it should
exhibit the longevity risk sharing properties. As such, we assume that EET
pension is mandatorily converted into life annuity at retirement. Thus,
EET pension has less flexibility and lower comparative efficiency after
retirement, compared with individual savings. Following the settings
in \cite{DKO2000}, \cite{DM2015} and \cite{HLSY2021}, we treat the private investment
as the individual savings. Besides, we assume that the sum of the PAYGO and EET contribution
rates can not exceed a maximum.
\vskip 5pt
For the participants, the closed-form value function  could be
established by using variational methods and  HJB
(Hamilton-Jacobi-Bellman) equations. Interestingly, the value
function is monotonous to the PAYGO and EET contribution rates. And
the signs of three coefficients which measure the comparative efficiencies of every two pensions
% $\widetilde{M_{1}}(\zeta),
%\widetilde{M_{2}}(\zeta)$ and $\widetilde{M_{1}}(\zeta)-\widetilde{M_{2}}(\zeta)$ in
%Eq.\eqref{equ14.1}-Eq.\eqref{equ16.1}
determine the preference orderings. Moreover, the comparative efficiencies among the pensions are
heterogeneous for different cohorts. For example, the younger
cohorts prefer EET pension for its effective investment return during the
long accumulation period. Meanwhile, the older cohorts prefer PAYGO
pension because that the rise of the PAYGO contribution rate
increases their benefit directly. Precisely, we establish three
critical ages as the boundaries to distinguish the preference in Theorems \ref{them2}-\ref{them4}, and
eventually set up each cohort's preference ordering based on the
characteristic parameters. According to the wide range of the
characteristic parameters, there are multiple outcomes of the
preference ordering, and thus we summarize it in the flowchart of Fig. \ref{flow}. Particularly, two outcomes well depict the scenarios in the U.S. and China.
These are the main contributions of this paper.
\vskip 5pt
For the government, it chooses the optimal
PAYGO and EET contribution rates to maximize the weighted utility of
all the participants. Naturally, it is the comprehensive result of
considering the heterogeneous preference orderings of the different cohorts.
Thus, the weight parameter plays an important role in determining
the optimal mix. \citet{HS1989} and \citet{MV1996} use
population as the weight parameter and find that aging leads to the
rise of the older participants' political power. Following the ideas,
we also take the population as the weight parameter. Echoing the
multiple outcomes of preference orderings and the complicated
impacts induced by the weight parameter (or the population growth
rate), there will be multiple optimal mixes. This result contributes
to explain the coexistence of multiple pension mixes in different
countries. Particularly, the negative population growth rate reduces
the attractiveness of PAYGO pension. However, it leads to the rise
of the older cohorts' weight parameter in the government
decision-making. And these eventually lead to the optimal mixes of half PAYGO and half EET in the U.S., and the majority of PAYGO in China. Thus, even considering the
shrinking population and serious aging problem, PAYGO pension is still an important part of the old-age security system. In addition, we consider the model of equal weighted objective. Under this objective, we obtain similar results and the optimal PAYGO contribution rate slightly shrinks. Furthermore, the sum of the PAYGO and EET contribution rates reaches the maximum in the two cases. In these circumstances, the welfare achieved by the obligatory pensions exceeds the one by the voluntary savings.
%achieved by the obligatory pensions exceeds the one by the voluntary
%savings.
%the characteristic parameters of  the two obligatory pensions are
%superior, and  thus the sum of the two optimal contribution rates
%reaches the maximum in the baseline model. In this circumstance, the welfare
%achieved by the obligatory pensions exceeds the one by the voluntary
%savings.
 %it is still optimal
%to take  as the main obligatory pension. Furthermore,
%the characteristic parameters of  the two obligatory pensions are
%superior, and  thus the sum of the two optimal contribution rates
%reaches the maximum in the baseline model. In this circumstance, the welfare
%achieved by the obligatory pensions exceeds the one by the voluntary
%savings.
\vskip 5pt
We also provide two modifications of the model settings to depict the more general actual situations.
First, considering the fact that EET pension is voluntarily participated in
some countries, we also study this case. In the new Nash
equilibrium, the participants choose the asset allocation policies of individual savings,
consumption and a constant contribution rate of EET pension to achieve
the objective under the assumption that the PAYGO contribution rate
is given. Moreover, %Meanwhile,
the government chooses optimal PAYGO
contribution rate to maximize the overall utility of the
participants. Fortunately, this is a
degenerate case of the previous optimization problem and the
closed-form solution can be obtained. Second, we explore a new demographic model with a one-off shock to depict the ``baby-boom" impacts during 1946$\sim$1964. Although there are no closed-form results, we can establish the similar preference orderings  and the optimal mix through numerical methods.
\vskip 5pt
The main contribution of this paper is threefold. First, we
originally establish the stochastic differential equations to depict
the characteristics of PAYGO, EET and individual savings. Particularly,  EET pension has special properties in preferential taxation, return efficiency and longevity risk sharing.
Moreover, these characteristic parameters are decisive in determining the
preference orderings among the three pensions. Under the perfect risk correlation assumption, the participants' value function is monotonous to the
contribution rates. Thus, we extend the ``Samuelson-Aaron" criterion to the age-dependent preference orderings. And each cohort's preference ordering is determined by the comparative efficiencies of the three pension schemes. Moreover, we identify three critical ages as the preference
boundaries and eventually obtain the multiple outcomes of the preference orderings based on the heterogeneous
characteristic parameters. Second, we establish
the Nash equilibrium between the participants and the government to
determine the optimal PAYGO and EET contribution rates. The optimal mix is the comprehensive result of each cohort's preference ordering and the political power. As such, we provide a reasonable explanation for the coexistence of multiple optimal mixes. The last, along with the fact of
shrinking population, the attractiveness of PAYGO pension and the
political power of the younger cohorts who do not prefer PAYGO
pension both decline.  The results show that as countries with serious aging problem, the United States should increase the share of PAYGO pension to meet the needs of the older cohorts, while China should increase the share of  EET pension to improve the overall return efficiency.
%
%
%
%
%the optimal PAYGO and EET contribution
%rates are the comprehensive result of the two effects. The theoretical results under the calibrated parameters are consistent with the ones in China and the U.S..
\vskip 5pt
The remainder of this paper is organized as follows. Section \ref{s2}
establishes the stochastic optimization problem of the participants
with the access of the obligatory PAYGO and EET pensions. And the closed-form value
function is obtained by variational methods. In
Section \ref{ordering}, based on the monotonicity of the value function to the
PAYGO and EET contribution rates, we establish three critical ages.
Moreover, we summarize the multiple outcomes of the age-dependent preference
orderings among the three pensions in the flowchart of Fig. \ref{flow}. Section \ref{government} establishes
the Nash equilibrium between the participants and the government.  In Section \ref{vgovernment}, we
add a new Nash equilibrium with the voluntary EET pension. In Section \ref{s6}, we do some numerical simulations, and Section \ref{s7} concludes this paper.

%Therefore, the
%demographic trend affects not only the preference between the two
%schemes, but also the importance of each cohort.

 \vskip 15pt
\setcounter{equation}{0}

\section{{ {\bf The optimization problem of the pension participants}}}\label{s2}
We consider a two-stage game of complete and perfect information between the participants and the government. In the first stage, the government determines the PAYGO and EET contribution rates. In the second stage, given the obligatory PAYGO and EET contribution rates, the participants choose the optimal asset allocation of the individual savings and the consumption policies to maximize the overall utilities. Using backward induction, we first solve the optimization problem of the participants in the second stage.
%We consider a Nash equilibrium between the participants and the government, where the participants and the government simultaneously solve an optimization problem by being fully aware of each other's optimal feedbacks.
 This can be formulated as a stochastic control problem in Subsection \ref{sec2.3} and solved in Subsection \ref{sec2.4}. Using the variational method, we obtain the closed-form solution and the value function of the optimization problem.

First, considering the access of PAYGO and EET pensions, we will establish a unified stochastic differential equation that the participant's wealth process satisfies. Particularly, we assume that the major risks are perfectly positively correlated. Thus, we can eliminate the risk diversification effect and extend the ``Samuelson-Aaron" criterion accordingly. The unified risk is described as the following standard Brownian motion.
\vskip 5pt
Let $(\Omega,\mathcal{F},\{\mathcal{F}_t\}_{0\leq t\leq T},P)$ be a filtered complete probability space satisfying the usual conditions. The filtration $\{\mathcal{F}_t\}_{0\leq t\leq T}$ is generated by a standard Brownian motion $B=\{B(t),0\leq t\leq T\}$ on $(\Omega,\mathcal{F},P)$, i.e., $\mathcal{F}_t=\sigma\{B(s),0\leq s\leq t\}$ which represents the information until time $t$.
\subsection{Market model}
We assume that there are two assets in the financial market, a risk-free asset and a risky asset. The prices of the risk-free asset and the risky asset satisfy the following stochastic differential equations (SDEs):
\begin{align*}
	&\left\{
	\begin{array}{lr}
		dS_0(t)=rS_0(t)dt, \\
		S_0(0)=s_0,
	\end{array}
	\right.
	\\
	&\left\{
	\begin{array}{lr}
		dS_1(t)=\mu S_1(t)dt+\sigma S_1(t)dB(t),   \\
		S_1(0)=s_1,
	\end{array}
	\right.
\end{align*}
where $r$ represents the risk-free interest rate,  $\mu$ and $\sigma$ are the expected return and the volatility of the risky asset, respectively. $s_0$ and $s_1$ are the prices of the risk-free asset and the risky asset at time $0$. Without loss of generality, we assume  $\mu>r>0$ and $\sigma>0$.
\vskip 5pt
The average salary $W=\{W(t),0\leq t \leq T\}$ satisfies the following SDE:
\begin{align}
	\left\{
	\begin{array}{lr}
		dW(t)=\gamma W(t)dt+\xi W(t)dB(t),  \\
		W(0)=W_0,
	\end{array}
	\right.\label{equ1}
\end{align}
where $\gamma$ and $\xi$ are the expected growth rate and the volatility of the average salary, respectively. $W_0$ is the average salary at time $0$.
\vskip 5pt

Notably, Brownian motion $B=\{B(t),0\leq t\leq T\}$ is the common stochastic source of the risky asset, the average salary and the EET pension return thereafter. The risks are all influenced by the economic capacity and vitality in practice. Moreover, the closed-form solution of the optimization problem in Theorem \ref{them1} and the exclusive preference orderings in Theorems \ref{them2}-\ref{them4} are only valid under the perfect correlation assumption.

\subsection{Participant's wealth process}
%First, we introduce the optimization problem of the pension participants.
The life cycle of the participants is divided into two periods: the working period and the retirement period. While working, the participants earn salary and contribute two fixed proportions to PAYGO and EET pensions obligatorily. After retirement, the participants no longer earn salary and they benefit from PAYGO and EET pensions in different ways. Furthermore, the investment of the individual savings is permitted throughout the life cycle. Therefore, we first establish the stochastic differential equations that the participant's wealth satisfies respectively in the two periods, and then establish a unified equation.
\vskip 5pt
For the demographic model, we assume a constant negative population growth rate to depict the scenario of shrinking population. Although the demographic trend changes sometimes, it will remain unchanged for the subsequent period of time. Moreover, the volatility of the population growth rate is relatively small, compared with the risky investment return and the salary growth rate. Thus, we explore the constant population growth rate in the baseline model. And we introduce a one-off shock model to depict the ``baby-boom" effect in Appendix \ref{babyboom}. However, under this model, we cannot obtain the preference orderings analytically like in Theorems \ref{them2}-\ref{them3}. Besides, in the numerical analysis part, the mortality rate is decreased to study the impacts of longevity risk.
\vskip 5pt
We assume that the participants join the pension at age $a$, retire at age $\tau$ and the maximal survival age is $\omega$. The number of the new entrants aged $a$ joining the pension at time $t$ is as follows:
\begin{align*}
    n(t)=n_0e^{\rho t},
\end{align*}
where $n_0$ is the number of the new entrants at time $0$ and $\rho$ is the constant negative population growth rate. %Particularly, to depict the scenario of low fertility rate and serious aging problem, $\rho<0$ usually holds.
\vskip 5pt
%{\color{red}Notably, for other population models, while the closed-form solution of the optimization problem and the preference ordering between EET pension and individual saving exist, we cannot be able to obtain the preference orderings analytically in Theorem \ref{them2}-\ref{them3}. For this case, we elaborate in the Appendix \ref{babyboom}, in particular, taking the example of the population model with ``baby-boom''.}
\vskip 5pt
Meanwhile, we assume that the force of mortality follows Makeham's law (cf. \cite{DHW2013}). As such, the survival function which is the conditional probability that a person who is alive at age $a$ is still alive at age $x$ is as follows:
\begin{align}
s(x)=e^{-A(x-a)-\frac{B}{\ln c}(c^x-c^a)}\label{mortality},
\end{align}
where $A, B$  and  $ c$ are constants. Moreover, referring to the model in \cite{Rong2023}, longevity trends can be captured by decreasing mortality rates over time and age. And the effect of longevity on the critical age and optimal contribution rates is discussed later in the sensitivity analysis. Thus, the amounts of the participants in the working and the retirement periods at time $t$ are respectively as follows:
\begin{align*}
    \Theta(t)&=\int_{a}^{\tau}n(t-u+a)s(u)du\\
    %&=\int_{a}^{\tau}n_0e^{\rho(t-u+a)}e^{-A(u-a)-\frac{B}{\ln c}(c^u-c^a)}du\\
    &=n_0e^{\rho t}\int_{a}^{\tau}e^{-(\rho+A)(u-a)-\frac{B}{\ln c}(c^u-c^a)}du\\
    &\triangleq n_0e^{\rho t}\Lambda(a,\tau),
\end{align*}
\begin{align*}
    \Sigma(t)&=\int_{\tau}^{\omega}n(t-u+a)s(u)du\\
    %&=\int_{\tau}^{\omega}n_0e^{\rho(t-u+a)}e^{-A(u-a)-\frac{B}{\ln c}(c^u-c^a)}du\\
    &=n_0e^{\rho t}\int_{\tau}^{\omega}e^{-(\rho+A)(u-a)-\frac{B}{\ln c}(c^u-c^a)}du\\
    &\triangleq n_0e^{\rho t}\Lambda(\tau,\omega).
\end{align*}
Note that $\Lambda \triangleq \frac{\Theta(t)}{\Sigma(t)}=\frac{\Lambda(a,\tau)}{\Lambda(\tau,\omega)}$ is independent of time $t$ and $\Lambda^{-1}$ is the dependency ratio of PAYGO pension. $\Lambda$ plays an important role in determining the return efficiency of PAYGO pension.
\vskip 5pt
In this paper, we particularly introduce EET pension. It has the following three characteristics. First,  EET pension is accumulated in the personal account but operated by the professional institutions.  Fortunately, the return efficiency is higher than the one of individual savings. Second, the accumulation of EET pension will be converted into life annuities at retirement to exhibit the longevity risk sharing properties. Third, the participants can enjoy the preferential taxation through deferred tax payment.   As such, the accumulation in EET pension $Y$ satisfies the following SDE for $t\in[z,z+\tau-a]$:
\begin{align}\label{equ2}
	\left\{
	\begin{array}{lr}
		dY(t)=\alpha Y(t)dt+\beta Y(t)dB(t)+kW(t)dt,   \\
		Y(z)=0,
	\end{array}
	\right.
\end{align}
where $\alpha$ and $\beta$ are the expected return and the volatility of EET pension. $z$ is the time that the participant joins the pension.
For the sake of simplicity, we assume that the participants receive the same annual income $A=\{A(t), t\in[z+\tau-a,z+\omega-a]\}$ from the life annuities after retirement, i.e.,
\begin{align}\label{annuity}
%A(t)=\frac{Y(z+\tau-a)}{\omega - \tau}(1-\tau_{2}),
&A(t)=\frac{Y(z+\tau-a)}{a_{\tau}}(1-\tau_2),\\
&a_{\tau}=\int_{0}^{\infty}e^{-rt}\frac{s(\tau+t)}{s(\tau)}dt=\int_{0}^{\infty}e^{-(r+A)t-\frac{B}{\ln c}c^{\tau}(c^t-1)}dt,
\end{align}
where $a_{\tau}$
%$\frac{1}{\omega - \tau}$ (cf. \cite{HL2013})
is the expected present value of one unit annuity. $\tau_2$ is the marginal tax rate after retirement. Moreover, EET pension has higher Sharpe ratio through professional management and diversified investment, compared with the individual savings. Thus, we assume that $\frac{\alpha-r}{\beta}>\frac{\mu-r}{\sigma}$.
\vskip 5pt
We suppose that $\theta$ and $k$ are the PAYGO and EET contribution rates, respectively. In practice, in order to ensure the minimal consumption and reduce the burden of the participants, the sum of $\theta$ and $k$ can not be too large, i.e., $\theta+k\leq m$, $m \in [0,1]$. In the optimization problem of the participants, the participants choose optimal asset allocation and consumption policies based on these two given  parameters. In Section \ref{government}, $\theta$ and $k$ are both obligatorily determined by the government at time $t_0$. In Section \ref{vgovernment}, $k$ is voluntarily determined by the participants and $\theta$ is still obligatorily determined by the government at time $t_0$. Throughout the life cycle, the participants' individual savings is invested in the risk-free and the risky assets. Besides,  we assume that the consumption should be nonnegative.
\vskip 5pt
During the working period, the participants earn salary and pay $\theta$ proportion to PAYGO pension and $k$ proportion to EET pension, respectively.
%The remaining wealth is used for investment and consumption after tax payment.
As such, we consider a participant who joins the pension at time $z$ and the wealth process $X$ of this participant in the working period satisfies the following SDE for $t\in[z,z+\tau-a]$:
\begin{align*}
    \left\{\!
    \begin{array}{lr}
        dX(t)\!=\![rX(t)\!+\!(\mu -r)\pi (t)]dt \!+\!\sigma \pi (t)dB(t)\!+\!(1\!-\!\theta\!-\! k )W(t)(1\!-\!\tau_{1})dt \!-\! C(t)dt,  \\
        X(z)=0,
    \end{array}
    \right.
\end{align*}
where $\pi$ is the amount allocated to the risky asset and $C$ is the consumption process which satisfies $C(t)\geq 0$. $\tau_{1}$ is the marginal tax rate of the salary. And $\tau_1>\tau_2$ usually holds due to the tax saving effect.
%We extend Y from $[z,z+\tau -a]$ to $[z,z+\omega -a ]$
%\begin{align}
%    \left\{
%    \begin{array}{lr}
%        dY(t)=\left\{[\alpha Y(t)+kW(t)]1_{\{t\in [z,z+\tau -a]\}}-\frac{Y(z+\tau-a)}{\omega - \tau}1_{\{t\in[z+\tau -a,z+\omega -a]\}}\right\}dt\\
%        \phantom{eeeeeeee}+\beta Y(t)1_{\{t\in [z,z+\tau -a]\}}dB(t),   \\
%        Y(z)=0.
%    \end{array}\label{equ2}
%    \right.
%\end{align}
\vskip 5pt
During the retirement period, the participants benefit from PAYGO and EET pensions. For PAYGO pension, the benefit depends on the dependency ratio, the PAYGO contribution rate and the current salary. For EET pension, the benefit is the same annual income from the life annuities described by Eq.\eqref{annuity}.
%Investment is also allowed in this period.
% because of the assumption that utility only depends on consumption.
Thus, the wealth process $X$ of the participant in the retirement period satisfies the following SDE for $t\in[z+\tau-a,z+\omega -a]$:
\begin{align}\label{afterretire}
    \left\{
    \begin{array}{lr}
        dX(t)=[rX(t)+(\mu -r)\pi (t)]dt+\sigma \pi (t)dB(t)+\theta \Lambda W(t) dt\\
       \ \ \ \ \ \ \ \ \ \ \  +\frac{Y(z+\tau -a)}{a_{\tau}}(1-\tau_{2})dt-C(t)dt,  \\
        X(z+\omega -a)=0,
    \end{array}
    \right.
\end{align}
where the third and the forth items in Eq.\eqref{afterretire} represent the benefit from PAYGO and EET pensions, respectively. Notably, borrowing and advance consumption are allowed. At the maximal survival age, the borrowed money needs to be fully repaid and the wealth is non-negative (actually zero). In fact, the determined maximal survival age is different from the random death time. Participants are not expected to live beyond the maximal survival age. Negative wealth at this moment is deliberate indebtedness and is therefore forbidden, and positive wealth is contrary to the goal of maximizing one's lifetime consumption. These two conditions are consistent with reality.
%for the sake of optimization
%
%$\tau_2$ is marginal tax rate after tax incentives which is lower than the marginal tax rate of the salary $\tau_1$.
\vskip 5pt
In order to obtain a unified state process of $X$ and make the HJB equation in Eq.\eqref{equ8} well solved, we introduce a technical approach and make the transformation in Eqs.\eqref{afterretire}, \eqref{ty} and \eqref{equ3}.  Let  $Y(z+\tau -a)=Y(t),t\in[z+\tau -a,z+\omega -a]$, and denote
%\frac{Y(z+\tau -a)}{\omega-\tau}=\frac{Y(t)}{(z+\omega -a)-t}
\begin{align*}
&a(t;z,\theta,k) \triangleq (1-\theta-k)(1-\tau_{1})1_{\{t\in [z,z+\tau -a]\}}+\theta\Lambda 1_{\{t\in [z+\tau -a,z+\omega -a]\}}.
%\frac{1}{(z+\omega -a)-t}
\end{align*}
Moreover, the $Y$  satisfies the following SDE:
\begin{align}
	\left\{
	\begin{array}{lr}
		dY(t)=[\alpha Y(t)+kW(t)]1_{\{t\in [z,z+\tau -a]\}}dt+\beta Y(t)1_{\{t\in [z,z+\tau -a]\}}dB(t),   \\
		Y(z)=0.  \\
	\end{array}
	\right.\label{ty}
\end{align}
Thus, we establish the unified stochastic differential equation that the participant's wealth process satisfies for $t\in [z,z+\omega -a]$:
\begin{align}
	\left\{
	\begin{array}{lr}
		dX(t)=[rX(t)+(\mu -r)\pi (t)]dt+\sigma \pi (t)dB(t)+a(t;z,\theta,k)W(t)dt\\
		\phantom{eeeeeeee}+\frac{1-\tau_{2}}{a_{\tau}}Y(t)1_{\{t\in [z+\tau-a,z+\omega -a]\}}dt-C(t)dt,  \\
		X(z)=0, \\
		X(z+\omega -a)=0,
	\end{array}\label{equ3}
	\right.
\end{align}
where $\pi$ and $C$ are the two control variables of the stochastic optimization problem of the participants.
\begin{Remark}\label{unnc}
	Closer to reality, consumption $C(t)$ can also be regarded as unnecessary consumption. Just take $\tilde{C}(t)=C(t)-\mathcal{X}(t)W(t)$, and replace $C(t)$ with $\tilde{C}(t)$. $\mathcal{X}(t)W(t)$ is the necessary consumption to ensure living.
\end{Remark}
\vskip 5pt
\subsection{The stochastic control problem}\label{sec2.3}
We formulate the stochastic control problem of the pension participants as follows:\\
The participant's admissible set $\mathcal{A}^{z,\theta,k}$ is defined by
$$\mathcal{A}^{z,\theta,k}=\left\{(\pi,C) \,\text{satisfies}\,\text{Eq.}\eqref{equ3}:\int_z^{z+\omega -a}[\pi^2(s)+C(s)]ds<\infty\, \,\text{a.s},\, C\geq 0\right\}.$$
The participant controls $(\pi, C)\in \mathcal{A}^{z,\theta,k}$ to maximize the following objective function with the fixed $\theta$ and $k$:
$$\mathrm{E}\left\{\int_{z}^{z+\tau-a}\!e^{\!-\!r(u\!-\!z)}s(u\!-\!z\!+\!a)U(C(u))du\!+\!\lambda \int_{z+\tau -a}^{z+\omega -a}\!e^{\!-\!r(u\!-\!z)}s(u\!-\!z\!+\!a)U(C(u))du\right\},$$
where $\lambda$ is the weight parameter of the utility after retirement. Considering the fact that the same consumption produces higher utility after retirement because of more leisure time, $\lambda>1$ usually holds (cf. \cite{CHX2018}). We also assume that the participants have CRRA utilities, i.e., $U(C)=\frac{1}{\delta}C^{\delta},\delta <1,\delta\neq 0.$ Denote $b(u;z)\triangleq e^{-r(u-z)}s(u-z+a)\left[1+(\lambda-1)1_{\{u\in[z+\tau -a,z+\omega -a]\}}\right].$ Then the stochastic optimization problem of the pension participant becomes
\begin{align}
\max_{(\pi,C)\in \mathcal{A}^{z,\theta,k}}\mathrm{E}\left\{\int_{z}^{z+\omega -a}b(u;z)\frac{C(u;z,\theta,k)^{\delta}}{\delta}du\right\},\label{equ4}
\end{align}
where $C(u;z,\theta,k)$ is the consumption process of the participant who joins the pension at time $z$ with contribution rates of $\theta$ and $k$.
\subsection{Solution to the optimization problem}\label{sec2.4}
We reformulate the optimization problem of the participants. At time $z$, the participants enter the labor market, optimize the objective in Eq.\eqref{equ4} and establish the optimal policies according to the initial $\theta$  and $k$. Following the optimal policies, $X(t)$ and $Y(t)$ realize accordingly. At time $t$, the government may decide to reselect the optimal $\theta$  and $k$ to cope with the new demographic trend, or may not reselect $\theta$ and $k$. In either case, the participants choose the optimal control polices $(\pi, C)$ based on the latest $\theta$ and $k$, as well as the realized $X(t)$, $Y(t)$ and $W(t)$. The participant dynamically controls $(\pi, C)$ and $W(\cdot)$, $Y(\cdot)$, $X(\cdot)$ become
%We reformulate the optimization problem of the participants. At time $z$, the participants enter the labor market, optimize the objective in Eq.\eqref{equ4} and establish the optimal policies according to the initial $\theta$  and $k$. Following the optimal policies, $X(t)$ and $Y(t)$ realize accordingly. At time $t$, the demographic trend changes significantly. As such, the government decides to reselect the optimal $\theta$  and $k$ to cope with the new trend. Thus, the participants choose the optimal control polices $(\pi, C)$ based on the new given $\theta$  and $k$, as well as the realized $X(t)$, $Y(t)$ and $W(t)$.  The participant dynamically controls $(\pi, C)$ and $W(\cdot)$, $Y(\cdot)$, $X(\cdot)$ become
%\begin{align}
%        &\left\{
%        \begin{array}{lr}
%            dW(s)=\gamma W(s)ds+\xi W(s)dB(s),\\
%            W(t)=w,
%        \end{array}
%        \right.
%        \\
%    &\left\{
%    \begin{array}{lr}
%        dY(s)=\left\{[\alpha Y(s)+kW(s)]1_{\{s\in [z,z+\tau -a]\}}-\psi(s;z)Y(s)\right\}ds\\
%        \phantom{eeeeeeee}+\beta Y(s)1_{\{s\in [z,z+\tau -a]\}}dB(s),   \\
%        Y(t)=y,  \\
%    \end{array}\label{newy}
%    \right.
%    \\
%    &\left\{
%    \begin{array}{lr}
%        dX(s)=[rX(s)+(\mu -r)\pi (s)]ds+\sigma \pi (s)dB(s)+a(s;z,\theta,k)W(s)ds\\
%        \phantom{eeeeeeee}+\psi(s;z)Y(s)(1-\tau_{2})ds-C(s)ds,  \\
%        X(t)=x, \\
%        X(z+\omega -a)=0,\label{equ2.7}
%    \end{array}
%    \right.
%\end{align}
\begin{align}
	&\left\{
	\begin{array}{lr}
		dW(s)=\gamma W(s)ds+\xi W(s)dB(s),\\
		W(t)=w,
	\end{array}
	\right.
	\\
	&\left\{
	\begin{array}{lr}
		dY(s)=[\alpha Y(s)+kW(s)]1_{\{s\in [z,z+\tau -a]\}}ds+\beta Y(s)1_{\{s\in [z,z+\tau -a]\}}dB(s),   \\
		Y(t)=y,  \\
	\end{array}\label{newy}
	\right.
	\\
	&\left\{
	\begin{array}{lr}
		dX(s)=[rX(s)+(\mu -r)\pi (s)]ds+\sigma \pi (s)dB(s)+a(s;z,\theta,k)W(s)ds\\
		\phantom{eeeeeeee}+\frac{1-\tau_{2}}{a_{\tau}}Y(s)1_{\{s\in [z+\tau-a,z+\omega -a]\}}ds-C(s)ds,  \\
		X(t)=x, \\
		X(z+\omega -a)=0,\label{equ2.7}
	\end{array}
	\right.
\end{align}
where $x$, $w$ and $y$ are the realized participant's wealth, average salary and the fund accumulated in EET pension at time $t$, respectively.
Then the participant's admissible set is redefined by
$$\mathcal{A}^{t,x,w,y; z,\theta,k}\!=\!\left\{\!(\pi,C) \,\mbox{satisfies Eq.\eqref{equ2.7}:}\ \int_t^{z+\omega -a}[\pi^2(s)+C(s)]ds<\infty\, \,\text{a.s.},\, C\geq 0\!\right\}\!.$$
And the value function is redefined by
\begin{align}
    V(t,x,w,y;z,\theta,k)=\max_{(\pi,C)\in \mathcal{A}^{t,x,w,y;z,\theta,k}}\mathrm{E}\left\{\int_{t}^{z+\omega -a}b(u;z)\frac{C(u;z,\theta,k)^{\delta}}{\delta}du\right\}.\label{equ5}
\end{align}
Following the ideas of \cite{YZ1999}, we use the dynamic programming method to solve the optimization problem. The main result of this section is summarized in the following Theorem \ref{them1}.
\begin{Them}\label{them1}
        For the fixed $\theta$ and $k$, the solution $(\pi^*,C^*)$ and the value function $V$ of Eq.\eqref{equ5} are given by

        \begin{small}
        \begin{align}
            \pi^*(t,x,w,y;z,\theta,k)\!=\!&-\frac{\sigma \xi w V_{xw}(t,x,w,y;z,\theta,k)\!+\!\sigma \beta y1_{\{t\in [z,z+\tau -a]\}}V_{xy}(t,x,w,y;z,\theta,k)}{\sigma^2V_{xx}(t,x,w,y;z,\theta,k)}\nonumber\\
                &-\frac{(\mu-r)V_x(t,x,w,y;z,\theta,k)}{\sigma^2V_{xx}(t,x,w,y;z,\theta,k)},\label{pi}\\
            C^*(t,x,w,y;z,\theta,k)\!=\!&\left(\frac{V_x(t,x,w,y;z,\theta,k)}{b(t)}\right)^{\frac{1}{\delta-1}},\\
            V(t,x,w,y;z,\theta,k)=&\frac{1}{\delta}L(t;z)[x+(M_1(t;z)\theta+M_2(t;z)k+M_3(t;z))w+N(t;z)y]^{\delta}\label{value},
        \end{align}
        \end{small}
    where  $\nu\triangleq\frac{\mu-r}{\sigma}$, $\epsilon\triangleq\gamma-r-\xi\nu \neq 0$ and $\widetilde{\epsilon}\triangleq\alpha-r-\beta\nu\neq\epsilon$,
    \begin{small}
    	\begin{align}
    		& L(t;z)\triangleq\left(\int_{t}^{z+\omega -a}b(s)^{\frac{1}{1-\delta}}e^{\frac{\delta}{1-\delta}\left[r+\frac{(\mu-r)^2}{2\sigma^2(1-\delta)}\right](s-t)}ds\right)^{1-\delta},   \label{L}                          \\
    	&M_1(t;z)\!\triangleq\!\left\{
    	\begin{array}{lr}
    		\frac{1}{\epsilon}\Lambda\left(e^{\epsilon(z-t+\omega-a)}-1\right),~~&z-t+\tau-a\leq 0,\\
    		\frac{1}{\epsilon}\left[(1-\tau_1)+\left(\Lambda e^{\epsilon(\omega-\tau)}-\Lambda-(1-\tau_1)\right)e^{\epsilon(z-t+\tau-a)}\right],~~&z-t+\tau-a>0,
    	\end{array}
    	\right.\label{equ14}
    	\\
    	&M_2(t;z)\!\triangleq\!\left\{\!
    	\begin{array}{lr}
    		0,\phantom{eeeeeeeeeeeeeeeeeeeeeeeeeeeeeeeeeeeeeeeeeeeeeee}z-t+\tau-a\leq 0,\\
    		\frac{1}{r(\epsilon-\widetilde{\epsilon})}\frac{1-\tau_2}{a_{\tau}}\left(1\!-\! e^{-r(\omega-\tau)}\right)\left(e^{\epsilon(z-t+\tau-a)}\!-\! e^{\widetilde{\epsilon}(z-t+\tau-a)}\right)
    \\-\frac{1-\tau_1}{\epsilon}\left(e^{\epsilon(z-t+\tau-a)}\!-\!1\right),\\
    		\phantom{eeeeeeeeeeeeeeeeeeeeeeeeeeeeeeeeeeeeeeeeeeeeeeeeee}z-t+\tau-a>0,
    	\end{array}
    	\right.\label{equ15}\\
    	&M_3(t;z)\!\triangleq\!\left\{
    	\begin{array}{lr}
    		0,~~&z-t+\tau-a\leq 0,\\
    		\frac{1}{\epsilon}(1\!-\!\tau_1)\left(e^{\epsilon(z-t+\tau-a)}-1\right),~~&z-t+\tau-a>0,
    	\end{array}
    	\right.\label{equ16}\\
    	&N(t;z)\!\triangleq\!\left\{
    	\begin{array}{llr}
    		\frac{1-\tau_{2}}{ra_{\tau}}\left(1-e^{-r(z-t+\omega-a)}\right),~~&z-t+\tau-a\leq 0,\\
    		\frac{1-\tau_{2}}{ra_{\tau}}\left(1-e^{-r(\omega-\tau)}\right)e^{\widetilde{\epsilon}(z-t+\tau-a)},~~&z-t+\tau-a>0.
    	\end{array}
    	\right.\label{N}
    	\end{align}
    \end{small}
\end{Them}
\begin{proof}
We first define
\begin{align*}
    &k(t,x,w,y,\pi,C)  =
    \left(
    \begin{array}{c}
        rx+\pi(\mu-r)+a(t)w+\frac{1-\tau_{2}}{a_{\tau}}y1_{\{t\in [z+\tau-a,z+\omega -a]\}}-C \\
        \gamma w \\
        (\alpha y+kw)1_{\{t\in [z,z+\tau -a]\}}\\
    \end{array}
    \right),
    \\
    &u(t,x,w,y,\pi,C)   =
    \left(
    \begin{array}{c}
        \sigma \pi  \\
        \xi w \\
        \beta y1_{\{t\in[z,z+\tau -a]\}}\\
    \end{array}
    \right)     .
\end{align*}
Then
\begin{eqnarray*}
    \mathrm{d}
    \left(
    \begin{array}{c}
        X(t)  \\
        W(t) \\
        Y(t)\\
    \end{array}
    \right) &=& k(t,X(t),W(t),Y(t),\pi(t),C(t))\mathrm{d}t\nonumber\\
    &&+ u(t,X(t),W(t),Y(t),\pi(t),C(t))\mathrm{d}B(t)\label{equ7} .
\end{eqnarray*}
We establish the HJB equation that $V$ satisfies:
\begin{eqnarray}
         V_t&+&\sup_{\pi,C}\left\{\frac{1}{2}\sigma^2\pi^2V_{xx}+\sigma \xi \pi w V_{xw}+\sigma \beta \pi y 1_{\{t\in [z,z+\tau -a]\}}V_{xy}+\frac{1}{2} \xi^2 w^2 V_{ww}\right.\nonumber\\
         &&\phantom{ee}+\xi \beta wy1_{\{t\in [z,z+\tau -a]\}}V_{wy}+\frac{1}{2}\beta^2y^21_{\{t\in [z,z+\tau -a]\}}V_{yy}\nonumber\\
         &&\phantom{ee}+[rx+\pi(\mu-r)+a(t)w+\frac{1-\tau_{2}}{a_{\tau}}y1_{\{t\in [z+\tau-a,z+\omega -a]\}}-C]V_x+\gamma w V_w\nonumber\\
        &&\phantom{ee}\left.+(\alpha y+kw)1_{\{t\in [z,z+\tau -a]\}}V_y+b(t;z)\frac{C^\delta}{\delta}
        \right\}=0.\label{equ8}
\end{eqnarray}
Applying the first order condition, we obtain the suprema of Eq.\eqref{equ8} as follows:
\begin{align}
    &\pi^*=-\frac{\sigma \xi w V_{xw}+\sigma \beta y1_{\{t\in [z,z+\tau -a]\}}V_{xy}+(\mu-r)V_x}{\sigma^2V_{xx}},\\
    &C^*=\left(\frac{V_x}{b(t)}\right)^{\frac{1}{\delta-1}}.\label{equ10}
\end{align}
In order to obtain the value function, we guess that it has the following form:
\begin{equation}
    V(t,x,w,y;z,\theta,k)=\frac{1}{\delta}L(t;z)[x+M(t;z,\theta,k)w+N(t;z,\theta,k)y]^{\delta}.\label{equ11}
\end{equation}
%Thus
%\begin{align*}
%    &V_t=\frac{1}{\delta}L'(t)[x+M(t)w+N(t)y]^{\delta}\!+\! L(t)[x+M(t)w+N(t)y]^{\delta-1}(M'(t)w+N'(t)y),\\
%    &V_x=L(t)[x+M(t)w+N(t)y]^{\delta-1},\\
%    &V_w=L(t)[x+M(t)w+N(t)y]^{\delta-1}M(t),\\
%    &V_y=L(t)[x+M(t)w+N(t)y]^{\delta-1}N(t),\\
%    &V_{xx}=(\delta -1)L(t)[x+M(t)w+N(t)y]^{\delta-2},\\
%    &V_{xw}=(\delta -1)L(t)[x+M(t)w+N(t)y]^{\delta-2}M(t),\\
%    &V_{xy}=(\delta -1)L(t)[x+M(t)w+N(t)y]^{\delta-2}N(t),\\
%    &V_{ww}=(\delta -1)L(t)[x+M(t)w+N(t)y]^{\delta-2}M^2(t),\\
%    &V_{wy}=(\delta -1)L(t)[x+M(t)w+N(t)y]^{\delta-2}M(t)N(t),\\
%    &V_{yy}=(\delta -1)L(t)[x+M(t)w+N(t)y]^{\delta-2}N^2(t).
%\end{align*}
Substituting Eq.\eqref{equ11} into Eq.\eqref{equ8}, $L(t;z),M(t;z,\theta,k)$ and $N(t;z,\theta,k)$ satisfy the following backwards ordinary differential equations (BODEs):
\begin{equation}
	\left\{
	\begin{array}{lr}
		\frac{1}{\delta}L'(t)+\left[r+\frac{(\mu-r)^2}{2\sigma^2(1-\delta)}\right]L(t)+\frac{(1-\delta)L(t)^{\frac{\delta}{\delta-1}}}{b(t)^{\frac{1}{\delta-1}}\delta}=0,&\\
		L(z+\omega -a)=0,&\\
		M'(t)+\left[-r-\frac{1}{\sigma}\xi(\mu-r)+\gamma\right]M(t)+a(t)+N(t)k1_{\{t\in[z,z+\tau -a]\}}=0,& \\
		M(z+\omega -a)=0,&\\
		N'(t)\!+\!\left[-r-\frac{1}{\sigma}\beta1_{\{t\in[z,z+\tau -a]\}}(\mu-r)+\alpha1_{\{t\in[z,z+\tau -a]\}}\right]\! N(t)&\\+\frac{1-\tau_{2}}{a_{\tau}}1_{\{t\in [z+\tau-a,z+\omega -a]\}}=0,&\\
		N(z+\omega -a)=0.&\label{equ12}
	\end{array}
	\right.
\end{equation}
Solving the last three BODEs, we obtain
\begin{small}
	\begin{align*}
		& L(t;z)=\left(\int_{t}^{z+\omega -a}b(s)^{\frac{1}{1-\delta}}e^{\frac{\delta}{1-\delta}\left[r+\frac{(\mu-r)^2}{2\sigma^2(1-\delta)}\right](s-t)}ds\right)^{1-\delta},                             \\
		& N(t;z)=\frac{1-\tau_{2}}{a_{\tau}}\int_{t}^{z+\omega -a}1_{\{s\in [z+\tau-a,z+\omega -a]\}}e^{\int_{t}^{s}\left[-r-\frac{1}{\sigma}\beta1_{\{u\in[z,z+\tau -a]\}}(\mu-r)+\alpha1_{\{u\in[z,z+\tau -a]\}}\right]du
		}ds, \\
		& M(t;z,\theta,k)=\int_{t}^{z+\omega -a}[a(s)+N(s)k1_{\{s\in[z,z+\tau -a]\}}]e^{\left(r+\frac{1}{\sigma}\xi(\mu-r)-\gamma\right)(t-s)}ds.
	\end{align*}
\end{small}

Moreover, $N(t;z)$ and $M(t;z,\theta,k)$ have the following specific forms:
\begin{align*}
	\text{if}~z-t+\tau-a\leq 0:\\
	N(t;z)&=\frac{1-\tau_{2}}{a_{\tau}}\int_{t}^{z+\omega -a}e^{-r(s-t)}ds\\
	&=\frac{1-\tau_{2}}{ra_{\tau}}\left(1-e^{-r(z-t+\omega-a)}\right),\\
	\text{if}~z-t+\tau -a>0:\\
	N(t;z)&=\frac{1-\tau_{2}}{a_{\tau}}\int_{z+\tau-a}^{z+\omega -a}e^{-r(s-t)-\frac{1}{\sigma}\beta(\mu-r)(z+\tau-a-t)+\alpha(z+\tau-a-t)}ds\\
	&=\frac{1-\tau_{2}}{ra_{\tau}}\left(1-e^{-r(\omega-\tau)}\right)e^{\widetilde{\epsilon}(z-t+\tau-a)}.\\
\end{align*}
And
\begin{align*}
    M(t;z,\theta,k)&=\left[\Lambda\int_{t\vee(z+\tau-a)}^{z+\omega-a}
    e^{\left(r+\frac{1}{\sigma}\xi(\mu-r)-\gamma\right)(t-s)}ds\right.\nonumber\\ &\phantom{ee}\left.-\int_{t\wedge(z+\tau-a)}^{z+\tau-a}
    (1-\tau_1)e^{\left(r+\frac{1}{\sigma}\xi(\mu-r)-\gamma\right)(t-s)}ds\right]\theta\nonumber\\
    &\phantom{ee}+\left[\int_{t\wedge(z+\tau-a)}^{z+\tau-a}[N(s;z)-(1-\tau_{1})]e^{\left(r+\frac{1}{\sigma}\xi(\mu-r)-\gamma\right)(t-s)}ds\right]k\nonumber\\
    &\phantom{ee}+(1-\tau_1)\int_{t\wedge(z+\tau-a)}^{z+\tau-a}e^{\left(r+\frac{1}{\sigma}\xi(\mu-r)-\gamma\right)(t-s)}ds\nonumber\\
    &\triangleq M_1(t;z)\theta+M_2(t;z)k+M_3(t;z).%\label{equ13}
\end{align*}
Thus, $M(t;z,\theta,k)$ is a linear function of $\theta$ and $k$. Moreover, we assume that $\epsilon\neq 0$ and $\widetilde{\epsilon}\neq\epsilon$ are valid. Then, we obtain $M_1(t;z)$, $M_2(t;z)$ and $M_3(t;z)$  expressed in Eqs.\eqref{equ14}-\eqref{equ16}.
%Denote $\epsilon\triangleq\gamma-r-\xi\nu \neq 0$, where $\nu=\frac{\mu-r}{\sigma}$. Moreover, we assume that $\widetilde{\epsilon}=\alpha-r-\beta\nu\neq\epsilon$ is valid. Then,
%\!\!\!\!\!\!\!\!\!\!
%\begin{small}
%\!\!\!\!\!\!\!\!\!\! \!\!\!\!\!\!\!\!\!\!\!\!\!\!\!\!\!\!\!\!\!
%\begin{align*}
%    &M_1(t;z)\!=\!\left\{
%    \begin{array}{lr}
%        \frac{1}{\epsilon}\Lambda\left(e^{\epsilon(z-t+\omega-a)}\!\!-1\right),&z-t+\tau-a\leq 0,\\
%        \frac{1}{\epsilon}\left[(1-\tau_1)\!\!+\left(\Lambda e^{\epsilon(\omega-\tau)}\!\!-\!\!\Lambda-\!\!(1-\!\!\tau_1)\right)e^{\epsilon(z-t+\tau-a)}\right],&z-t+\tau-a>0,
%    \end{array}
%    \right.
%    \\
%    &M_2(t;z)\!=\!\left\{\!
%    \begin{array}{lr}
%        0,\phantom{eeeeeeeeeeeeeeeeeeeeeeeeeeeeeeeeeeeeeeeeeeee
%        }z-t+\tau-a\leq 0,\\
%        \frac{1}{r(\epsilon-\widetilde{\epsilon})}\frac{1-\tau_2}{\omega-\tau}\left(1\!-\! e^{-r(\omega-\tau)}\right)\left(e^{\epsilon(z-t+\tau-a)}\!-\! e^{\widetilde{\epsilon}(z-t+\tau-a)}\right)\\
%        - \frac{1-\tau_1}{\epsilon}\left(e^{\epsilon(z-t+\tau-a)}
%        \!-\!1\right),\\
%        \phantom{eeeeeeeeeeeeeeeeeeeeeeeeeeeeeeeeeeeeeeeeeeeeee}z-t+\tau-a>0,
%    \end{array}
%    \right.\\
%    &M_3(t;z)\!=\!\left\{
%    \begin{array}{lr}
%        0,~~&z-t+\tau-a\leq 0,\\
%        \frac{1}{\epsilon}(1\!-\!\tau_1)\left(e^{\epsilon(z-t+\tau-a)}-1\right),~~&z-t+\tau-a>0.
%    \end{array}
%    \right.
%\end{align*}
%\end{small}
\end{proof}

\begin{Remark}
    $(\pi^*(t,x,w,y;z,\theta,k),C^*(t,x,w,y;z,\theta,k))$ is the feedback function  of the optimal asset allocation and consumption. Based on SDE\eqref{equ3} and using verification theorem, we obtain that $(\pi^*,C^*)=\{(\pi^*(t),C^*(t)  )\triangleq(\pi^*(t;z,\theta,k), C^*(t;z,\theta,k)) \triangleq(\pi^*(t,X^*(t),W^*(t),Y^*(t);z,\theta,k),C^*(t,X^*(t),W^*(t),Y^*(t);z,\theta,k))\}$ is the optimal asset allocation and consumption policy.
\end{Remark}
\begin{proposition}
	For the fixed $\theta$ and $k$, the solution $(\pi^*,C^*)$ given in Theorem \ref{them1} is admissible, i.e., the condition $X^*(z+\omega-a)=0$ is satisfied in SDE (\ref{equ3}).
\end{proposition}
\begin{proof}
	Similar to He, Liang, Song and Ye (2021), by using the martingale method, the optimal terminal wealth is given by
	\begin{equation}\nonumber
		\begin{split}
			X^*(z\!+\!\omega\!-\!a;z)=&-M(z\!+\!\omega\!-\!a;z)W(z\!+\!\omega\!-\!a)-N(z\!+\!\omega\!-\!a;z)Y(z\!+\!\omega\!-\!a)\\
			&+\left[\frac{L(z\!+\!\omega\!-\!a;z)}{L(z;z)}\right]^{\frac{1}{1-\delta}}\!\!M(z;z)e^{\gamma z+\frac{r(\omega-a)}{1-\delta}}\eta(z\!+\!\omega\!-\!a;z)^{\frac{1}{\delta-1}}W_0,
		\end{split}
	\end{equation}
	where $\eta(t;z) \triangleq \mathrm{exp}\left\{-\nu[B(t)-B(z)]-\frac{1}{2}\nu^2(t-z)\right\}$ for $t\in[z,z+\omega-a]$, and $\nu=\frac{\mu-r}{\sigma}$. Based on Eq. (\ref{equ12}), we have $M(z\!+\!\omega\!-\!a;z)=N(z\!+\!\omega\!-\!a;z)=L(z\!+\!\omega\!-\!a;z)=0$, as such, $X^*(z+\omega-a)=0$.
\end{proof}
\vskip 15pt
\setcounter{equation}{0}
\section{{ {\bf Age-dependent preference orderings}}}\label{ordering}
In this section, we study the preference ordering among PAYGO, EET and individual savings of each cohort. We establish three critical ages to distinguish the preference and obtain the multiple outcomes of the age-dependent preference orderings based on the heterogeneous characteristic parameters. By introducing EET pension and the maximum constraint on the sum of the PAYGO and EET contribution rates, the preference orderings are determined by the comparative efficiencies of the three pension schemes. It is the extension of the ``Samuelson-Aaron'' criterion. The main results are the three critical ages established in Theorems \ref{them2}-\ref{them4} and the multiple outcomes of the preference orderings summarized in the flowchart of Fig. \ref{flow}.
\vskip 5pt
We study the preference ordering among the three pensions at time $t_0$, when the government decides to reselect the optimal $\theta$ and $k$. Denote $\zeta\triangleq a+t_0-z$ $(\zeta\leq\omega)$, i.e., $\zeta$ is the age of the participant at time $t_0$, who joined the pension at time $z$. In the following, we first prove that the cohorts ( $\zeta\leq a$, $\zeta\geq\tau$) have simple preference orderings, and the cohorts ($a\leq\zeta<\tau$) have complicated age-dependent preference orderings with the heterogeneous characteristic parameters.
\vskip 10pt
If $\zeta\leq a$, i.e., $z\geq t_0$. The value function of the participants is
\begin{align*}
V(z,0,W(z),0;z,\theta,k)&=\frac{1}{\delta}L(z;z)(M_1(z;z)\theta+M_2(z;z)k+M_3(z;z))^{\delta}W^{\delta}(z).
                      %&=\frac{1}{\delta}L_0(M_{01}\theta+M_{02}k+M_{03})^{\delta}W^{\delta}(z),
\end{align*}
Moreover,
\begin{align*}
	L(z;z)%&=\left\{\int_{z}^{z+\omega -a}b(s;z)^{\frac{1}{1-\delta}}e^{\frac{\delta}{1-\delta}\left[r+\frac{(\mu-r)^2}{2\sigma^2(1-\delta)}\right](s-z)}ds\right\}^{1-\delta}\\
	&=\left\{\int_{0}^{\tau -a}e^{-\frac{r}{1-\delta}u}s(u+a)^{\frac{1}{1-\delta}}e^{\frac{\delta}{1-\delta}\left[r+\frac{(\mu-r)^2}{2\sigma^2(1-\delta)}\right]u}du\right.\\
	&\left.+\lambda^{\frac{1}{1-\delta}}
\int_{\tau-a}^{\omega-a}e^{-\frac{r}{1-\delta}u}s(u+a)^{\frac{1}{1-\delta}}
e^{\frac{\delta}{1-\delta}\left[r+\frac{(\mu-r)^2}{2\sigma^2(1-\delta)}\right]u}du\right\}^{1-\delta}\triangleq L_0,\\
	M_1(z;z)&=\frac{1}{\epsilon}\left[\Lambda e^{\epsilon(\omega-a)}-(\Lambda+1-\tau_1)e^{\epsilon(\tau-a)}+1-\tau_1\right]\triangleq M_{01},\\
	M_2(z;z)&=\!\frac{1}{r(\epsilon\!-\!\widetilde{\epsilon})}\frac{1\!-\!\tau_2}{a_{\tau}}\!\left(\!1\!-\! e^{-r(\omega\!-\!\tau)}\!\right)\!\left(\! e^{\epsilon(\tau\!-\! a)}\!-\! e^{\widetilde{\epsilon}(\tau\!-\! a)}\!\right)\!-\!\frac{1}{\epsilon}(1\!-\!\tau_1)\!\left(\! e^{\epsilon(\tau\!-\! a)}\!-\!1\!\right)\!\triangleq M_{02},\\
	M_3(z;z)&=\frac{1}{\epsilon}(1-\tau_1)\left(e^{\epsilon(\tau-a)}-1\right)\triangleq M_{03},
\end{align*}
where $L_0$, $M_{01}$, $M_{02}$ and $M_{03}$ are independent of $z$, and $W(z)>0$ is independent of $\theta$ and $k$. As such, if $\zeta\leq a$, the participants who will join the pension in the future and who are joining at the moment have the same preference ordering among PAYGO, EET and individual savings.

%{\color{red}If $\zeta\geq a$, the value function of the participants is shown in Eq.\eqref{value}. From SDE \eqref{equ1}, $w>0$ and from Eq.\eqref{L},  $L(t;z)\geq0$ and $L(t;z)=0$ if and only if $t=z+\omega-a$. Moreover, $x+(M_1(t;z)\theta+M_2(t;z)k+M_3(t;z))w+N(t;z)y$ is supposed to be nonnegative for admissible $\theta$ and $k$. $x+(M_1(t;z)\theta+M_2(t;z)k+M_3(t;z))w+N(t;z)y=0$ is rarely observed in practice, because the government is reluctant to choose $\theta$ and $k$ to make some participants' disposable income of 0. Therefore, we omit this case in the following discussion.}
\vskip 5pt
Denote $\widetilde{M_1}(\zeta)\triangleq M_1(t_0;a+t_0-\zeta)$ and $\widetilde{M_2}(\zeta)\triangleq M_2(t_0;a+t_0-\zeta)$, based on Eq.\eqref{equ14} and Eq.\eqref{equ15}, we have
\begin{small}
	\begin{align}
		&\widetilde{M_1}(\zeta)\!=\!\left\{
		\begin{array}{lr}
			\frac{1}{\epsilon}\Lambda\left(e^{\epsilon(\omega-\zeta)}-1\right),~&\zeta\geq\tau,\\
			\frac{1}{\epsilon}\left[(1-\tau_1)+\left(\Lambda e^{\epsilon(\omega-\tau)}-\Lambda-(1-\tau_1)\right)e^{\epsilon(\tau-\zeta)}\right],~&\zeta<\tau,
		\end{array}
		\right.\label{equ14.1}
		\\
		&\widetilde{M_2}(\zeta)\!=\!\left\{\!
		\begin{array}{lr}
			0,~~&\zeta\geq\tau,\\
			\frac{1-\tau_2}{\epsilon-\widetilde{\epsilon}}\frac{1- e^{-r(\omega-\tau)}}{ra_{\tau}}\left(e^{\epsilon(\tau-\zeta)}\!-\! e^{\widetilde{\epsilon}(\tau-\zeta)}\right)-\frac{1-\tau_1}{\epsilon}\left(e^{\epsilon(\tau-\zeta)}\!-\!1\right),~&\zeta<\tau,
		\end{array}
		\right.\label{equ15.1}\\
		&\widetilde{M_1}(\zeta)\!-\!\widetilde{M_2}(\zeta)\!=\!\left\{
		\begin{array}{lr}
			\frac{1}{\epsilon}\Lambda\left(e^{\epsilon(\omega-\zeta)}-1\right),~&\zeta\geq\tau,\\%\phantom{eeeeeeeeeeeeeeeeeeeeeeeeeeeeeeeeeeeee}
			\!\!\left[\frac{1}{\epsilon}\Lambda\left(e^{\epsilon(\omega-\tau)}\!-1\right)\!-\!\frac{1-\tau_2}{\epsilon-\widetilde{\epsilon}}\frac{1 - e^{-r(\omega-\tau)}}{ra_{\tau}}\!\left(1\!- \!e^{(\widetilde{\epsilon}-\epsilon)(\tau-\zeta)}\right)\right]\! e^{\epsilon(\tau-\zeta)},~&\zeta<\tau.%\\
			%            \phantom{eeeeeeeeeeeeeeeeeeeeeeeeeeeeeeeeeeeeeeeeeeeeeeeeeeeeee}\zeta<\tau.
		\end{array}
		\right.\label{equ16.1}
	\end{align}
\end{small}

Based on Eq.\eqref{value}, when $\zeta\geq a$, the preference orderings are determined by the signs of $\widetilde{M_1}(\zeta)$, $\widetilde{M_2}(\zeta)$ and $\widetilde{M_1}(\zeta)-\widetilde{M_2}(\zeta)$. The signs of the three coefficients measure the comparative efficiencies between every two pensions.
\vskip 10pt
If $\zeta\geq\tau$, $\widetilde{M_1}(\zeta)>0$, $\widetilde{M_2}(\zeta)=0$. In this case, higher $\theta$ increases the utility of the participants.
\vskip 10pt
If $a\leq\zeta<\tau$, unlike the results in the existing literature, the comparative efficiency depends on the age of the participants.
\vskip 10pt
\textbf{Case 1: preference between PAYGO pension and individual savings}
\vskip 5pt
Denote $$\Lambda_{FP}\triangleq\frac{(1-\tau_1)\left(1-e^{-\epsilon(\tau-a)}\right)}{e^{\epsilon(\omega-\tau)}-1}.$$
\begin{Them}\label{them2}
If $\Lambda\leq\Lambda_{FP}$, then there exists a critical age $\hat{\zeta}$. When $\zeta<\hat{\zeta}$, $\widetilde{M_1}(\zeta)<0$, the participants of the age $\zeta$ prefer individual savings to PAYGO pension. When $\zeta>\hat{\zeta}$, $\widetilde{M_1}(\zeta)>0$, the participants of the age $\zeta$ prefer PAYGO pension to individual savings. When $\zeta=\hat{\zeta}$, $\widetilde{M_1}(\zeta)=0$, the participants of the age $\zeta$ treat PAYGO pension and individual savings equally.\\

If $\Lambda>\Lambda_{FP}$, $\widetilde{M_1}(\zeta)>0$, the participants of all ages %$\zeta$
prefer PAYGO pension to individual savings, and their utilities increase as $\theta$ increases.
\end{Them}
\begin{proof}
From Eq.\eqref{equ14.1}, $\widetilde{M_1}(\zeta)$ is continuous and monotonous in $[a,\tau]$ and $\widetilde{M_1}(\tau)=\frac{1}{\epsilon}\Lambda\left(e^{\epsilon(\omega-\tau)}-1\right)>0$. Thus, we know that $\widetilde{M_1}(\zeta)$ has at most one zero point in $[a,\tau)$ and the zero point exists if and only if $\widetilde{M_1}(a)\leq 0$.

As such, if $\widetilde{M_1}(a)\leq 0$, i.e., $\Lambda\leq\Lambda_{FP}$, there exists a critical age $\hat{\zeta}$ that satisfies $\widetilde{M_1}(\zeta)=0$ and $\widetilde{M_1}(\zeta)$ is monotonically increasing in $[a,\tau)$. Furthermore,
$$\hat{\zeta}=\tau+\frac{1}{\epsilon}\ln\left[1+\frac{\Lambda-\Lambda e^{\epsilon(\omega-\tau)}}{1-\tau_1}\right].$$

If $\widetilde{M_1}(a)>0$, i.e., $\Lambda>\Lambda_{FP}$, $\widetilde{M_{1}}(\zeta)$ is always greater than 0.

\end{proof}
\vskip 5pt
In fact, the rise of the PAYGO contribution rate has two impacts on the participants' utilities. It increases the contribution during the working period and the benefit during the retirement period simultaneously. In the case of $\Lambda\leq\Lambda_{FP}$, the former impact dominates the latter impact for the younger cohorts and the latter dominates the former for the older cohorts. However, in the case of  $\Lambda>\Lambda_{FP}$, the comparative efficiency of PAYGO pension is superior, and thus all the participants prefer PAYGO pension.\\
\vskip 10pt
\textbf{Case 2: preference between PAYGO and EET pensions}
\vskip 5pt
Denote
 $$\Lambda_{EP}\triangleq\frac{\epsilon(1-\tau_2)\left(1-e^{-r(\omega-\tau)}\right)\left(e^{(\widetilde{\epsilon}-\epsilon)(\tau-a)}-1\right)}{ra_{\tau}\left(e^{\epsilon(\omega-\tau)}-1\right)(\widetilde{\epsilon}-\epsilon)}.$$
\begin{Them}\label{them3}
If $\Lambda\leq\Lambda_{EP}$, then there exists a critical age $\widetilde{\zeta}$. When $\zeta<\widetilde{\zeta}$, $\widetilde{M_1}(\zeta)-\widetilde{M_2}(\zeta)<0$, the participants of the age $\zeta$ prefer EET pension to PAYGO pension. When $\zeta>\widetilde{\zeta}$, $\widetilde{M_1}(\zeta)-\widetilde{M_2}(\zeta)>0$, the participants of the age $\zeta$ prefer PAYGO pension to EET pension. When $\zeta=\widetilde{\zeta}$, $\widetilde{M_1}(\zeta)-\widetilde{M_2}(\zeta)=0$, the participants of the age $\zeta$ treat the two pensions equally.\\

If $\Lambda>\Lambda_{EP}$, $\widetilde{M_1}(\zeta)-\widetilde{M_2}(\zeta)>0$, the participants of all ages %$\zeta$
prefer PAYGO pension to EET pension.
\end{Them}
\begin{proof}
From Eq.\eqref{equ16.1}, the sign of $\widetilde{M_1}(\zeta)-\widetilde{M_2}(\zeta)$ is consistent with the sign of $\widetilde{M_3}(\zeta)\triangleq (\widetilde{M_1}(\zeta)-\widetilde{M_2}(\zeta))e^{-\epsilon(\tau-\zeta)}$ and
$$\widetilde{M_3}(\zeta)=\frac{1}{\epsilon}\Lambda\!\left(\!e^{\epsilon(\omega-\tau)}\!-\!1\!\right)\!-\!\frac{1}{r(\epsilon\!-\!\widetilde{\epsilon})}\frac{1-\tau_2}{a_{\tau}}\!\left(\!1\! -\! e^{-r(\omega-\tau)}\!\right)\!\left(\!1\!-\! e^{(\widetilde{\epsilon}-\epsilon)(\tau-\zeta)}\!\right).$$

Similar to the proof in Case 1, $\widetilde{M_3}(\zeta)$ is continuous and monotonous in $[a,\tau]$ and $\widetilde{M_3}(\tau)>0$. As such, if $\widetilde{M_3}(a)\leq 0$, i.e., $\Lambda\leq\Lambda_{EP}$, there exists a critical age $\widetilde{\zeta}$ that satisfies $\widetilde{M_3}(\zeta)=0$ and $\widetilde{M_3}(\zeta)$ is monotonically increasing in $[a,\tau)$. Furthermore,
$$\widetilde{\zeta}=\tau+\frac{1}{\epsilon-\widetilde{\epsilon}}\ln\left[1-\frac{\Lambda\left(e^{\epsilon(\omega-\tau)}-1\right)r(\epsilon-\widetilde{\epsilon})a_{\tau}}{\epsilon(1-\tau_2)\left(1-e^{-r(\omega-\tau)}\right)}\right].$$
If $\widetilde{M_3}(a)>0$, i.e., $\Lambda>\Lambda_{EP}$, $\widetilde{M_{3}}(\zeta)$ is always greater than 0.
\end{proof}
\vskip 5pt
Similarly,  in the case of $\Lambda\leq\Lambda_{EP}$, the contribution increasing impact dominates the benefit increasing impact for the younger cohorts and the opposite is valid for the older cohorts. However, in the case of $\Lambda>\Lambda_{EP}$, the comparative efficiency of PAYGO pension is superior, and thus all the participants prefer PAYGO pension. Interestingly, the distinguishing criterions of $\Lambda$ are heterogeneous in the two cases because that the  criterions measure the comparative efficiencies between PAYGO pension and individual savings, and between PAYGO and EET pensions, respectively.

\vskip 10pt
\textbf{Case 3: preference between EET pension and individual savings}
\vskip 5pt
Echoing the fact that EET pension is operated by professional institutions, we assume that the Sharpe ratio of EET pension is higher than the one of individual savings, namely $\widetilde{\epsilon}>0$.

From Eq.\eqref{equ15.1}, we have
\begin{align*}
    &\widetilde{M_2}(a)=\frac{1}{r(\epsilon-\widetilde{\epsilon})}\frac{1-\tau_2}{a_{\tau}}\!\left(\!1-e^{-r(\omega-\tau)}\!\right)\!\left(\! e^{\epsilon(\tau-a)}-e^{\widetilde{\epsilon}(\tau-a)}\!\right)\!-\frac{1}{\epsilon}(1-\tau_1)\!\left(\! e^{\epsilon(\tau-a)}-1\!\right)\!,\\
    &\widetilde{M_2}'(\tau)=1-\tau_1-\frac{1-\tau_2}{ra_{\tau}}\!\left(\!1-e^{-r(\omega-\tau)}\!\right)\!.
    %&\widetilde{M_2}'(\tau)=\frac{-\epsilon}{r(\epsilon-\widetilde{\epsilon})}\frac{1-\tau_2}{{\color{red%}a_{\tau}}}\!\left(\!1-e^{-r(\omega-\tau)}\!\right)\!+(1-\tau_1)+\frac{\widetilde{\epsilon}}{r(\epsil%on-\widetilde{\epsilon})}\frac{1-\tau_2}{{\color{red}a_{\tau}}}\!\left(\!1-e^{-r(\omega-\tau)}\!\righ%t)\!.
\end{align*}

\begin{Them}\label{them4}
If
$\widetilde{M_2}(a)\geq0$, and $\widetilde{M_2}'(\tau)>0$, then there exists a critical age $\bar{\zeta}$. When $\zeta<\bar{\zeta}$, $\widetilde{M_2}(\zeta)>0$, the participants of the age $\zeta$ prefer EET pension to individual savings. When $\zeta>\bar{\zeta}$, $\widetilde{M_2}(\zeta)<0$, the participants of the age $\zeta$ prefer individual savings to EET pension. When $\zeta=\bar{\zeta}$, $\widetilde{M_2}(\zeta)=0$, the participants of the age $\zeta$ treat EET pension and individual savings equally.\\

If $\widetilde{M_2}(a)<0$, the participants of all ages %$\zeta$
prefer individual savings to EET pension.\\

If $\widetilde{M_2}'(\tau)\leq0$, the participants of all ages %$\zeta$
prefer EET pension to individual savings.
\end{Them}

\begin{proof}
Differentiating Eq.\eqref{equ15.1}, we have
\begin{small}
\begin{align*}
    \widetilde{M_2}'(\zeta)\!=\!\left[\!\frac{-\epsilon}{r(\epsilon\!-\!\widetilde{\epsilon})}\frac{1\!-\!\tau_2}{a_{\tau}}\!\left(\!1\!-\! e^{-r(\omega-\tau)}\!\right)\!+\!(1\!-\!\tau_1)\!\right]\! e^{\epsilon(\tau\!-\!\zeta)}\!+\!\frac{\widetilde{\epsilon}}{r(\epsilon\!-\!\widetilde{\epsilon})}\frac{1\!-\!\tau_2}{a_{\tau}}\!\left(\!1\!-\! e^{-r(\omega-\tau)}\!\right)\! e^{\widetilde{\epsilon}(\tau\!-\!\zeta)}.
\end{align*}
\end{small}
The sign of $\widetilde{M_2}'(\zeta)$ is consistent with the sign of $\widetilde{N}(\zeta)\triangleq \widetilde{M_2}'(\zeta)e^{-\epsilon(\tau-\zeta)}$ and
\begin{small}
    \begin{align*}
        \widetilde{N}(\zeta)\!=\!\frac{-\epsilon}{r(\epsilon\!-\!\widetilde{\epsilon})}\frac{1\!-\!\tau_2}{a_{\tau}}\!\left(\!1\!-\! e^{-r(\omega-\tau)}\!\right)\!+\!(1\!-\!\tau_1)\!+\!\frac{\widetilde{\epsilon}}{r(\epsilon\!-\!\widetilde{\epsilon})}\frac{1\!-\!\tau_2}{a_{\tau}}\!\left(\!1\!-\! e^{-r(\omega-\tau)}\!\right)\! e^{(\widetilde{\epsilon}-\epsilon)(\tau\!-\!\zeta)}.\\
    \end{align*}
\end{small}
Furthermore,
\begin{align}
    \widetilde{N}'(\zeta)=\frac{\widetilde{\epsilon}}{r}\frac{1-\tau_2}{a_{\tau}}\left(1-e^{-r(\omega-\tau)}\right)e^{(\widetilde{\epsilon}-\epsilon)(\tau-\zeta)}\label{equ21}.
\end{align}
\vskip 5pt
 Using Eq.\eqref{equ21} and $\widetilde{\epsilon}>0$, we know that $\widetilde{N}(\zeta)$ strictly increases. As such, $\widetilde{N}(\zeta)$ is always smaller than 0, or greater than 0, or first smaller and then greater than 0. So is $\widetilde{M_2}'(\zeta)$. Thus, the maximum point of $\widetilde{M_2}(\zeta)$ in $[a,\tau]$ must be taken at the boundary. Combined with $\widetilde{M_2}(\tau)=0$, $\widetilde{M_2}(\zeta)$ has at most one zero point in $[a,\tau)$.
\vskip 5pt
First, if $\widetilde{M_2}(a)<0$, based on the fact that the maximum point of $\widetilde{M_2}(\zeta)$ must be taken at the boundary and  $\widetilde{M_2}(\tau)=0$, thus $\widetilde{M_2}(\zeta)<0$ for all $\zeta\in[a,\tau)$. Second, if $\widetilde{M_2}'(\tau)\leq0$, then $\widetilde{M_2}'(\zeta)$ is always less than 0 and $\widetilde{M_2}(\zeta)$ strictly decreases and $\widetilde{M_2}(\zeta)>0$ for all $\zeta\in[a,\tau)$. Last, if $\widetilde{M_2}(a)\geq0 $ and $\widetilde{M_2}'(\tau)>0$, there exists only one $\bar{\zeta}\in[a,\tau)$ that satisfies $\widetilde{M_2}(\bar{\zeta})=0$. When $\zeta<\bar{\zeta}$, $\widetilde{M_2}(\zeta)>0$ and when $\zeta>\bar{\zeta}$, $\widetilde{M_2}(\zeta)<0$.
\end{proof}
\vskip 5pt
In fact, EET pension has the following pros and cons, compared with individual savings. The return efficiency of EET pension during the accumulation period is relatively higher due to professional operation. Besides, the participants can enjoy preferential taxation by joining EET pension. However, EET pension lacks flexibility and has lower return efficiency after retirement. In the case of $\widetilde{M_2}(a)\geq0$, and $\widetilde{M_2}'(\tau)>0$, for the younger cohorts, the accumulation period is longer and they can benefit from the higher return efficiency and the tax saving effect. Thus, the younger cohorts prefer EET pension and the older cohorts prefer individual savings. However, in the extreme case $\widetilde{M_2}(a)<0$ ($\widetilde{M_2}'(\tau)\leq0$), the overall comparative efficiency of individual savings (EET pension) is superior, and thus the participants of all ages %$\zeta$
prefer individual savings (EET pension).
\vskip 15pt
Combining the above three cases, we establish the preference orderings among PAYGO, EET and individual savings of different cohorts based on heterogeneous characteristic parameters. In the flowchart of Fig. \ref{flow}, $I\succ E$ means the participants  prefer individual savings to EET pension and $I\sim E$ means the participants  treat individual savings and EET pension equally. The text on the arrow is the condition to distinguish the preference  and the text in the box is the result of the preference ordering.  The last column shows the multiple outcomes of the preference orderings. Specifically, the texts in bold are the conditions to realize the final outcomes.
\vskip 5pt
%We establish the preference orderings based on the following classification procedures. According to Case 3, we distinguish the preference between EET pension and individual saving. Because $\widetilde{M_2}(a)=\frac{1}{\epsilon}(e^{\epsilon(\omega-\tau)}-1)e^{\epsilon(\tau-a)}(\Lambda_{EP}-\Lambda_{FP})$, the sign of $\widetilde{M_2}(a)$ also determines the relative size of $\Lambda_{EP}$ and $\Lambda_{FP}$. As such, according to the first classification and the results in Case 1 and Case 2, we have three more categories of the outcomes. The first eight classifications in the third column establish the outcomes of the preference orderings. However, the last one still needs further analytical work. We use the closed-form $\hat{\zeta}$ and $\widetilde{\zeta}$ to make the classification. Actually, when the relative size of $\hat{\zeta}$ and $\widetilde{\zeta}$ is determined, the relative size of $\hat{\zeta}$, $\widetilde{\zeta}$ and $\bar{\zeta}$ is also determined. Take $\hat{\zeta}<\widetilde{\zeta}$ for example. If $\bar{\zeta}\leq\widetilde{\zeta}$, considering the age $\zeta=\widetilde{\zeta}$,  we have $E\sim P$ from Case 2, and we have $P\succ I$ from $\zeta>\hat{\zeta}$ and Case 1. Meanwhile,  we have $I\succeq E$ from $\zeta\geq\bar{\zeta}$ and Case 3, and this leads to a contradiction. Eventually, based on the relative size of $\hat{\zeta}$ and $\widetilde{\zeta}$, we obtain the last three outcomes of the preference orderings.
\begin{figure}
    \includegraphics[scale=0.5]{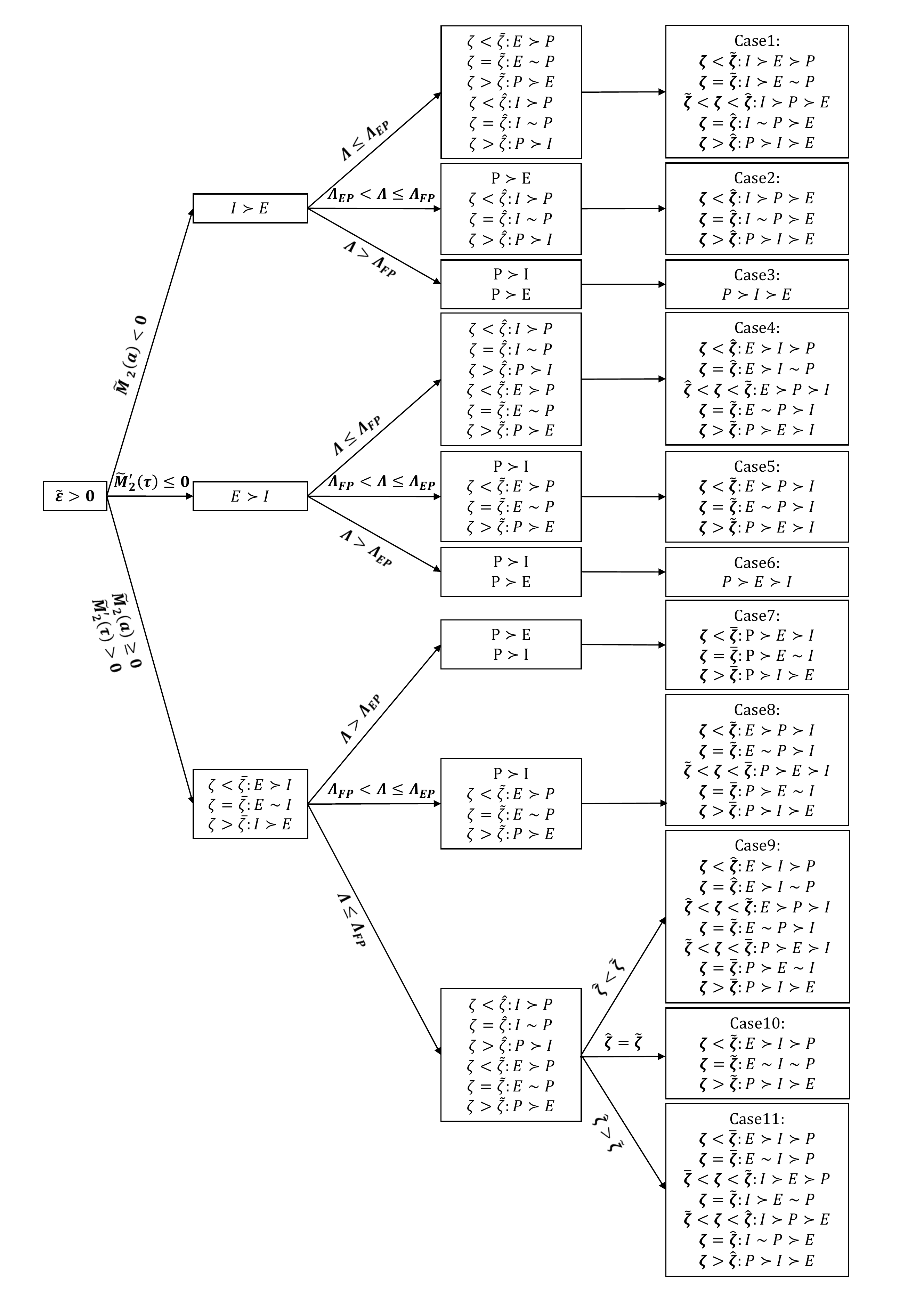}
    \caption{Age-dependent preference orderings}\label{flow}
\end{figure}
\vskip 15pt
\begin{remark}
	We choose the perfectly correlated Brownian process for the following  reasons. Under the perfect correlation assumption, the analytical solution Eqs.\eqref{pi}-\eqref{N} can be obtained. Then as shown in Fig. \ref{flow}, each cohort has an exclusive preference ordering among PAYGO, EET and individual savings according to their comparative efficiencies. And the optimal mix is the comprehensive result of balancing the heterogeneous preference orderings of different cohorts. We extend the ``Samuelson-Aaron" criterion and give a new explanation for the mix of pensions in society from the perspective of social welfare.
	
	However, under the partial correlation assumption, we will get a fully nonlinear parabolic HJB equation which can only be solved by numerical method, similar to \cite{BCG2007}. And due to the effect of risk diversification, the optimal choice of each cohort is a mix of the three pensions, and naturally the optimal pension in society is a mix. The role of social welfare in mixing pensions is overshadowed by risk diversification.
\end{remark}
%\newpage
\setcounter{equation}{0}
\section{{ {\bf The optimization problem of the government}}}\label{government}
%\subsection{{\color{red}The objective function of the government}
	Echoing the two-stage game between the participants and the government, the government could anticipate the participants' optimal feedback functions with respect to any contribution rates. The government gives no credence to threats by the participants to act in ways that will not be in their self-interest when the second stage arrives. Thus, back to the first stage, the government determines the optimal PAYGO and EET contribution rates based on the participants' optimal feedback functions. In this section, we establish the optimization problem of the government.
	%Echoing the two-stage game between the participants and the government, back to the first stage, the government is fully aware of the participants' optimal feedback functions with respect to any given contribution rates and then determines the optimal PAYGO and EET contribution rates based on the optimal feedback functions.
	\vskip 5pt
	Facing the new demographic changes, the government decides to reselect the optimal $\theta$ and $k$ to maximize the overall utility of the cohorts at time $t_0$.  Based on the feedback function $(\pi^*(t,x,w,y;z,\theta,k),C^*(t,x,w,y;z,\theta,k))$ of the participants, the optimization problem of the government is as follows:
\begin{align}
	\Phi_{1}(\theta,k)=&\max_{\theta,k}\{\phi_{1}(\theta,k)\}=\max_{\theta,k}\left\{\int_{t_0-\omega+a}^{t_0}n(z)V(t_0,x_0(z),w_0,y_0(z);z,\theta,k)dz\right. \nonumber\\
	&\left.+\int_{t_0}^{\infty}e^{-r(z-t_0)}n(z)\mathrm{E}[V(z,0,W(z),0;z,\theta,k)|W(t_0)=w_0]dz\right\},\label{equ17}
\end{align}
or
\begin{align}
	\Phi_2(\theta,k)=&\max_{\theta,k}\{\phi_2(\theta,k)\}=\max_{\theta,k}\left\{\int_{t_0-\omega+a}^{t_0}V(t_0,x_0(z),w_0,y_0(z);z,\theta,k)dz\right. \nonumber\\
	&\left.+\int_{t_0}^{\infty}e^{-r(z-t_0)}\mathrm{E}[V(z,0,W(z),0;z,\theta,k)|W(t_0)=w_0]dz\right\}.\label{equ171}
\end{align}
Inspired by \cite{HS1989} and \cite{MV1996}, $\Phi_1(\theta,k)$ takes the population of each cohort as the weight parameter. We also consider the case of equal weight, and the weight parameter of $\Phi_2(\theta,k)$ is 1. The first item of Eqs.\eqref{equ17}-\eqref{equ171} includes the utilities of the participants who have joined the pension and who are joining the pension at the moment. $w_0$, $y_0(z)$ and $x_0(z)$ are the realized average salary, EET pension accumulation and wealth of the participant at time $t_0$. $ W(t_0)=w_0$ is the information fully observed by the government at time $t_0$, while the information of $y_0(z)$ and $x_0(z)$ is privately kept by the participants. This will be estimated by the government in the expectation way. Because $W$ is a Markov process, we know that $\mathrm{E}[V(z,0,W(z),0;z,\theta,k)|W(t_0)]=\mathrm{E}[V(z,0,W(z),0;z,\theta,k)|\mathcal{F}_{t_0}]$, that is,  all of the information about $V(z,0,W(z),0;z,\theta,k)$ of the government at time $t_0$ is included in the $\sigma$-field generated by $W(t_0)$. As such, the second item of Eqs.\eqref{equ17}-\eqref{equ171} includes all of the information of the utilities of the participants who will join the pension in the future. For the risk aversion attitude, we suppose that the cohorts of different ages have different risk aversion coefficients, that is, $\delta=\delta(z)$. The risk aversion coefficient of the cohorts who have not been employed yet is $\delta_0$. And older cohorts are more risk averse. Through further calculations, we summarize the government's optimization objective in the following Proposition \ref{objective1} whose proof is given in Appendix $\ref{pp4.1}$.
	\begin{proposition}\label{objective1}
		The optimization objective of the government is a deterministic function given by
		\begin{small}
		\begin{align}
			&\phi_1(\theta,k)\!=\!\int_{t_0-\omega+a}^{t_0}\!\frac{1}{\delta(z)}n(z)L(t_0;z)\left\{x_0(z)\!+\!\left[M_1(t_0;z)\theta\!+\! M_2(t_0;z)k \!+\! M_3(t_0;z)\right ]w_0\right.\nonumber\\
			&\!+\!\left. N(t_0;z)y_0(z)\right\}^{\delta(z)}dz \!+\!\frac{n_0L_0e^{\rho t_0}w_0^{\delta_0}}{\delta_0\!\left\{\! r \!-\!\rho\!-\!\delta_0\!\left[\! \gamma\!+\!\frac{1}{2}(\delta_0\!-\!1)\xi^2\!\right]\!\right\}\!}(M_{01}\theta\!+\! M_{02}k \!+\! M_{03})^{\delta_0},\label{utility1}
		\end{align}
	or
	 \begin{align}
		&\phi_2(\theta,k)\!=\!\int_{t_0-\omega+a}^{t_0}\!\frac{1}{\delta(z)}L(t_0;z)\left\{x_0(z)\!+\!\left[M_1(t_0;z)\theta\!+\! M_2(t_0;z)k \!+\! M_3(t_0;z)\right ]w_0\right.\nonumber\\
		&\!+\!\left. N(t_0;z)y_0(z)\right\}^{\delta(z)}dz \!+\!\frac{L_0e^{\rho t_0}w_0^{\delta_0}}{\delta_0\!\left\{\! r \!-\!\rho\!-\!\delta_0\!\left[\! \gamma\!+\!\frac{1}{2}(\delta_0\!-\!1)\xi^2\!\right]\!\right\}\!}(M_{01}\theta\!+\! M_{02}k \!+\! M_{03})^{\delta_0}.\label{utility11}
	\end{align}
	\end{small}
		%where $L_0\triangleq L(z;z)$, $M_{01}\triangleq M_1(z;z)$, $M_{02}\triangleq M_2(z;z)$ and $M_{03}\triangleq M_3(z;z)$ are constants independent of $z$.
		
		Moreover, the admissible scope of $\theta$ and $k$ is as follows:
		\begin{align}
			\left\{
			\begin{array}{lr}
				%M_1(t_0;z)\theta+M_2(t_0;z)k \geq M_1(t_0;z)\theta_0+M_2(t_0;z)k_0\\
				%-\left[\frac{L(t_0;z)}{L(z;z)}\right]^{\frac{1}{1-\delta}}M(z;z,k_0,\theta_0)e^{\left[-\gamma+\frac{r}{1-\delta}+\frac{2-\delta}{2(1-\delta)^2}\nu^2\right](t_0- z)},\forall ~z\in[t_0-\omega+a,t_0],\\
				x_0(z)+[M_1(t_0;z)\theta\!+\! M_2(t_0;z)k \!+\! M_3(t_0;z)]w_0+N(t_0;z)y_0(z)\geq 0,\\
				\phantom{eeeeeeeeeeeeeeeeeeeeeeeeeeeeeeeeeeeeeeeeeee}\forall ~z\in[t_0-\omega+a,t_0],\\
				\theta+k\leq m, \forall\  \theta,\ k\in[0,1].\\
			\end{array}
			\right.\label{equ19}
		\end{align}
\end{proposition}
%Where
%\begin{align*}
%   &M_1(t;z)=\int_{t\vee(z+\tau-a)}^{z+\omega-a}\frac{\Theta(s)}{\Sigma(s)}e^{(r+\frac{1}{\sigma}\xi(\mu-r)-\gamma)(t-s)}ds-\int_{t\wedge(z+\tau-a)}^{z+\tau-a}(1-\tau_1)e^{(r+\frac{1}{\sigma}\xi(\mu-r)-\gamma)(t-s)}ds,\\
%   &M_2(t;z)=\int_{t\wedge(z+\tau-a)}^{z+\tau-a}[N(s;z)-(1-\tau_{1})]e^{(r+\frac{1}{\sigma}\xi(\mu-r)-\gamma)(t-s)}ds,\\
%   &M_3(t;z)=(1-\tau_1)\int_{t\wedge(z+\tau-a)}^{z+\tau-a}e^{(r+\frac{1}{\sigma}\xi(\mu-r)-\gamma)(t-s)}ds.
%\end{align*}
\begin{Remark}
	Following the ideas of \citet{HLSY2021}, $y_0(z)$ and $x_0(z)$ are estimated in the expectation way. $y_0(z)$ is directly obtained by solving the SDEs Eq.\eqref{equ1} and Eq.\eqref{equ2}. $x_0(z)$ is obtained by using the martingale method. The details are given in Appendix \ref{A.1}. Particularly, for the optimization of the government utility, the method is similar to the one used in  \citet{HLSY2021}. Thus, we omit most of the deduction process and focus on the discussions of the preference orderings among the pensions in the previous section.
\end{Remark}
\vskip 15pt
\setcounter{equation}{0}
\section{{ {\bf The new Nash equilibrium with voluntary EET contribution rate}}}\label{vgovernment}
Unlike the obligatory scheme in the U.S. and the U.K., EET pension is voluntarily participated in some countries, such as German and France. In order to depict this general situation, we study a new Nash equilibrium between the participants and the government in this section. Although EET pension is voluntarily participated, the contribution rate is not allowed to be adjusted timely according to the market situations. Usually, the constant contribution rate $k$ could be determined by the participants when they join EET pension. Besides, it can also be adjusted occasionally when the government decides to reselect the contribution rates.
\vskip 5pt
Following the ideas, we establish the new Nash equilibrium between the participants and the government. We assume that the participants adjust the constant EET contribution rate $k$, when the government reselects the PAYGO contribution rate $\theta$ at time $t_0$. It is a degenerate problem of the previous Nash equilibrium and we can derive the analytical solutions. The optimal $k$ chosen by the participants based on the adjusted $\theta$ is summarized in the following Proposition \ref{k} and the optimization objective of the government is shown in the following Corollary \ref{objective2} whose proofs are given in Appendix \ref{pp5.1} and Appendix \ref{pc5.1}.
\begin{proposition}\label{k}
The optimal EET contribution rate $k^*(\theta;z)$ chosen by the participants of all ages based on the adjusted PAYGO contribution rate is divided into the following three cases:\\
If $\zeta\geq\tau$, i.e., $z\leq t_0-\tau+a$, then $k^*(\theta;z)=[0,m-\theta]$.\\
If $\zeta\leq a$, i.e., $z\geq t_0$, then
\begin{align}\label{k1}
		k^*(\theta;z)=&\left\{
		\begin{array}{lr}
			m-\theta,~~&M_{02}>0,\\
			\left[0,m-\theta\right],&M_{02}=0,\\
			0,~~&M_{02}<0.
		\end{array}
		\right.
\end{align}
If $a\leq\zeta<\tau$, i.e., $t_0-\tau+a<z\leq t_0$, then
\begin{align}\label{k2}
	k^*(\theta;z)=&\left\{
	\begin{array}{lr}
		m-\theta,~~&M_2(t_0,z)>0,\\
		\left[0,m-\theta\right],&M_2(t_0,z)=0,\\
		0,~~&M_2(t_0,z)<0.
	\end{array}
	\right.
\end{align}
\end{proposition}

\begin{corollary}\label{objective2}
The optimization objective of the government under the new Nash equilibrium is a deterministic function given by
	%\begin{align}
	%	\widetilde{\phi}(\theta)&\!=\!
	%\int_{(t_0-\tau+a,t_0]\cap\{z:M_2(t_0,z)>0\}}\!\frac{1}{\delta}n(z)L(t_0;z)\{x_0(z;k_0,\theta_0)\!+\![(M_1(t_0;z)\!-\! M_2(t_0;z))\theta\nonumber\\\nonumber
	%	&\!+\!mM_2(t_0;z)\!+\! M_3(t_0;z)]w_0\!+\! N(t_0;z)y_0(z;k_0)\}^{\delta}dz\\\nonumber
	%	&\!+\!\int_{[t_0-\omega+a,t_0-\tau+a]\cup\left((t_0-\tau+a,t_0]\cap\{z:M_2(t_0,z)\leq0\}\right)}\!\frac{1}{\delta}n(z)L(t_0;z)\{x_0(z;k_0,\theta_0)\\\nonumber
	%	&\!+\![M_1(t_0;z)\theta\!+\! M_3(t_0;z)]w_0\!+\! N(t_0;z)y_0(z;k_0)\}^{\delta}dz\\\nonumber
	%	&\!+\!\frac{n_0L_0e^{\rho t_0}w_0^{\delta}}{\delta\left\{r-\rho-\delta\left[ \gamma+\frac{1}{2}(\delta-1)\xi^2\right]\right\}}\left\{\left[\left((M_{01}-M_{02})\theta+mM_{02}+M_{03}\right)^{\delta}\right]1_{\{M_{02}>0\}}\right.\\
	%	&\left.\!+\!\left[(M_{01}\theta+M_{03})^{\delta}\right]1_{\{M_{02}\leq0\}}\right\}.\label{utility2}
	%	\end{align}
\begin{small}
\begin{align}
&\widetilde{\phi}_{1}(\theta)\!	=\!\int_{t_0\!-\!\omega+a}^{t_0}\!\frac{1}{\delta(z)}n(z)L(t_0;z)\left\{x_0(z)\!+\!\left[(M_1(t_0;z)\!-\! M_2(t_0;z)^+)\theta\!+\! M_2(t_0;z)^+m\right.\right.\nonumber\\
&\left.\left.\phantom{ee}\!+\! M_3(t_0;z)\right]w_0\!+\! N(t_0;z)y_0(z)\right\}^{\delta(z)}dz+\!\frac{n_0L_0e^{\rho t_0}w_0^{\delta_0}}{\delta_0\!\left\{\! r \!-\!\rho\!-\!\delta_0\!\left[\! \gamma\!+\!\frac{1}{2}(\delta_0\!-\!1)\xi^2\!\right]\!\right\}\!}[(M_{01}-M_{02}^+)\theta\nonumber\\
&\phantom{eee}\!+\! M_{02}^+m\!+\! M_{03}]^{\delta_0},\label{utility2}
\end{align}
or
\begin{align}
	&\widetilde{\phi}_{2}(\theta)\!	=\!\int_{t_0\!-\!\omega+a}^{t_0}\!\frac{1}{\delta(z)}L(t_0;z)\left\{x_0(z)\!+\!\left[(M_1(t_0;z)\!-\! M_2(t_0;z)^+)\theta\!+\! M_2(t_0;z)^+m\right.\right.\nonumber\\
	&\left.\left.\phantom{ee}\!+\! M_3(t_0;z)\right]w_0\!+\! N(t_0;z)y_0(z)\right\}^{\delta(z)}dz+\!\frac{L_0e^{\rho t_0}w_0^{\delta_0}}{\delta_0\!\left\{\! r \!-\!\rho\!-\!\delta_0\!\left[\! \gamma\!+\!\frac{1}{2}(\delta_0\!-\!1)\xi^2\!\right]\!\right\}\!}[(M_{01}-M_{02}^+)\theta\nonumber\\
	&\phantom{eee}\!+\! M_{02}^+m\!+\! M_{03}]^{\delta_0},\label{utility21}
\end{align}
\end{small}
where $M_2(t_0;z)^+\triangleq M_2(t_0;z)1_{\{M_2(t_0;z)>0\}}$ and $M_{02}^+\triangleq M_{02}1_{\{M_{02}>0\}}$.

Moreover, the admissible scope of $\theta$ is as follows:
\begin{align}
	\theta\in[0\vee\underline{\theta},\ \ m\wedge\bar{\theta}],
	\label{equ19.1}
\end{align}
where
\begin{align*}
	&\underline{\theta}\triangleq\sup_{\{z:a+t_0-z\in A_1\}}-\frac{x_0(z)+N(t_0;z)y_0(z)+(M_2(t_0;z)^+m+M_3(t_0;z))w_0}{(M_1(t_0;z)-M_2(t_0;z)^+)w_0},\\
	&\bar{\theta}\triangleq\inf_{\{z:a+t_0-z\in A_2\}}-\frac{x_0(z)+N(t_0;z)y_0(z)+(M_2(t_0;z)^+m+M_3(t_0;z))w_0}{(M_1(t_0;z)-M_2(t_0;z)^+)w_0},
\end{align*}
\begin{align}\label{A1}
	&A_1\triangleq[\tau,\omega]\cup\left\{
	\begin{array}{lr}
		[a,\tau),\phantom{eee}\widetilde{M}_2(a)<0,\ \Lambda>\Lambda_{FP}~or~\Lambda>\Lambda_{EP},\\
		\left(\hat{\zeta},\tau\right),\phantom{ee}\widetilde{M}_2(a)<0,\ \Lambda\leq\Lambda_{FP}
		\\\phantom{eeeeeeeee}or~\widetilde{M}_2(a)\geq0,\  \widetilde{M}_2'(\tau)>0,\ \Lambda\leq\Lambda_{FP},\ \hat{\zeta}>\widetilde{\zeta},\nonumber\\
		\left(\widetilde{\zeta},\tau\right),\phantom{ee}otherwise,
	\end{array}
	\right.\\
	\end{align}
and
\begin{align}\label{A2}
	&A_2\triangleq\left\{
	\begin{array}{lr}
		\emptyset,\phantom{eeeeeee}\widetilde{M}_2(a)<0,\ \Lambda>\Lambda_{FP}~or~\Lambda>\Lambda_{EP},\\\left[a,\hat{\zeta}\right),
\phantom{ee}\widetilde{M}_2(a)<0,\ \Lambda\leq\Lambda_{FP}
		\\\phantom{eeeeeeeee}or~\widetilde{M}_2(a)\geq0, \ \widetilde{M}_2'(\tau)>0,\ \Lambda\leq\Lambda_{FP},\ \hat{\zeta}>\widetilde{\zeta},\\
		\left[a,\widetilde{\zeta}\right),\phantom{ee}otherwise.
	\end{array}
	\right.
\end{align}
\end{corollary}
\vskip 20pt
\setcounter{equation}{0}
\section{{ {\bf Numerical results}}}\label{s6}
In this section, we calibrate the values of the parameters according to the empirical data. We first set up the baseline model in Subsection \ref{Ncb} to depict the scenario of the U.S.. Then we set up the model to depict the scenario of China in Subsection \ref{UC}. Finally, in Subsection \ref{figures}, based on the baseline model, we explore numerical simulations to study the impacts of the characteristic parameters on the optimal policies, the critical ages and the optimal PAYGO and EET contribution rates.
%In this section, we calibrate the values of the parameters according to the empirical data. We start with data from the United States and China and then set up a baseline model with a reasonable mix of data from these two countries. Moreover, based on the baseline model, we explore numerical simulations to study the impacts of the characteristic parameters on the optimal policies, the critical ages and the optimal PAYGO and EET contribution rates.
\subsection{Baseline model of the U.S.}\label{Ncb}
%\vskip 5pt
The  United States is a typical developed country which suffers from the serious aging problem. And we choose it as the baseline model. For the labor parameters, the participant joins the pension at the age of $a=30$, and retires at the age of $\tau=65$. The maximal survival age is $\omega=100$. The density of the new entrants at time $0$ could be an arbitrary number and we choose $n_0=10$. Moreover, the population growth rate is set by $\rho=-0.005$ to depict the scenario of shrinking population and serious aging problem. The parameters in the survival function are $A=0.000022$, $B=2.7\times10^{-6}$, $c=1.124$, respectively (cf. \cite{DHW2013}). Based on these parameter settings, the dependency ratio is $\Lambda^{-1}=0.7271$. For the market parameters, we have $r=0.02$, $\mu=0.1$, $\sigma=0.26$; $W_0=1$, $\gamma=0.02$, $\xi=0.09$; $ \alpha=0.06$, $\beta=0.12$.  EET has higher Sharpe ratio compared with the individual savings, and the difference is 0.0256.  For the old-age security parameters, the initial PAYGO contribution rate is $\theta_0=0.08$ and the initial EET contribution rate is $k_0=0.12$. The marginal tax rates are $\tau_1=0.25$, $\tau_2=0$. The above data comes from U.S. Bureau of Labor Statistics, U.S. Centers for Disease Control and Prevention, Federal Reserve Economic Data, U.S. Department of the Treasury and Organization for Economic Co-operation and Development (OECD). For the weight parameter, we have $\lambda=1.5$. For the risk averse parameters, we have $\delta_0=-2.8$ for the cohorts who have not been employed yet, $\delta_1=-2.9$ for the working cohorts and $\delta_2=-3$ for the retirees. Besides, the government decides to reselect the optimal PAYGO and EET contribution rates at time $t_0=0$, when the old contribution rates cannot adapt well to the new demographic change. The maximum of the sum of the two contribution rates is $m=25\%$.
\vskip 5pt
Applying the above data to Theorems \ref{them2}-\ref{them4}, we obtain that Case 4 in  Fig. \ref{flow} is consistent with the situation in the U.S.. Within the admissible scope in Eq. \eqref{equ19}, numerically solving the maximum points of the utility functions in Eqs. \eqref{utility1}-\eqref{utility11}, we obtain
the optimal contribution rates $\theta^*=0.1169$, $k^*=0.1331$ and $\theta^{**}=0.1029$, $k^{**}=0.1471$ under the two objectives, respectively. Under the population weighted objective, the PAYGO contribution rate rises dramatically and the EET contribution rate rises slightly. After decades of professional operation, the comparative efficiency of EET pension dominates the individual savings. However, PAYGO pension is more favorable by the older participants. Because the negative population growth rate also increases the weight of the older cohorts, the U.S. government should increase the share of PAYGO pension to improve the overall utility of the society.  From the perspective of intergenerational fairness, we also study the equal weighted objective. The weight of the older cohorts is reduced, thus the optimal PAYGO contribution rate slightly decreases. Besides, the comparative efficiencies of the two obligatory pensions are superior, the sum of the two contribution rates reaches the maximum, i.e., $\theta^*+k^*=25\%$.  From Theorems \ref{them2}-\ref{them4}, we get two critical ages. The critical age of PAYGO pension and individual savings is $\hat{\zeta}=37.5596$. The critical age of PAYGO and EET pensions is $\widetilde{\zeta}=48.3200$.  All non-retired cohorts prefer EET pension to individual savings. Moreover, the voluntary EET case is the same as the mandatory EET case.
%Under the equal weighted objective,
\subsection{Model of  China}\label{UC}

To better understand the different mixes among pensions under heterogeneous characteristic parameters, we also explore China as a typical scenario.  China is a developing country with rapid economic development and income growth rate. However, China also suffers from serious aging problem.  And China's population growth rate is $\rho=-0.004$. The participants join the pension at $a=25$, retire at $\tau=60$ and the maximal age is $\omega=95$. The expected growth rate and the volatility of the average salary are set by $\gamma=0.03$ and $\xi=0.14$. The expected return and the volatility of the risky asset are $\mu=0.08$ and $\sigma=0.2$. And the risk-free interest rate is $r=0.02$. The expected return and the volatility of EET pension are $\alpha=0.05$ and $\beta=0.09$. EET pension also has higher Sharpe ratio compared with individual savings. And the difference is 0.0333, which is higher than that of the U.S.. With the enrichment of private investment experience and the improvement of the capital market efficiency in China, the gap between two countries will be reduced. The initial contribution rates of PAYGO and EET pensions are $\theta_0=0.16$ and $k_0=0.04$. These data come from National Bureau of Statistics of China, Ministry of Finance of China and the annual report of China Life Insurance Company. According to Theorems \ref{them2}-\ref{them4}, we obtain that Case 4 in  Fig. \ref{flow} is also consistent with the situation in China.  All non-retired participants prefer EET pension to individual savings. And the result of voluntary EET pension is the same as that of mandatory EET pension. Within the admissible scope in Eq. \eqref{equ19}, numerically solving the maximum points of the utility functions in Eqs. \eqref{utility1}-\eqref{utility11}, we obtain
the optimal contribution rates $\theta^*=0.1764$, $k^*=0.0736$ and $\theta^{**}=0.1686$, $k^{**}=0.0814$ under the two objectives, respectively. The two critical ages are $\hat{\zeta}=37.3233$, $\widetilde{\zeta}=44.6371$.
%In the voluntary EET case, based on Proposition \ref{k} and Corollary \ref{objective2}, the optimal contribution rate is $\theta^*=0.1799$ and when $\zeta<35.3620$, $k^*(\zeta)=0.0701$; when $35.3620<\zeta<60$, $k^*(\zeta)=0$; when $\zeta=35.3620$ or $\zeta\geq60$, $k^*(\zeta)\in[0,0.0701]$.
Although China suffers from serious aging problem, the comparative efficiency of PAYGO pension is still relatively high due to the rapid salary growth rate. Meanwhile, EET pension has appeared in the recent ten years, and its comparative efficiency gradually improves. At the same time, the comparative efficiency of individual savings is the lowest due to the lack of experience and the capital market fluctuations. In China, the government could increase the overall utility by adding a small obligatory EET share. With the further improvement of the return efficiency of EET pension and the slowdown of the salary growth rate, the optimal mix will be inclined to more EET share in the future.
%Accordingly, PAYGO pension can still provide high benefits for all cohorts. As for EET pension, it is just emerging in China, immature, and thus the investment efficiency is relatively lower. Consequently, at present, PAYGO pension is better than individual saving and individual saving is better than EET pension for all cohorts.
%However, it is worth noting that such a high salary growth rate cannot be expected to exist in the long term and EET pension will become more mature over time. For these reasons, EET pension is also necessary for China and cohorts' preference orderings may change in the near future. Faced with low birth rate and a possible economic recession, developing EET pension is a good way for the government to guarantee retirement benefits.
\subsection{Numerical results of the baseline model}\label{figures}
%\vskip 5pt
Below are the numerical results based on the baseline model. The following analysis is based on the population weighted objective. And we obtain similar results based on the equal weighted objective.
In the first three figures, we study the optimal control policies of the participants at different ages. We select $\zeta=30$, $\zeta=40$, $\zeta=70$, $\zeta=15$ as the typical participants. They represent the participants who are joining the pension at the moment, who are in the working period, who are in the retirement period, and who will join the pension in the future.
\vskip 5pt
Fig. \ref{fig1} shows the expected optimal wealth $\mathrm{E}X^*(t)$ of the participants with respect to time $t$. We obtain $\mathrm{E}X^*(t)$ by martingale method and the details are in Appendix \ref{A.2}. Take the $30$-year-old participants for example.  After joining the plan, the wealth of the participants gradually decreases. Due to the good performance of the mandatory pensions, the individuals are willing to participate in PAYGO pension and accumulate wealth in the EET pension account. Thus, they initiate debts to meet consumption requirements. The turning point of the wealth is at the retirement age. At that time, the participants start to receive benefits from the pensions and the debts are gradually repaid.
\vskip 5pt
\begin{figure}
\includegraphics[totalheight=5.5cm]{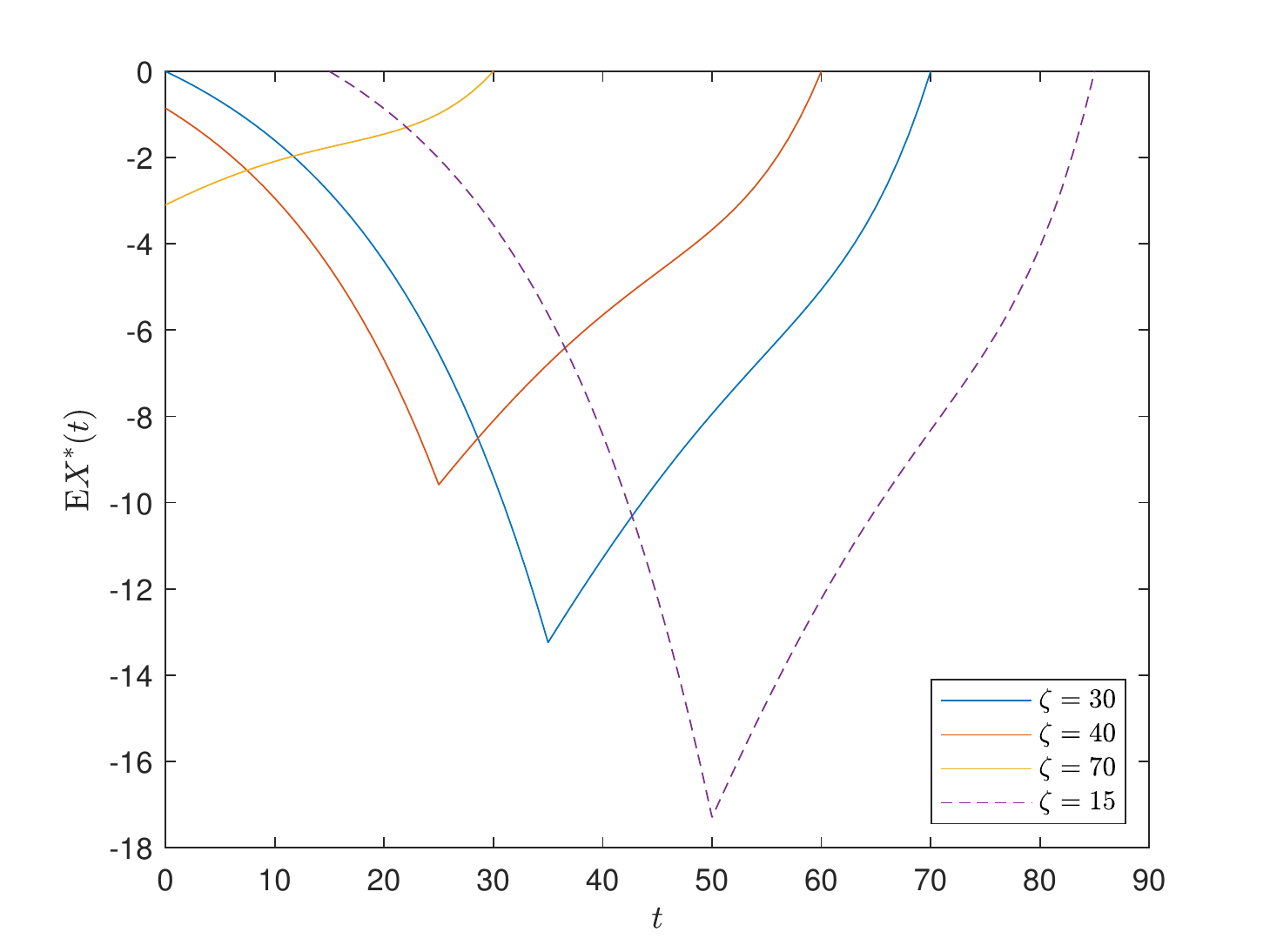}
\caption{Expected optimal wealth $\mathrm{E}X^*(t)$ of the participants}
\label{fig1}
\end{figure}
\vskip 5pt
\begin{figure}[h]
    \includegraphics[totalheight=5.5cm]{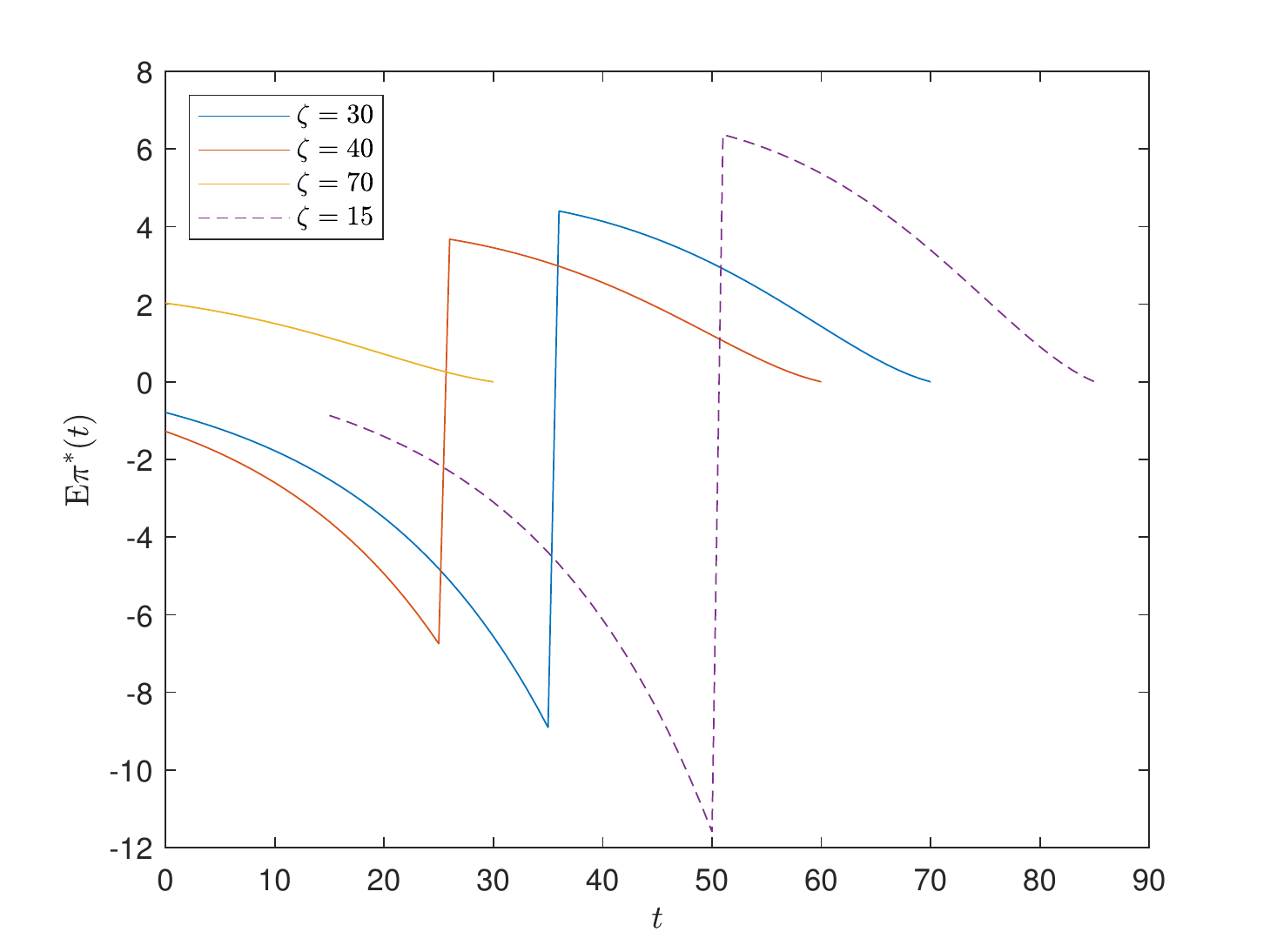}
    \caption{Expected optimal risky investment $\mathrm{E}\pi^*(t)$ of the participants}
    \label{fig2}
\end{figure}

Fig. \ref{fig2} shows the expected optimal risky investment $\mathrm{E}\pi^*(t)$ of the participants with respect to time $t$. % Basically, the optimal risky investment is positively correlated with the optimal wealth, which is consistent with the classical \cite{Merton1971} model.
The discontinuity point at the age of $65$ is caused by the change of the investment flexibility in EET pension. When the participants retire, EET pension is fully converted into life annuities and its risky investment suddenly decreases to zero. Therefore, individual risky investment increases accordingly.

\vskip 5pt

%Fig.\ref{fig2} shows the expected optimal risky investment ratio $E\pi^*(t)$ of the participants with respect to time t. When the participant is young, a high ratio of wealth is allocated to the risky asset. The optimal risky investment ratio decreases with age. The discontinuity point at the age of 65 is caused by the changes in EET pension. EET pension can be regarded as a special risk asset. When participants retire, their investment in EET pension suddenly decreases to zero, and therefore, investment in other risky assets suddenly increases. Notably, due to the good market conditions, relatively high return of risky assets and relatively small volatility of risky assets, overall investment is aggressive.
\begin{figure}[h]
    \includegraphics[totalheight=5.5cm]{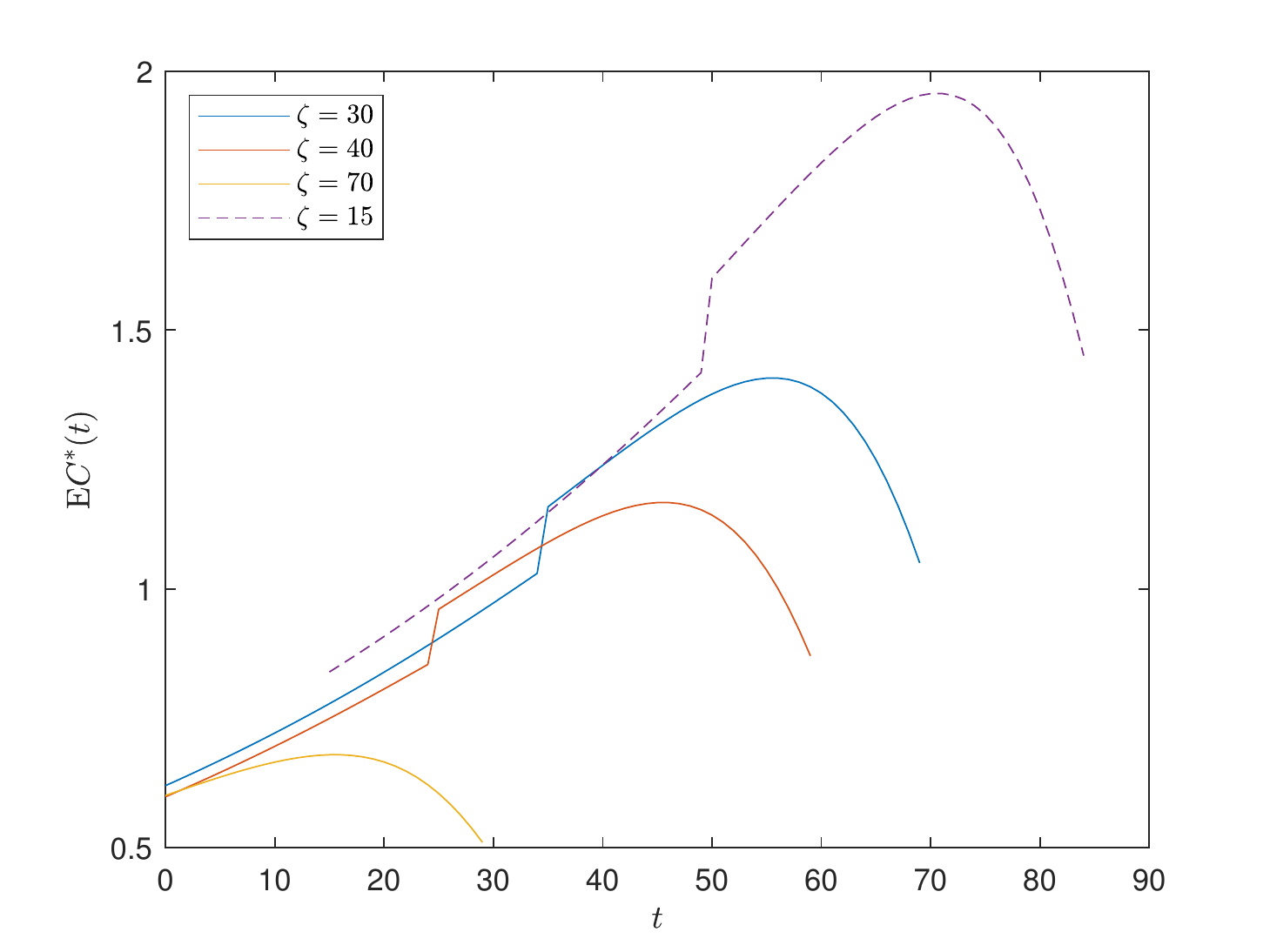}
    \caption{Expected optimal consumption $\mathrm{E}C^*(t)$ of the participants}
    \label{fig3}
\end{figure}

Fig. \ref{fig3} shows the expected optimal consumption $\mathrm{E}C^*(t)$ of the participants with respect to time $t$. The participant consumes more after retirement when consumption produces higher utility. The weight parameter $\lambda=1.5$ causes the discontinuity point at the age of $65$.  As shown in Remark \ref{unnc}, the decline in consumption in old age can be regarded as a decline in unnecessary consumption, which is in line with reality. %And the turning point at the age of $87$ corresponds to the decline in wealth.

\vskip 5pt

\begin{figure}[h]
    \includegraphics[totalheight=5.5cm]{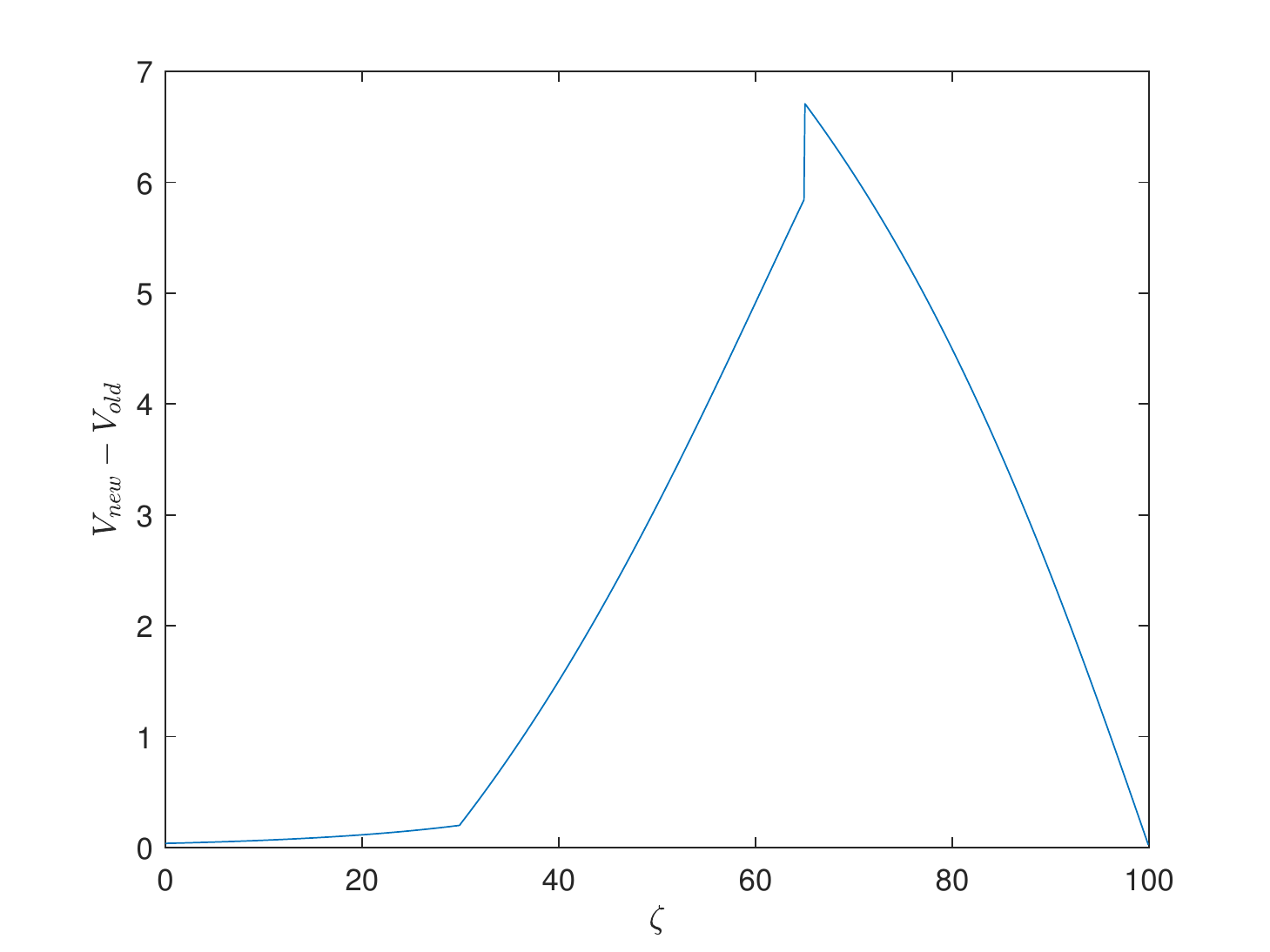}
    \caption{Variation of participants' value functions}
    \label{fig4}
\end{figure}

In Fig. \ref{fig4}, we study the variation of the participants' value
function $V_{new}-V_{old}$ with respect to the age $\zeta$. The
variation of the value function exhibits a hump shape curve. Based on the parameters in the baseline model, non-retirees prefer EET pension to individual savings and retirees treat the two as the same. Meanwhile, the younger cohort prefers individual savings to PAYGO pension and the older cohort is the contrary. Furthermore, the PAYGO contribution rate increases dramatically and the EET contribution rate increases slightly under the new parameter selection. Thus, the utility improvement from the increase of EET pension is diminished by the utility loss from the increase of PAYGO pension for the younger cohort. On the contrary, the increase of the PAYGO and EET contribution rates both improves the utility of the older cohort. Thus, the utility improvement of the older cohort rises sharply with age. However, the exceeding utility drops after retirement because that the remaining utility gradually decays to zero. Interestingly, there are two turning points at ages of 30 and 65. This is due to the difference in risk aversion parameters of the three cohorts. The utility gain increases sharply with the increase of risk aversion, which also means that the government gives more consideration of the cohorts with higher risk aversion in decision-making.

%Thus, the utility loss caused by the
%rise of the PAYGO contribution rate becomes the dominance for the younger cohorts. On the
%contrary, the increase of the PAYGO contribution rate directly
%improves the utility of the older cohorts. Thus, the older cohort's
%utility rises and eventually exceeds the original utility around the
%age of $40$. However, the exceeding utility drops after retirement
%because that the remaining utility gradually decays to zero.
\vskip 5pt
In Fig. \ref{fig5}, we exhibit the variation tendency of the different
cohorts' value function with respect to $\theta$ (with fixed
$k=k^{*}$), $k$ (with fixed $\theta=\theta^{*}$), and $\theta$ (with
fixed $k+\theta=k^{*}+\theta^{*}$). In the first case, we observe
the downward trend of the younger cohort's value function and the
upward trend of the older cohort's value function along with the
increase of PAYGO contribution rate $\theta$. The increase of
$\theta$ rises the contribution burden of the younger cohort and
this dominates the increase of the benefit. However, because the
contribution period is relatively short for the older cohort, the
increase of $\theta$ rises the benefit directly and this becomes the
dominant effect eventually. In the second case, we observe upward
trend of the non-retired cohort's value function along with the increase
of EET contribution rate. This is due to the high investment return efficiency during the accumulation period and the high withdrawal rate after retirement.
%For the younger cohort, the impact of the high
%return efficiency during the long accumulation period exceeds the
%one of the inflexibility and the low return efficiency after retirement.
Moreover, the trend is relatively stable for the older cohort because
that the EET account is not affected by the contribution rate after
retirement. In the third case, the result is the combination of the
first two cases, thus we observe more significant trends. Furthermore, how the government balances the heterogeneous preference orderings of different cohorts is a complex issue to be studied. Considering intergenerational differences and intergenerational fairness, we study the population weighted and equal weighted objectives.

\begin{figure}
    \centering
    \begin{subfigure}{0.49\textwidth}
        \centering
        \includegraphics[width=1\textwidth]{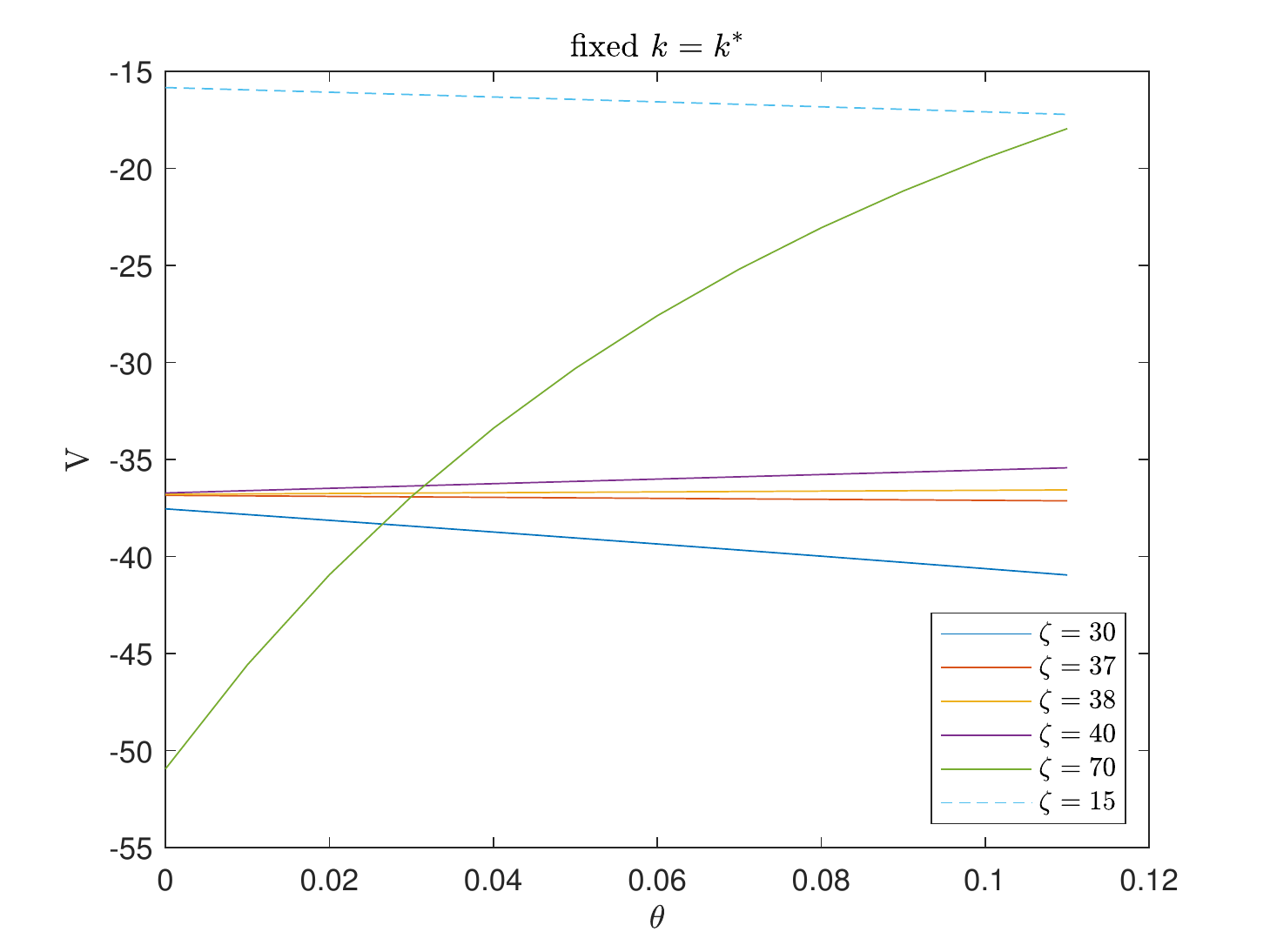}
    \end{subfigure}
    \begin{subfigure}{0.49\textwidth}
        \centering
        \includegraphics[width=1\textwidth]{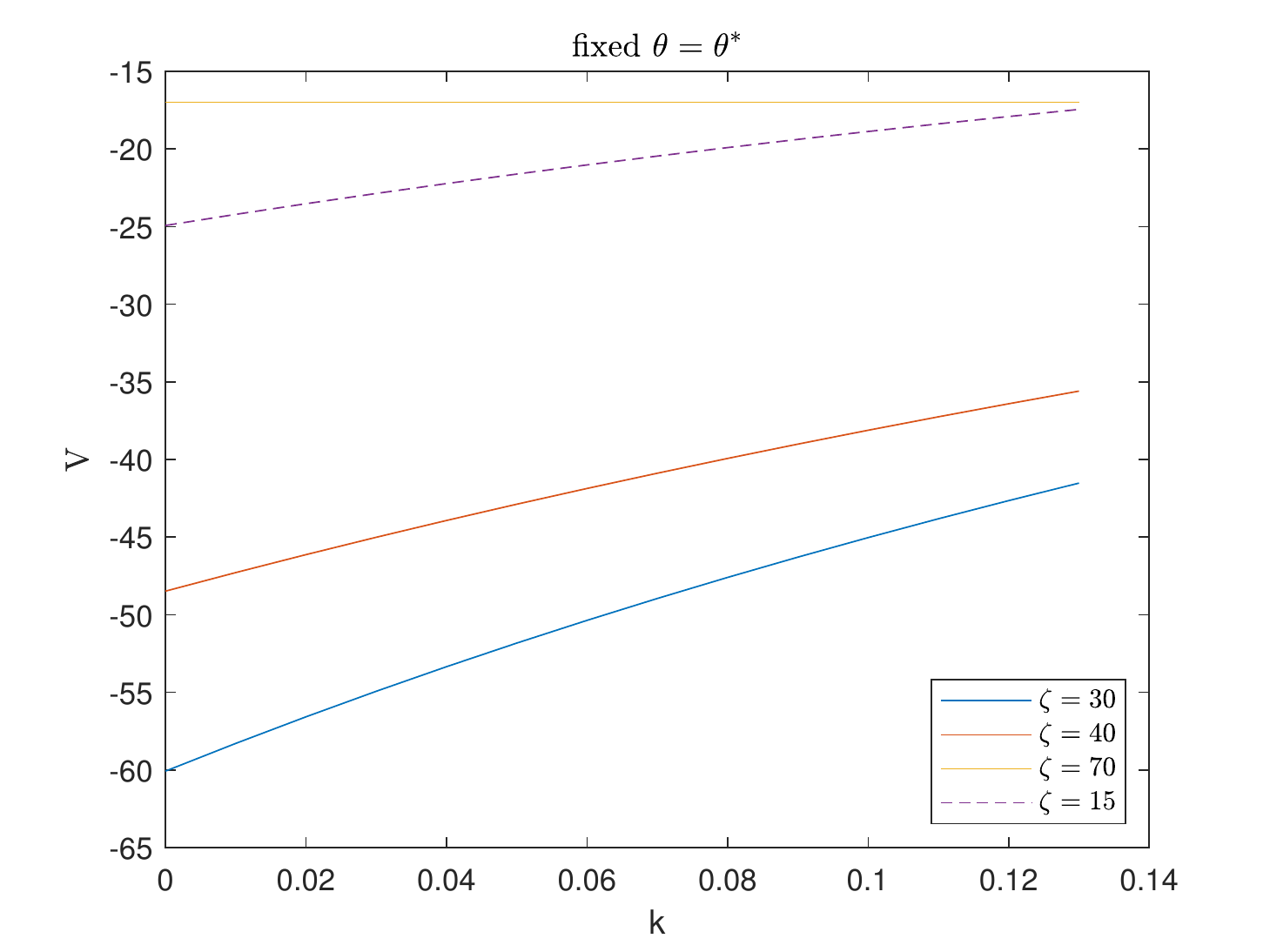}
    \end{subfigure}
    \begin{subfigure}{0.49\textwidth}
        \centering
        \includegraphics[width=1\textwidth]{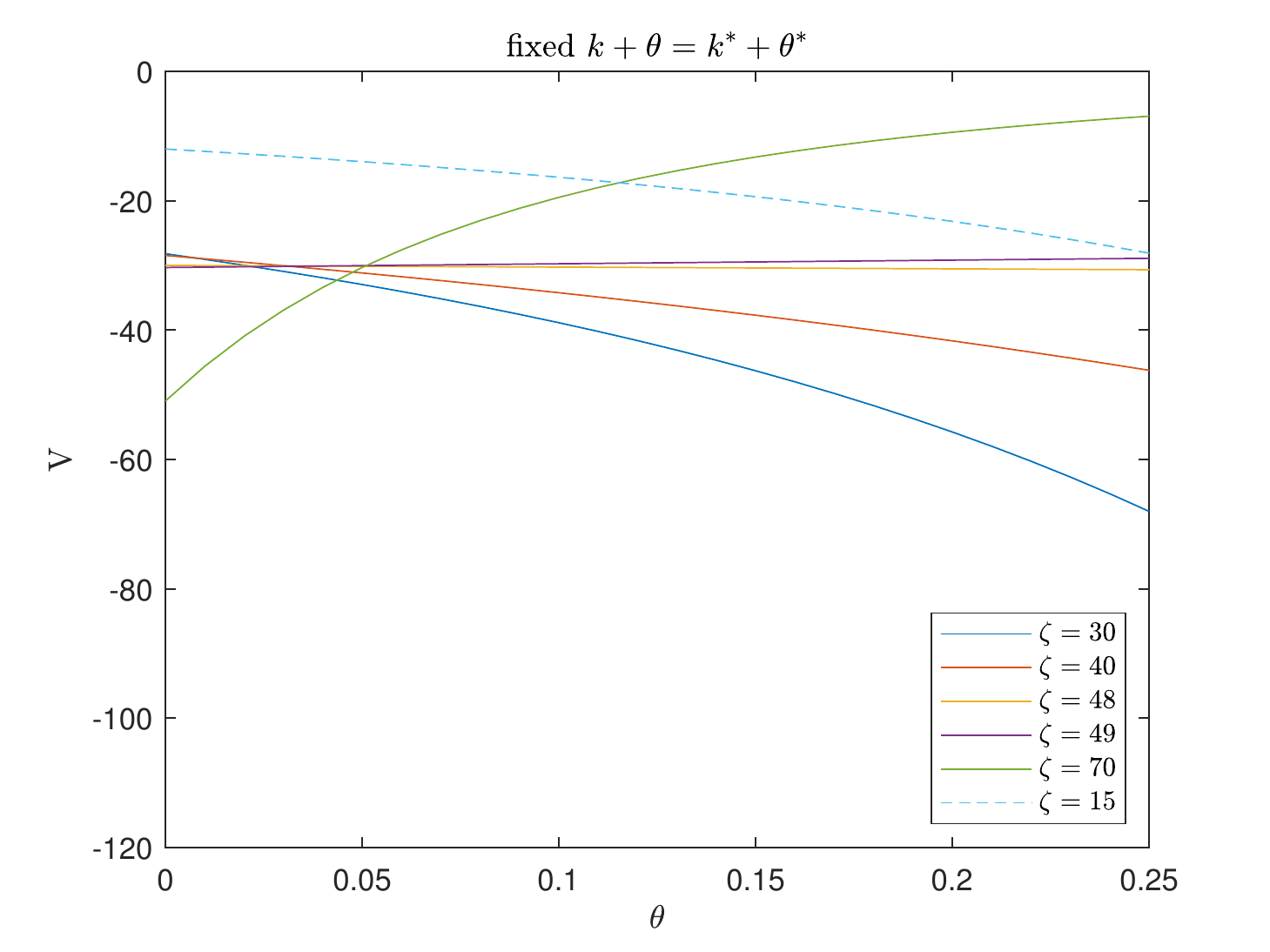}
    \end{subfigure}
\caption{Variation tendency of the value function}
\label{fig5}
\end{figure}
\vskip 5pt
In Fig. \ref{fig6}, we study the preference among PAYGO pension, EET pension and individual savings of different cohorts. In the first two cases,
$\widetilde{M_{1}}(\zeta)$ and $\widetilde{M_{1}}(\zeta)-\widetilde{M_{2}}(\zeta)$ are
the decisive coefficients in determining the preference between
PAYGO pension and individual savings, and between PAYGO and EET pensions,
respectively. We observe similar hump shape trends in the two cases.
Obviously, the older cohort prefers PAYGO because that the rise of
PAYGO contribution rate increases the benefit directly. In the third
case, positive $\widetilde{M_{2}}(\zeta)$ represents the preference for
EET pension.  Participants prefer EET pension due to its high investment efficiency under professional management and reasonable annualized withdrawal after retirement.
%and the opposite represents the preference for individual saving. For the younger cohort, the high return efficiency
%during the long accumulation period exceeds the inflexibility and the
%low return efficiency after retirement in EET pension. For the
%middle-aged cohorts, the remaining accumulation period is relatively
%short and they concern about the low return efficiency after
%retirement in EET pension. Thus, they prefer individual saving
%due to flexibility and high return efficiency after retirement.
Meanwhile, for the older cohorts who are retired, the rise of the EET
contribution rate does not affect the value function. Thus, they
treat EET pension and individual savings equally.

\begin{figure}
    \centering
    \begin{subfigure}{0.49\textwidth}
        \centering
        \includegraphics[width=1\textwidth]{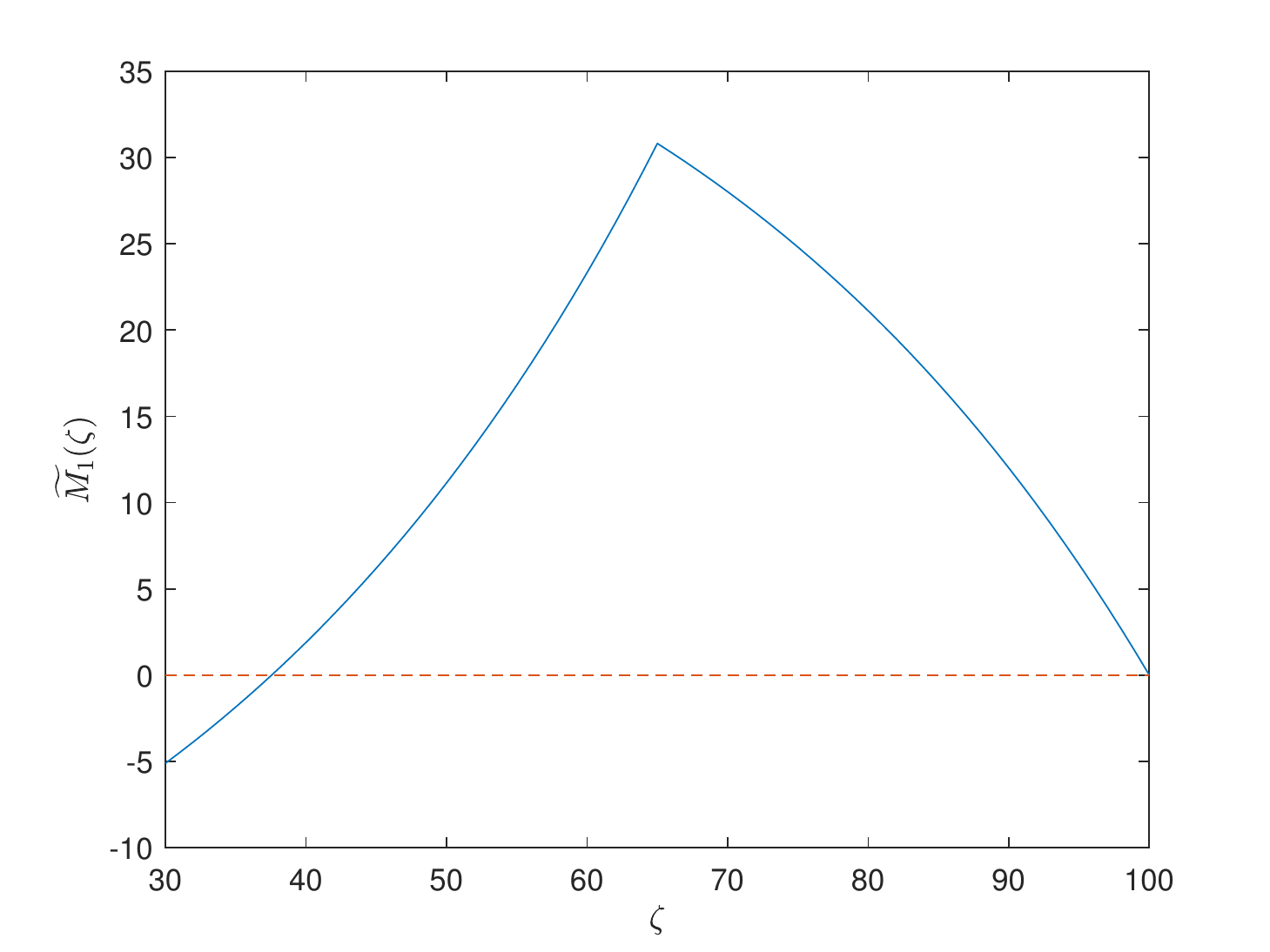}
    \end{subfigure}
    \begin{subfigure}{0.49\textwidth}
        \centering
        \includegraphics[width=1\textwidth]{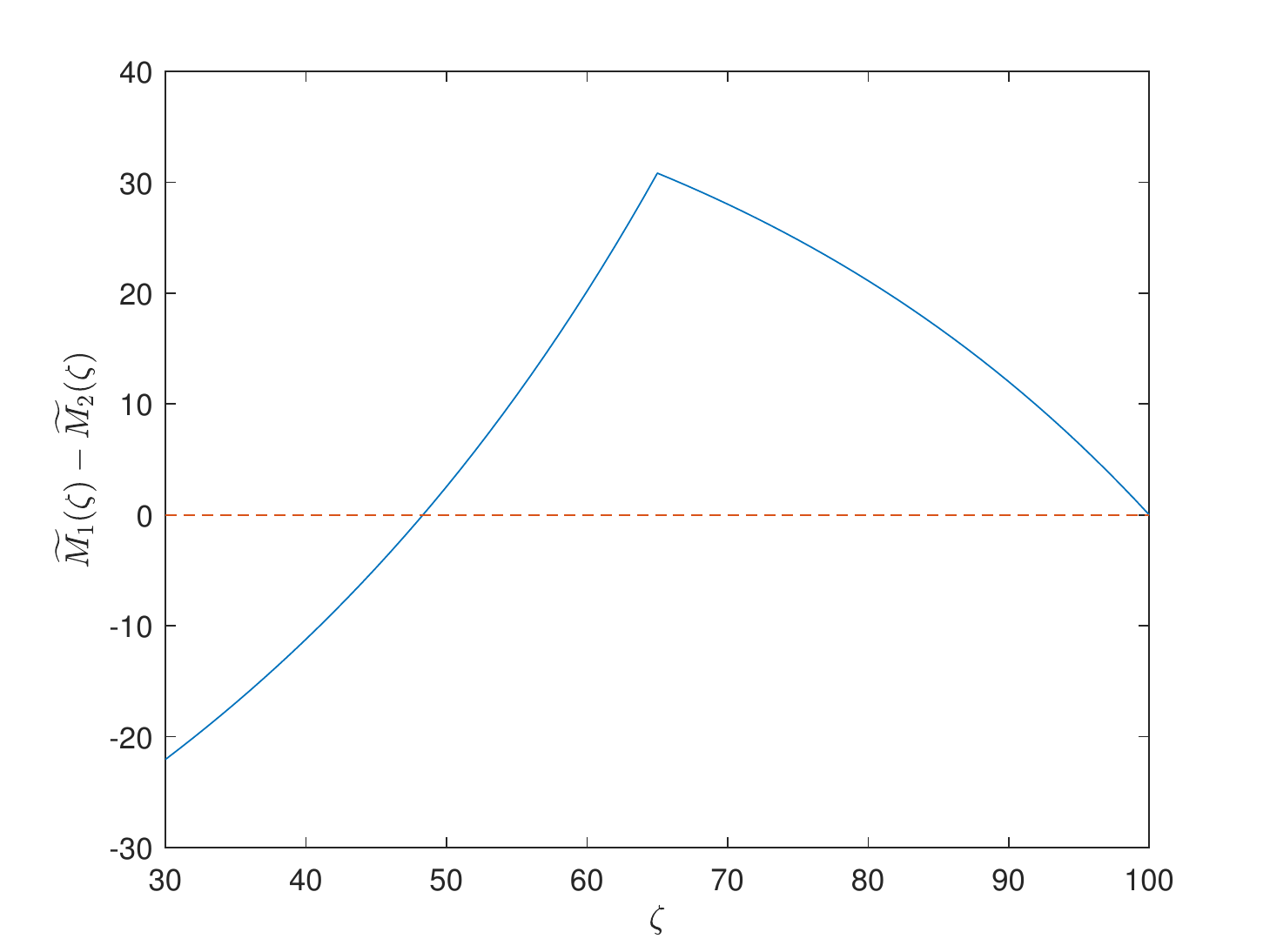}
    \end{subfigure}
    \begin{subfigure}{0.49\textwidth}
        \centering
        \includegraphics[width=1\textwidth]{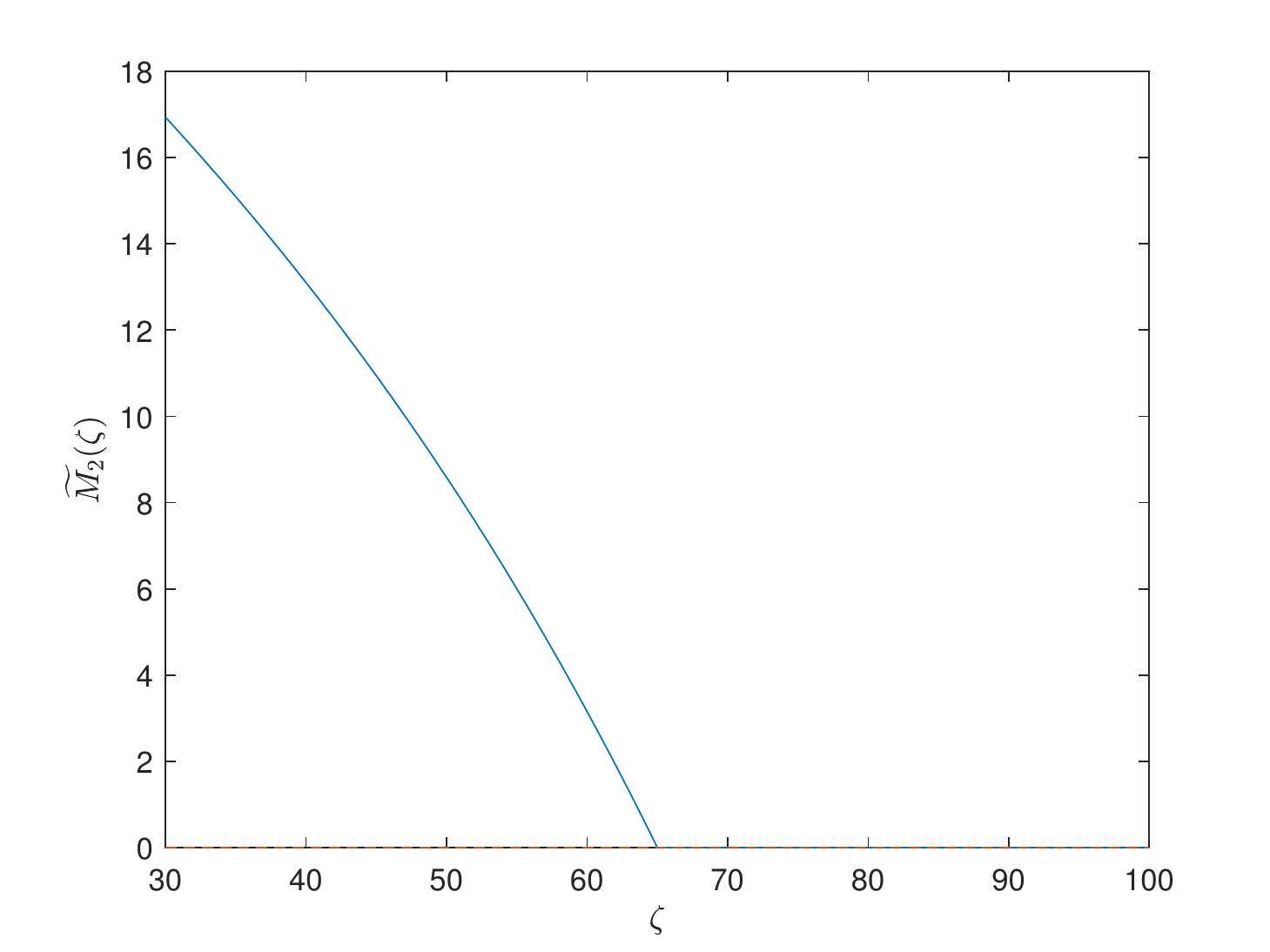}
    \end{subfigure}
\caption{Preference among PAYGO, EET and individual savings}%at different ages
\label{fig6}
\end{figure}
\vskip 5pt
In the following three figures, we study the impacts of the
exogenous parameters, especially the population growth rate $\rho$
on the critical ages. In the first case of Fig. \ref{fig7}, we study the
impacts of $\rho$ and $\mu$ on the critical age $\hat{\zeta}$, which
is the boundary age for the preference between PAYGO pension and individual savings. When the population growth rate decreases,  PAYGO
pension is less preferable and thus the critical age rises.
Meanwhile, when the individual risky investment return increases, individual savings is more preferable and thus the
critical age rises too. These results are consistent with the
criterions in \citet{SAM1958} and \citet{AAR1966}. Particularly,
there is a truncation area when $\rho$ is extremely high and $\mu$
is extremely low. In this circumstance, all the cohorts prefer PAYGO
pension. In the second case of Fig. \ref{fig7}, we study the impacts of
$\rho$ and $\gamma$ on the critical age $\hat{\zeta}$. When the
salary growth rate increases, PAYGO pension is more preferable and
the critical age declines accordingly. In the third case of Fig. \ref{fig7},
we study the impacts of $\rho$ and $\Delta$ on the critical age
$\hat{\zeta}$. $\Delta$ represents the variation proportion of the
parameters $A$ and $B$ ($\Delta \triangleq \frac{A^{'}}{A}=\frac{B^{'}}{B}$) in the force of mortality model (\ref{mortality}). And lower $\Delta$
represents lower probability of death and thus the larger longevity
risk. We observe that when the longevity problem becomes more
serious, PAYGO pension is less preferable and the critical age
rises.  In the last case of Fig. \ref{fig7}, we study the impacts of $\rho$
and $\tau$ on the critical age $\hat{\zeta}$. Interestingly,
postponing retirement first increases and then decreases the
critical age. According to the common sense, postponing retirement
reduces the attractiveness of PAYGO pension because of the longer
contribution  period and the shorter benefit period. However, when
the retirement time is very late, the benefit received is so
sufficient that it can offset the adverse impact of the short
benefit period. Thus, the critical age declines eventually.

\begin{figure}
    \centering
    \begin{subfigure}{0.49\textwidth}
        \centering
        \includegraphics[width=1\textwidth]{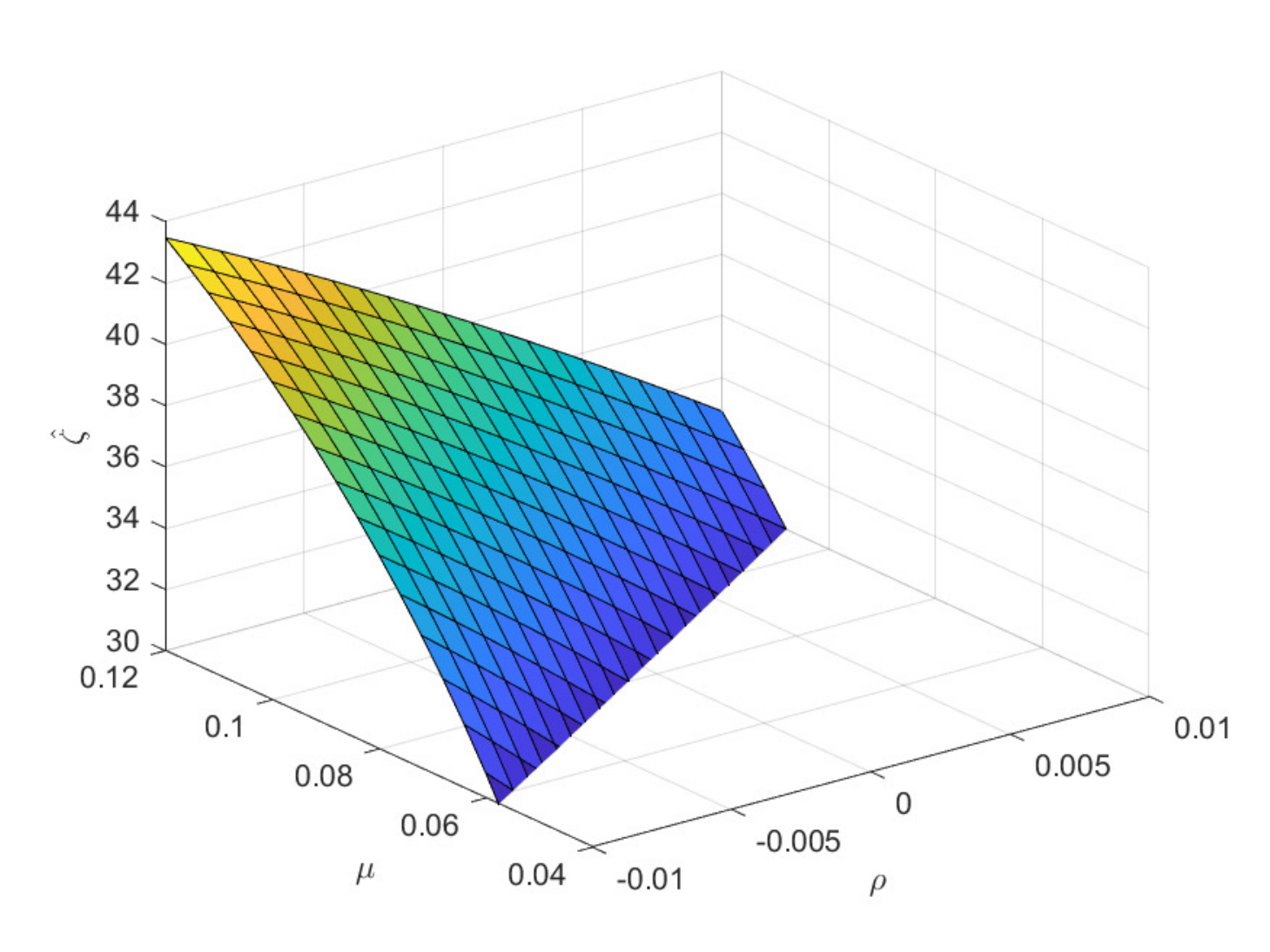}
    \end{subfigure}
    \begin{subfigure}{0.49\textwidth}
        \centering
        \includegraphics[width=1\textwidth]{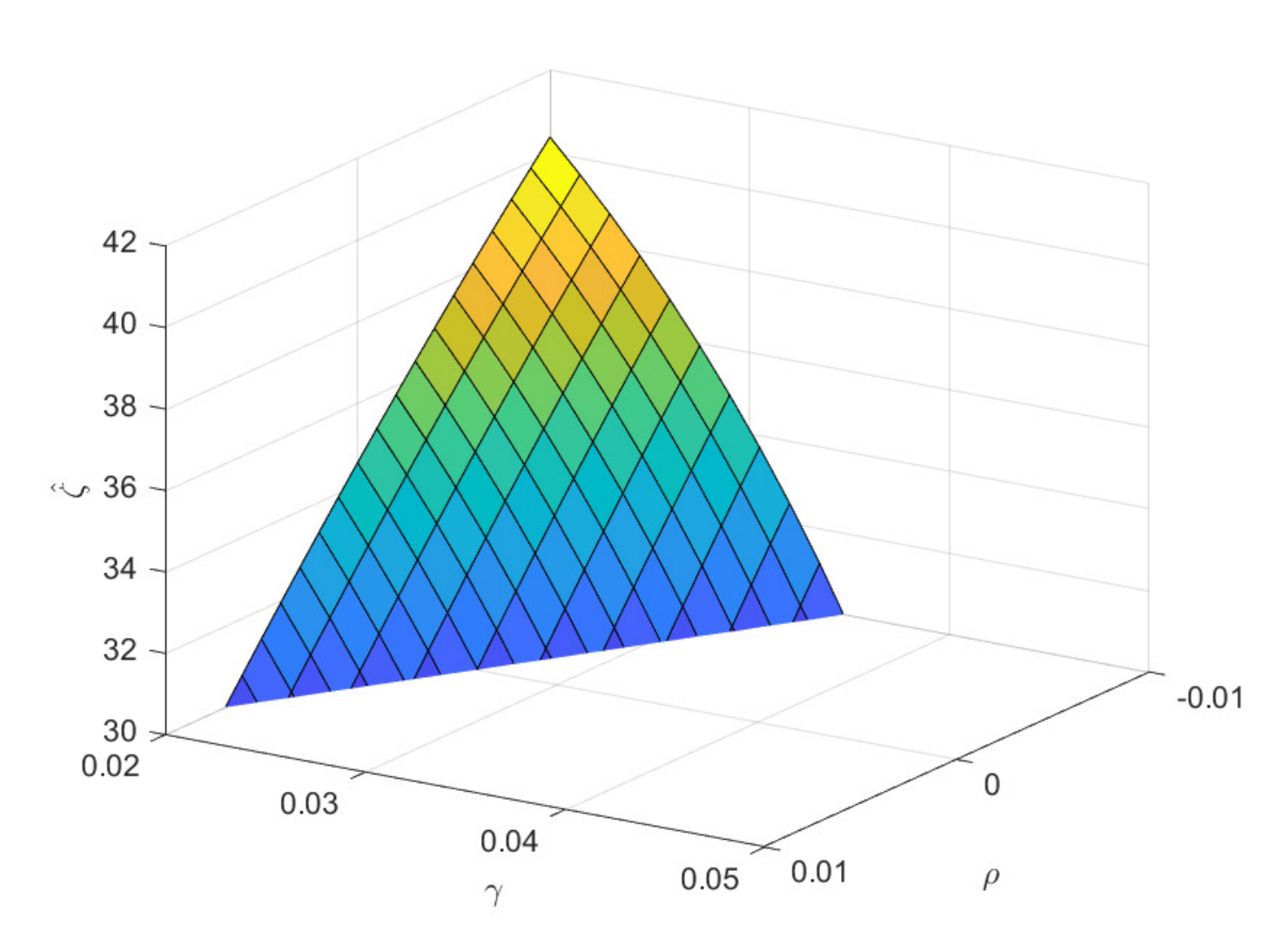}
    \end{subfigure}
    \begin{subfigure}{0.49\textwidth}
        \centering
        \includegraphics[width=1\textwidth]{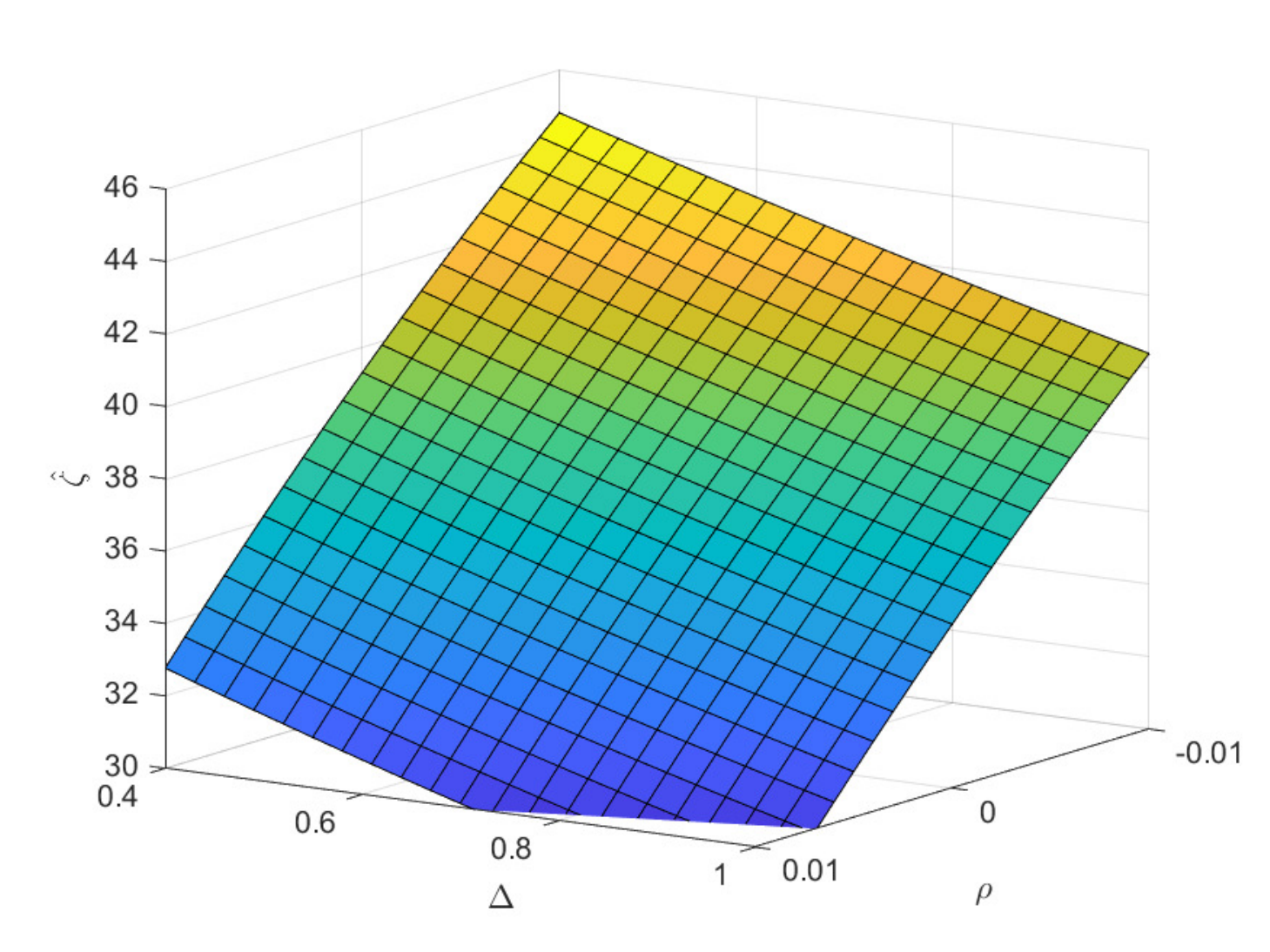}
    \end{subfigure}
    \begin{subfigure}{0.49\textwidth}
        \centering
        \includegraphics[width=1\textwidth]{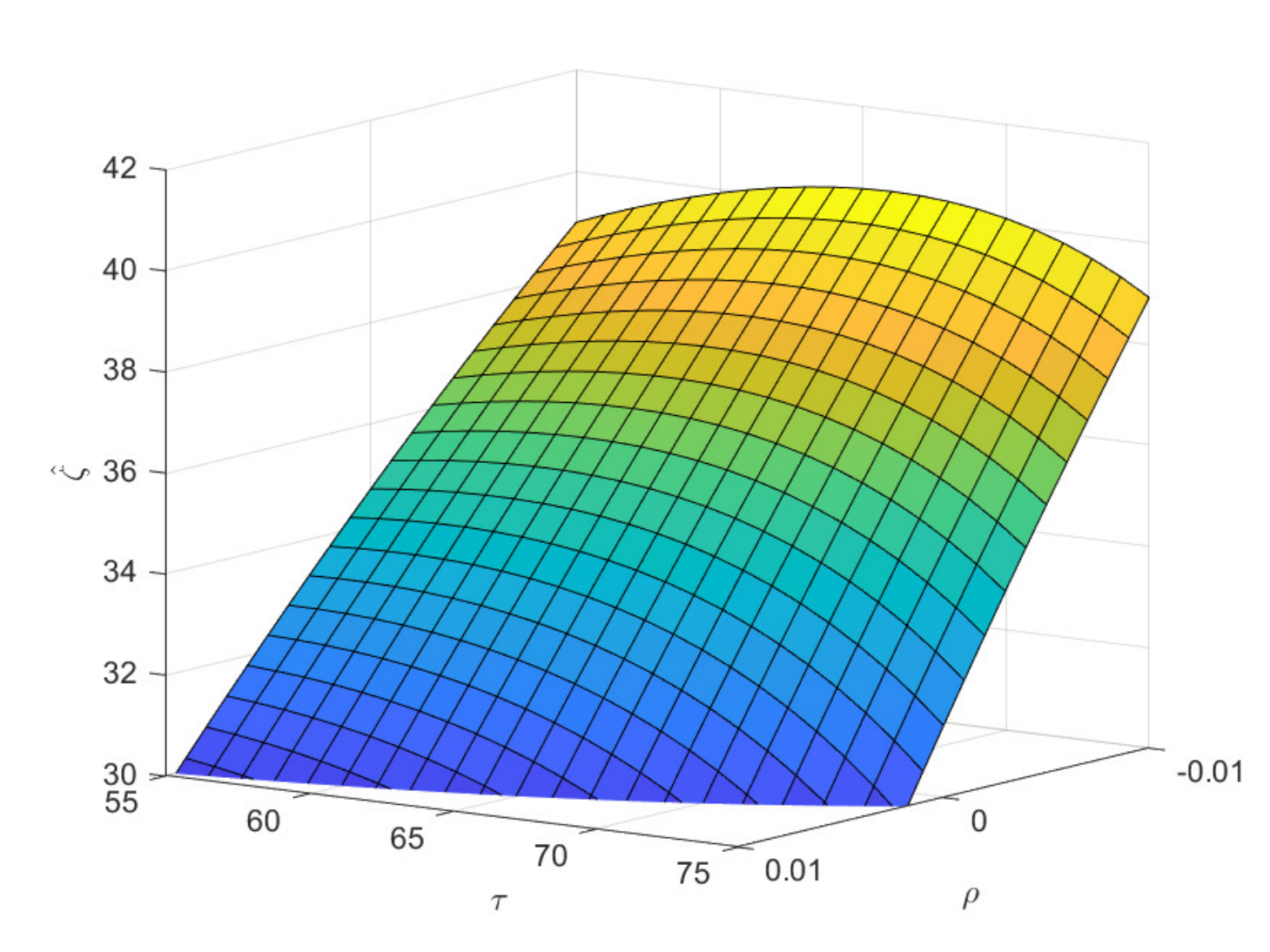}
    \end{subfigure}
\caption{Impacts of $\rho$ and $\mu$, $\rho$ and $\gamma$, $\rho$ and $\Delta$, $\rho$ and $\tau$ on critical age $\hat{\zeta}$}
\label{fig7}
\end{figure}
\vskip 5pt
In the first case of Fig. \ref{fig8}, we study the impacts of $\rho$ and
$\tau_{2}$ on the critical age $\tilde{\zeta}$, which is the
boundary age for the preference between PAYGO and EET pensions.
$\tau_{2}$ is the marginal tax rate for the benefit in EET pension.
When $\tau_{2}$ declines, EET pension becomes more preferable due to
the increase of tax saving effect. Thus, the critical age rises
and more cohorts prefer EET pension in this circumstance. In practice, the government
can increase the preferential taxation incentives to improve the attractiveness of EET pension.
Particularly, $\tau_1$ has the opposite effect as $\tau_2$, and we omit the discussions on $\tau_1$. Meanwhile, the decrease of $\rho$ reduces the attractiveness of
PAYGO pension and increases the critical age accordingly. In the
second case of Fig. \ref{fig8}, we study the impacts of $\rho$ and $\gamma$
on the critical age $\tilde{\zeta}$. The result is consistent with
the one in the second case of Fig. \ref{fig7}. In the third case of Fig. \ref{fig8},
we study the impacts of $\rho$ and $\alpha$ on the critical age
$\tilde{\zeta}$. When the investment return of EET pension rises,
EET pension naturally becomes more preferable. And the critical age
rises accordingly. In the last case of Fig. \ref{fig8}, we study the impacts
of $\rho$ and $\tau$ on the critical age $\tilde{\zeta}$. We observe
a hump shape structure similar with the one in the last case of Fig. \ref{fig7}.
However, the terminal decline is less significant. When the
retirement time is very late, the attractiveness of EET pension also
increases, which is due to the enlarge of the accumulation period
and the return efficiency in this period is high. Thus, the slight
drop in the critical age is the synthetic result when both pensions
become more attractive at the same time.
\begin{figure}
    \centering
    \begin{subfigure}{0.49\textwidth}
        \centering
        \includegraphics[width=1\textwidth]{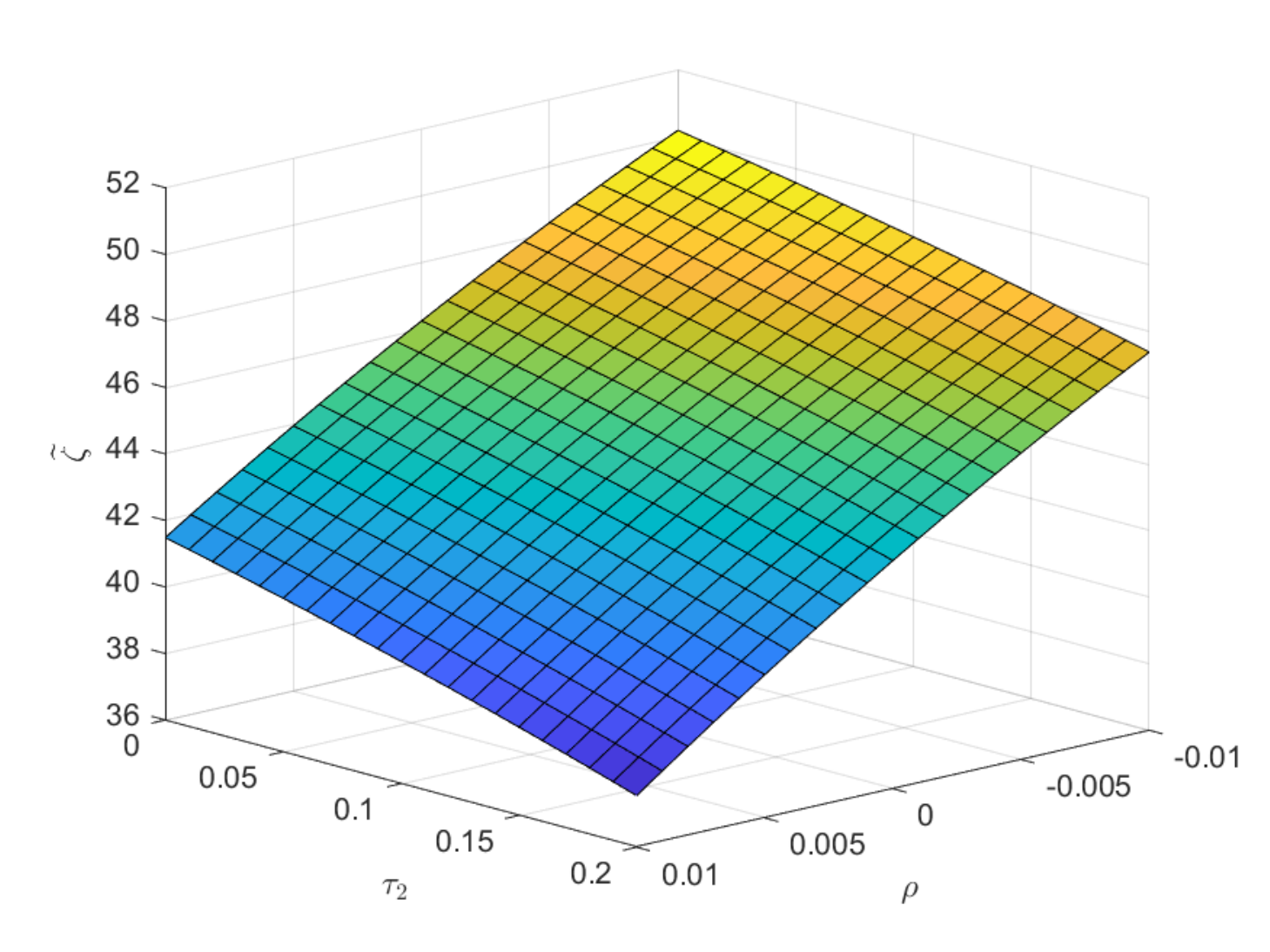}
    \end{subfigure}
    \begin{subfigure}{0.49\textwidth}
        \centering
        \includegraphics[width=1\textwidth]{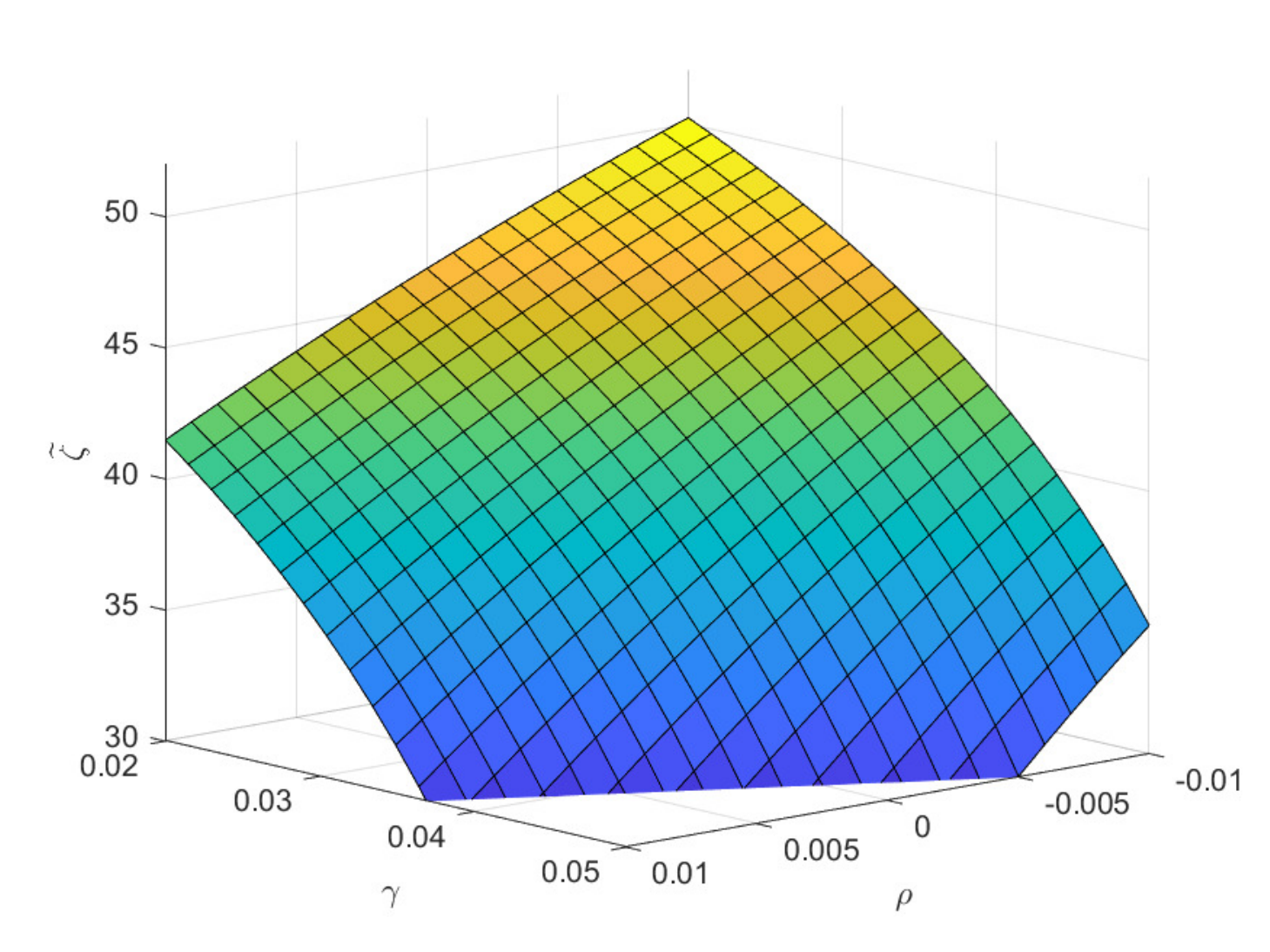}
    \end{subfigure}
    \begin{subfigure}{0.49\textwidth}
        \centering
        \includegraphics[width=1\textwidth]{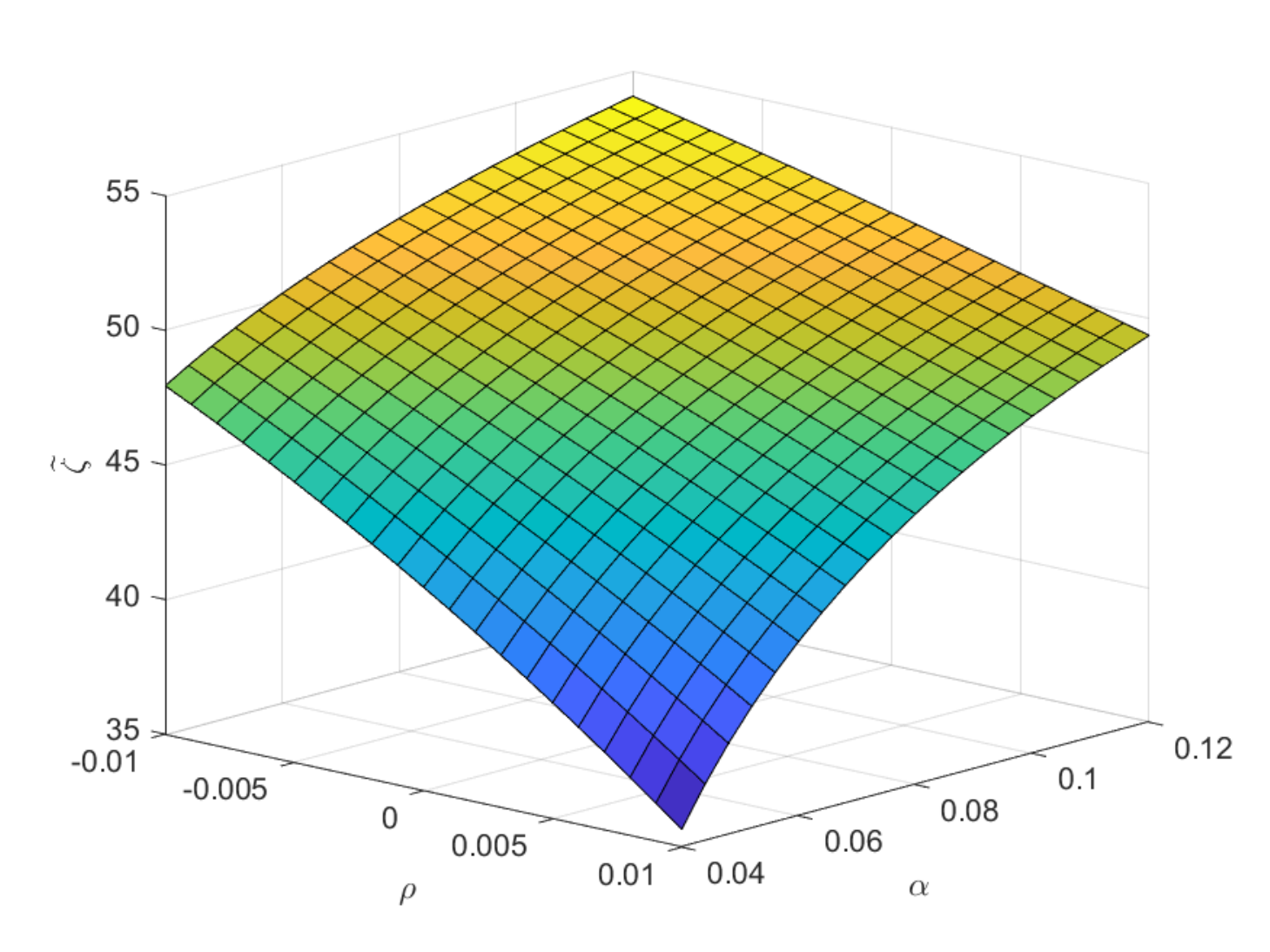}
    \end{subfigure}
    \begin{subfigure}{0.49\textwidth}
        \centering
        \includegraphics[width=1\textwidth]{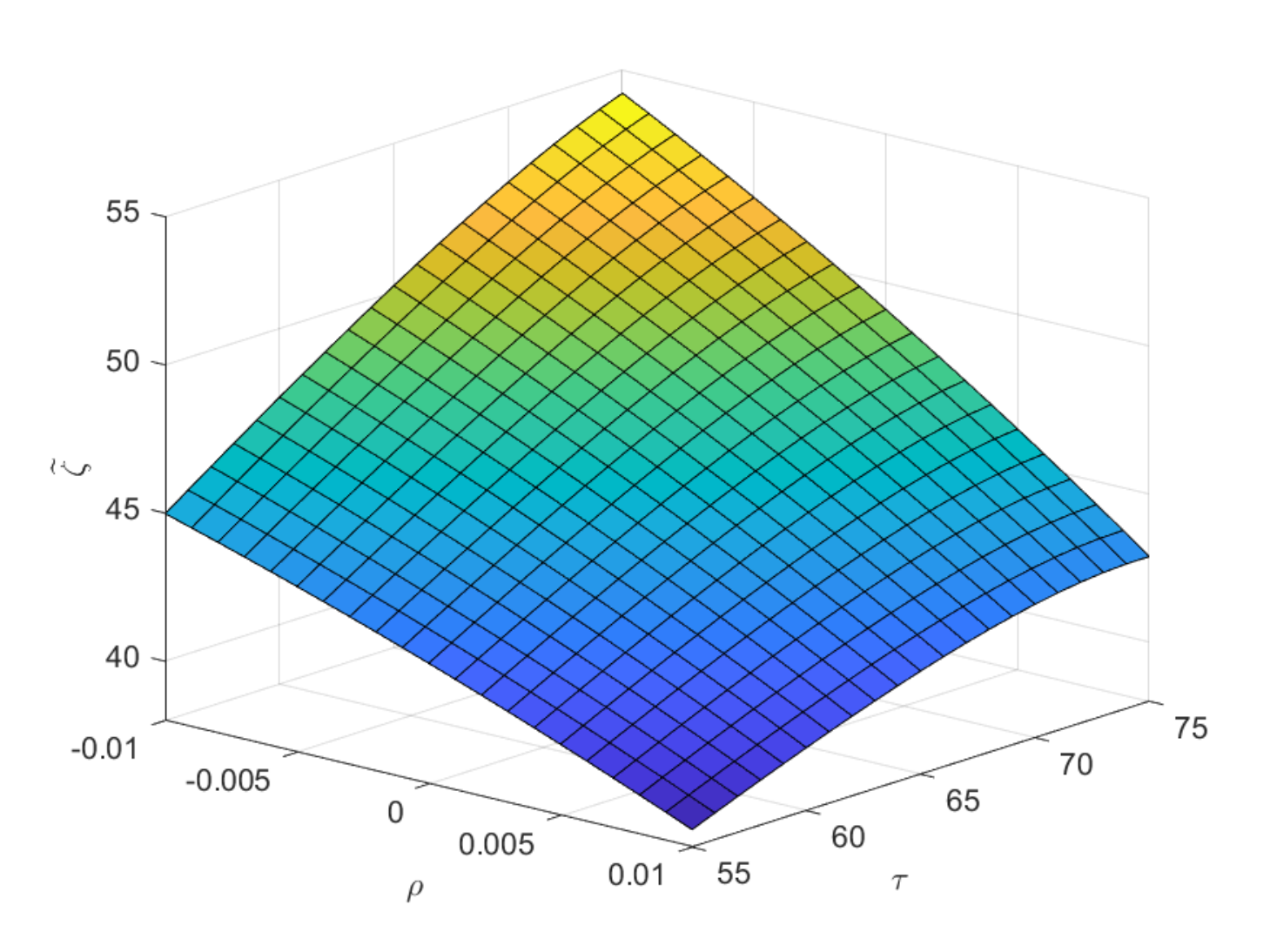}
    \end{subfigure}
    \caption{Impacts of $\rho$ and $\tau_2$, $\rho$ and $\gamma$, $\rho$ and $\alpha$, $\rho$ and $\tau$ on critical age $\widetilde{\zeta}$}
    \label{fig8}
\end{figure}
\begin{figure}[h]
    \centering
    \begin{subfigure}{0.45\textwidth}
        \centering
        \includegraphics[width=0.9\textwidth]{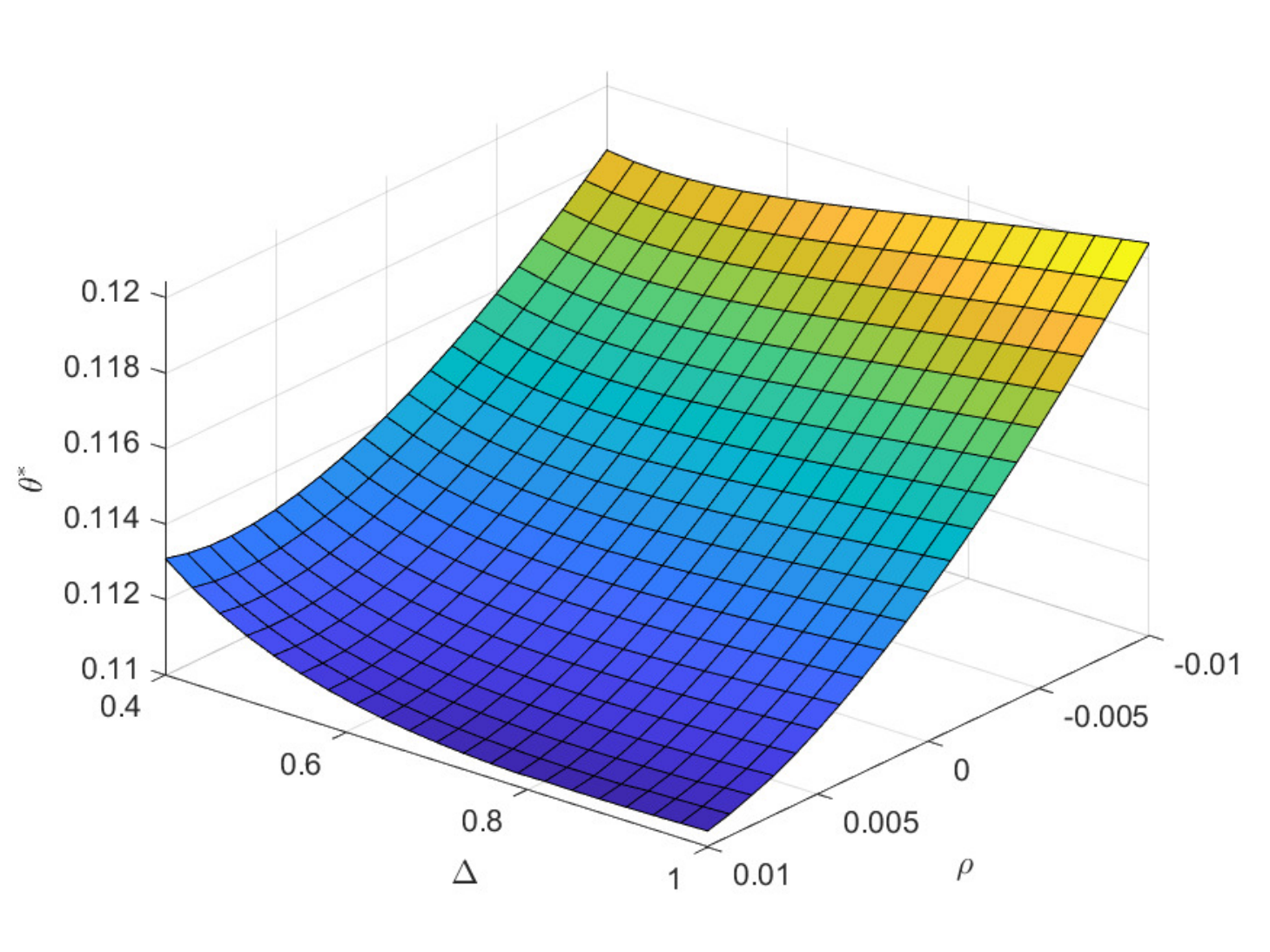}
    \end{subfigure}
    \begin{subfigure}{0.45\textwidth}
        \centering
        \includegraphics[width=0.9\textwidth]{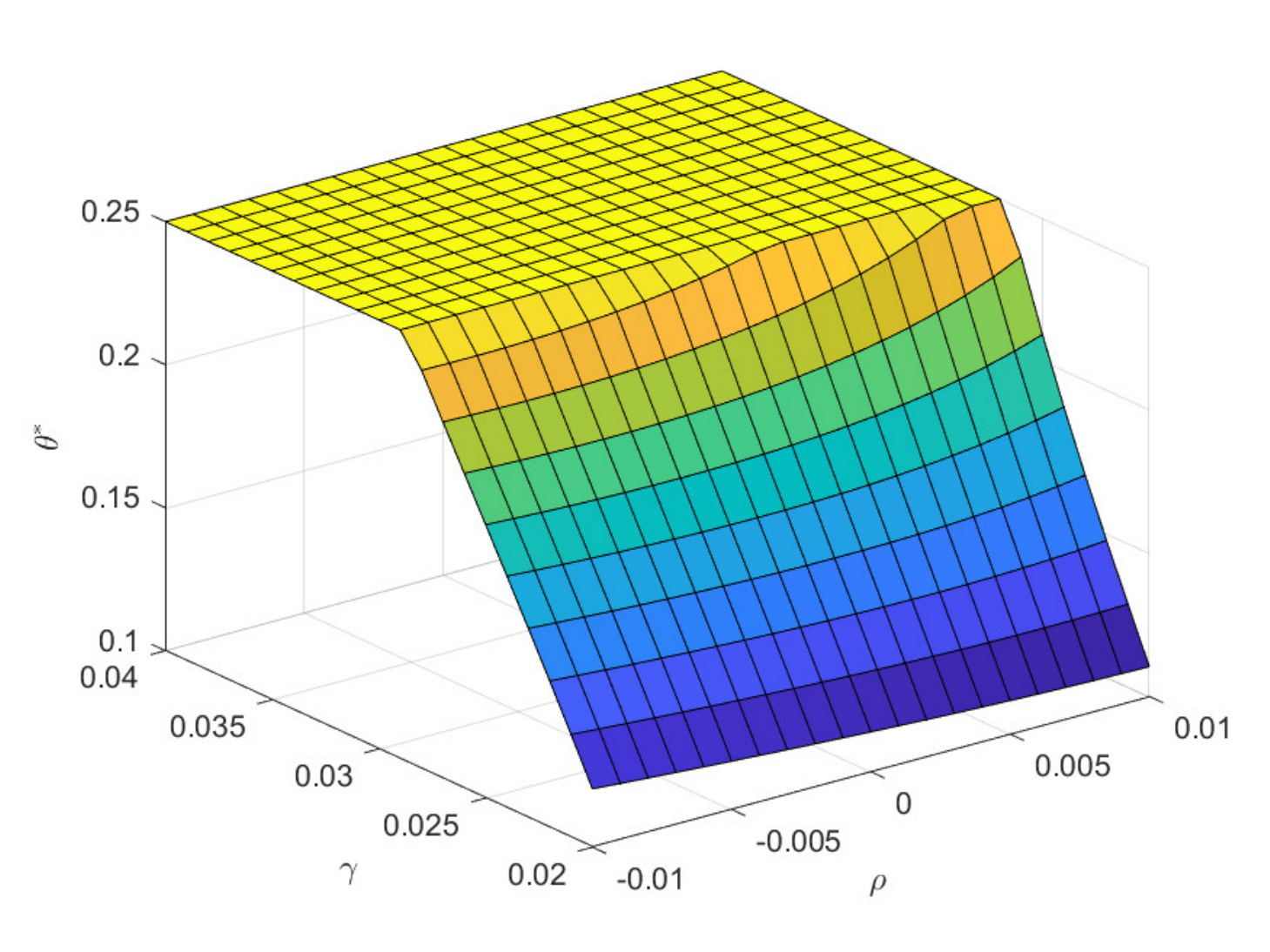}
    \end{subfigure}
\begin{subfigure}{0.45\textwidth}
    \centering
    \includegraphics[width=0.9\textwidth]{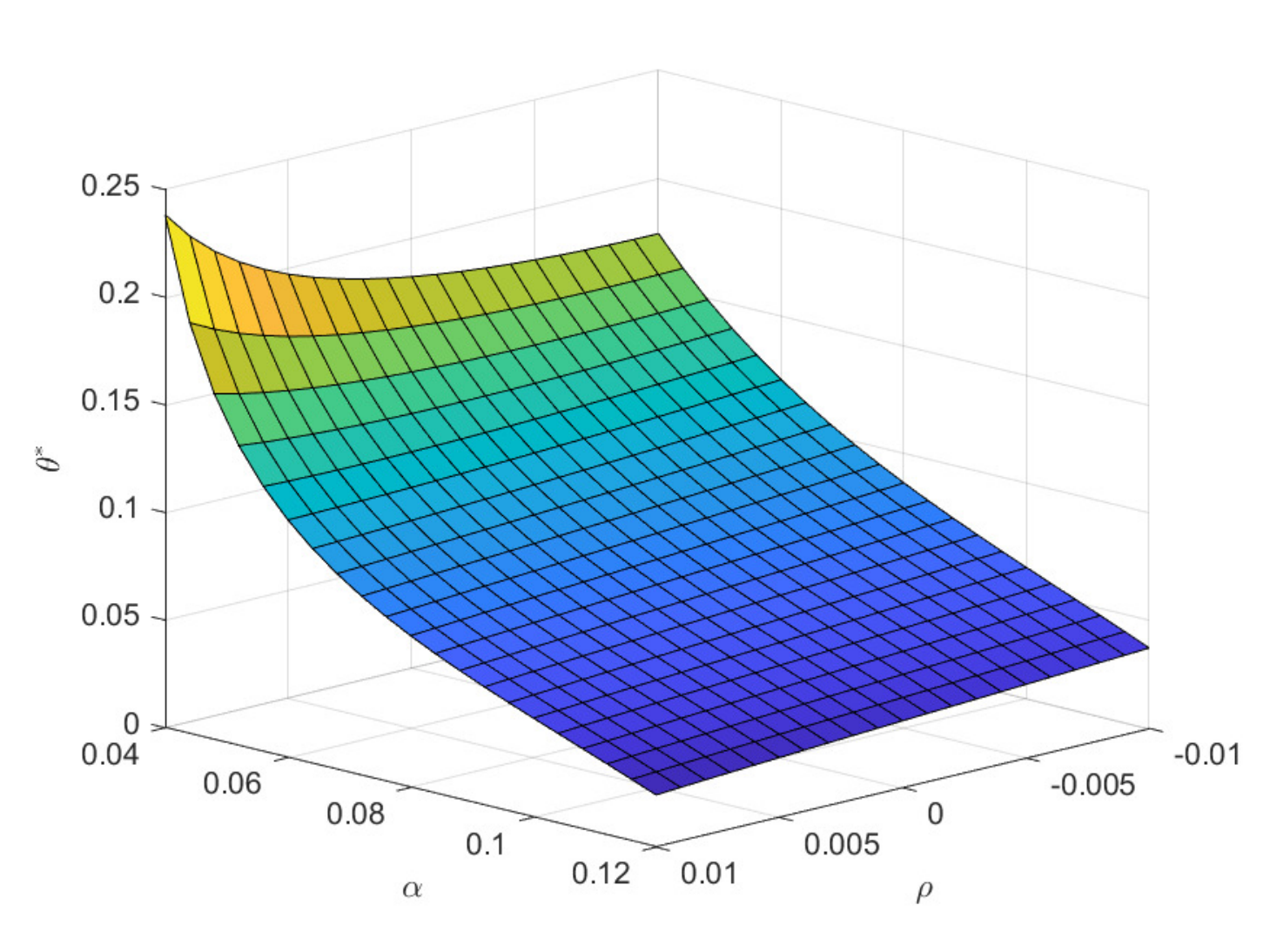}
\end{subfigure}
\begin{subfigure}{0.45\textwidth}
    \centering
    \includegraphics[width=\textwidth]{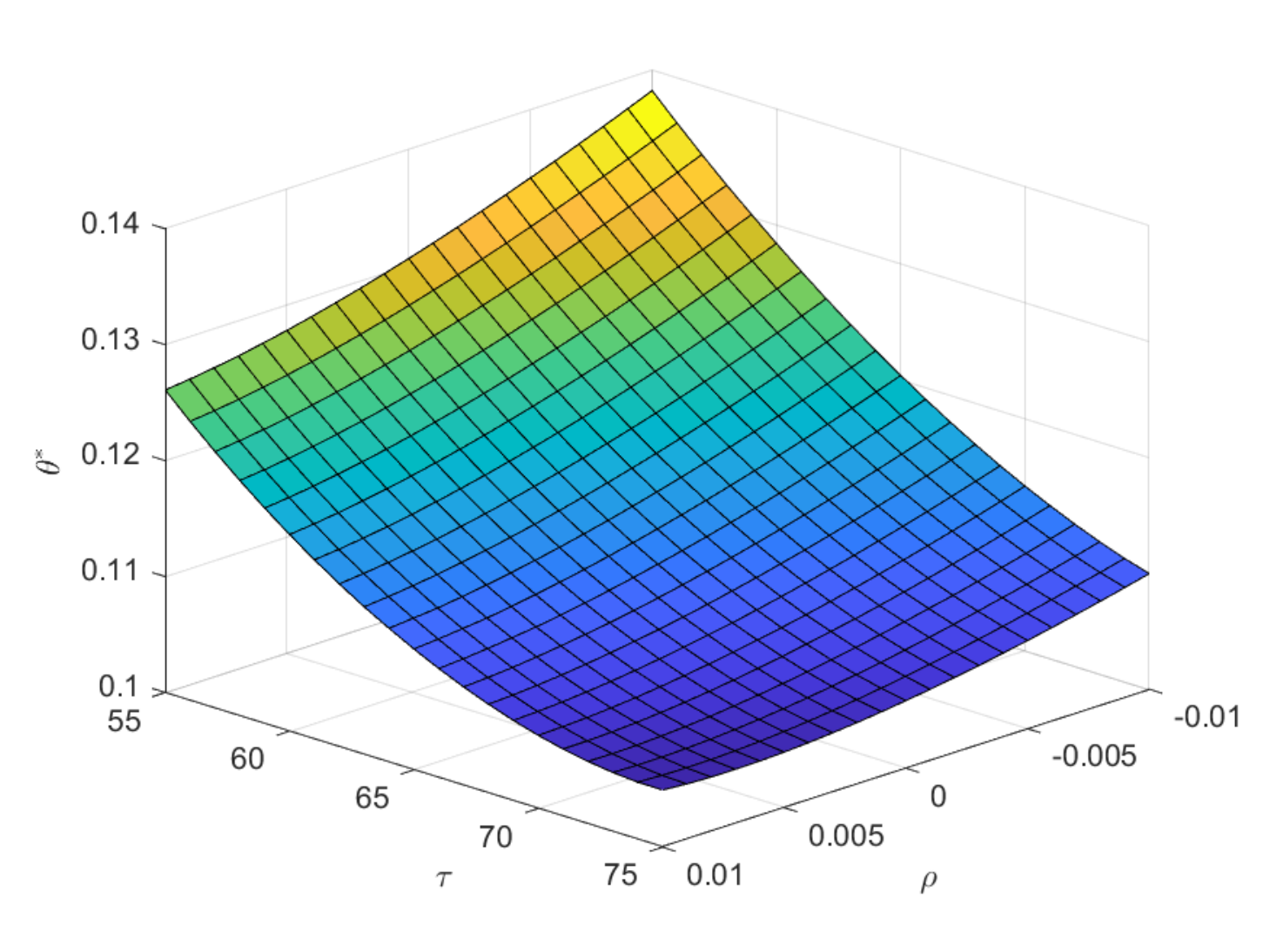}
\end{subfigure}
    \caption{Impacts of $\rho$ and $\Delta$, $\rho$ and $\gamma$, $\rho$ and $\alpha$, $\rho$ and $\tau$ on $\theta^*$}
    \label{fig10}
\end{figure}
\begin{figure}
    \centering
    \begin{subfigure}{0.45\textwidth}
        \centering
        \includegraphics[width=0.9\textwidth]{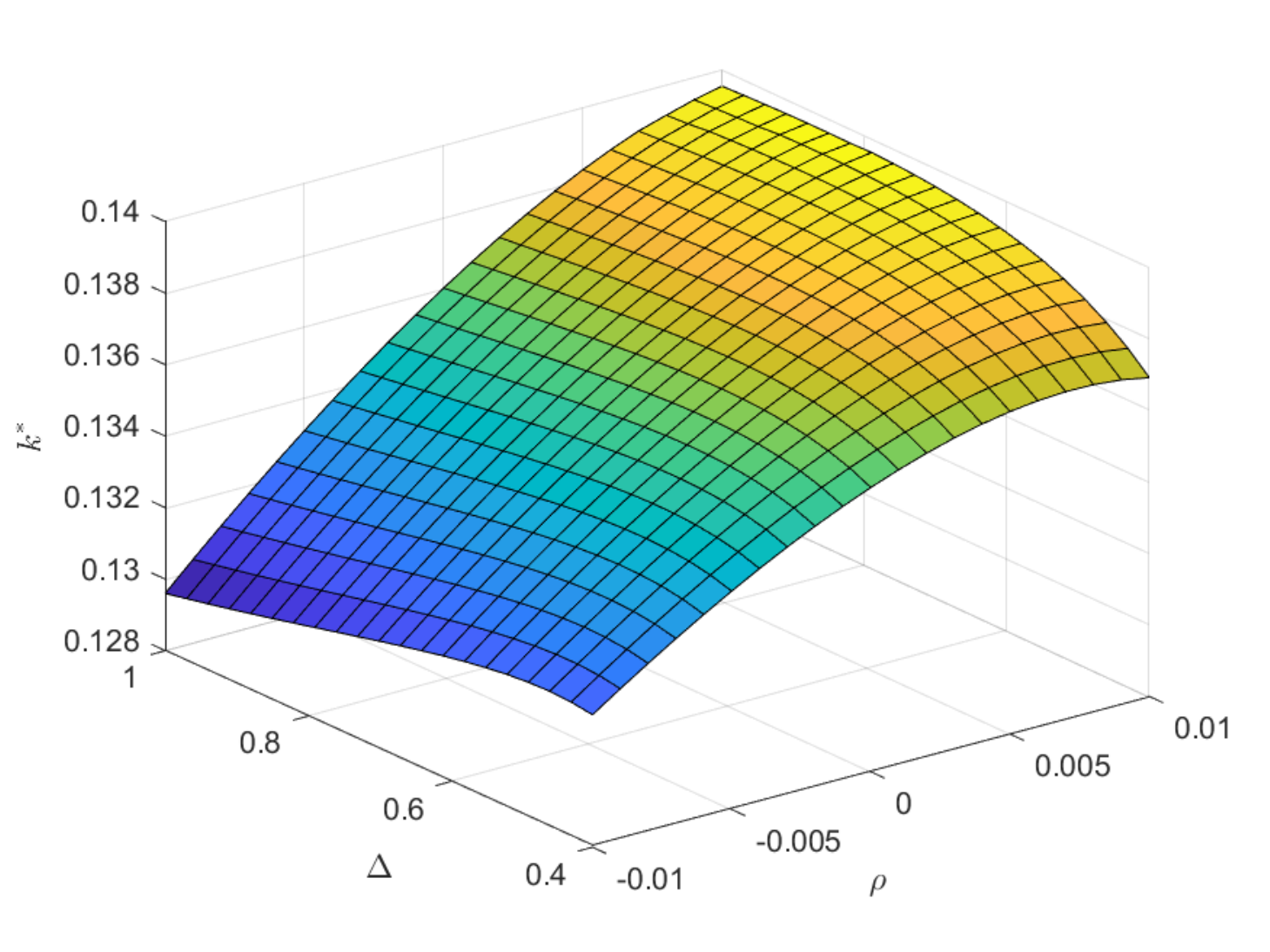}
    \end{subfigure}
    \begin{subfigure}{0.45\textwidth}
        \centering
        \includegraphics[width=0.9\textwidth]{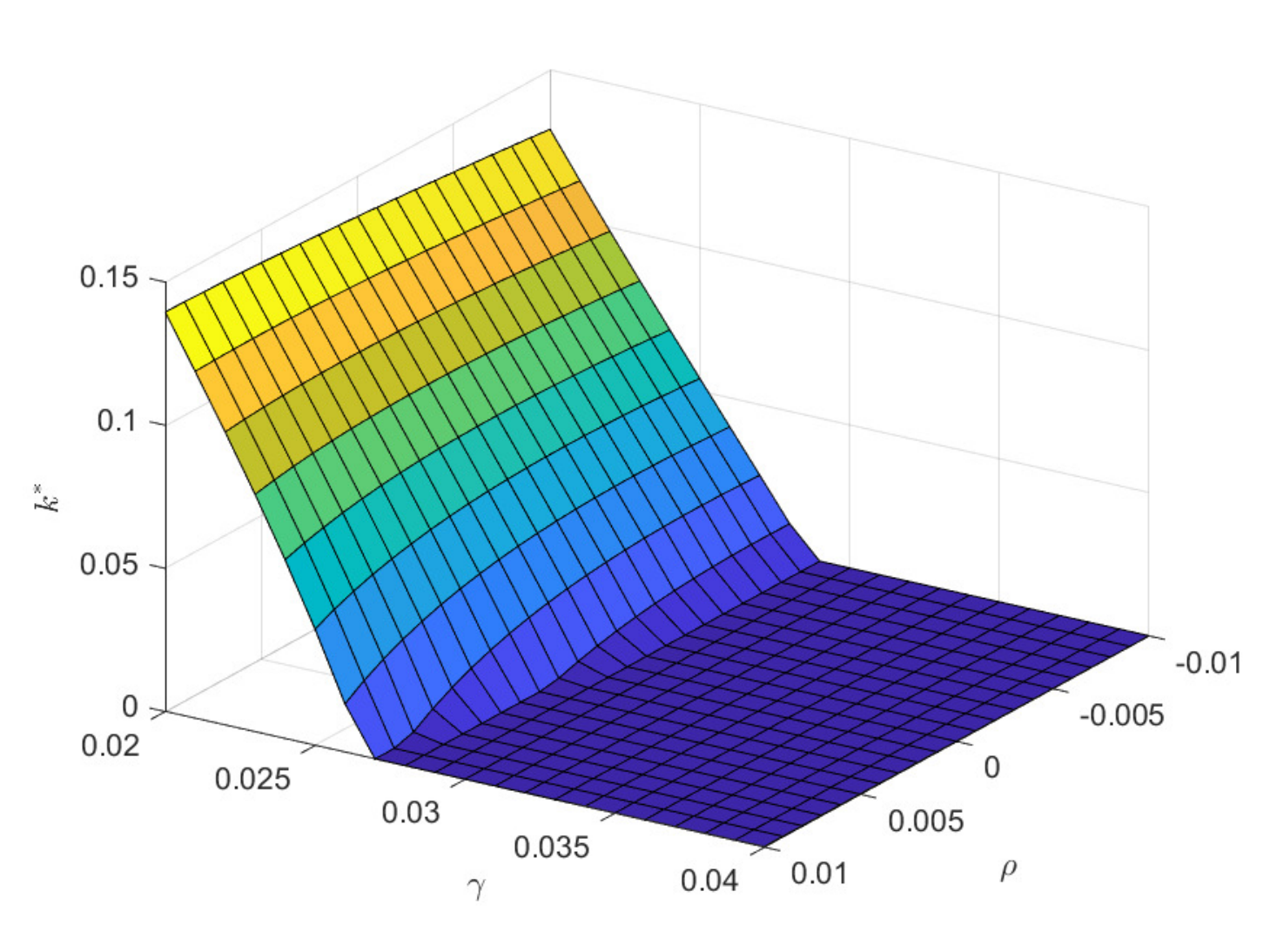}
    \end{subfigure}
    \begin{subfigure}{0.45\textwidth}
        \centering
        \includegraphics[width=0.9\textwidth]{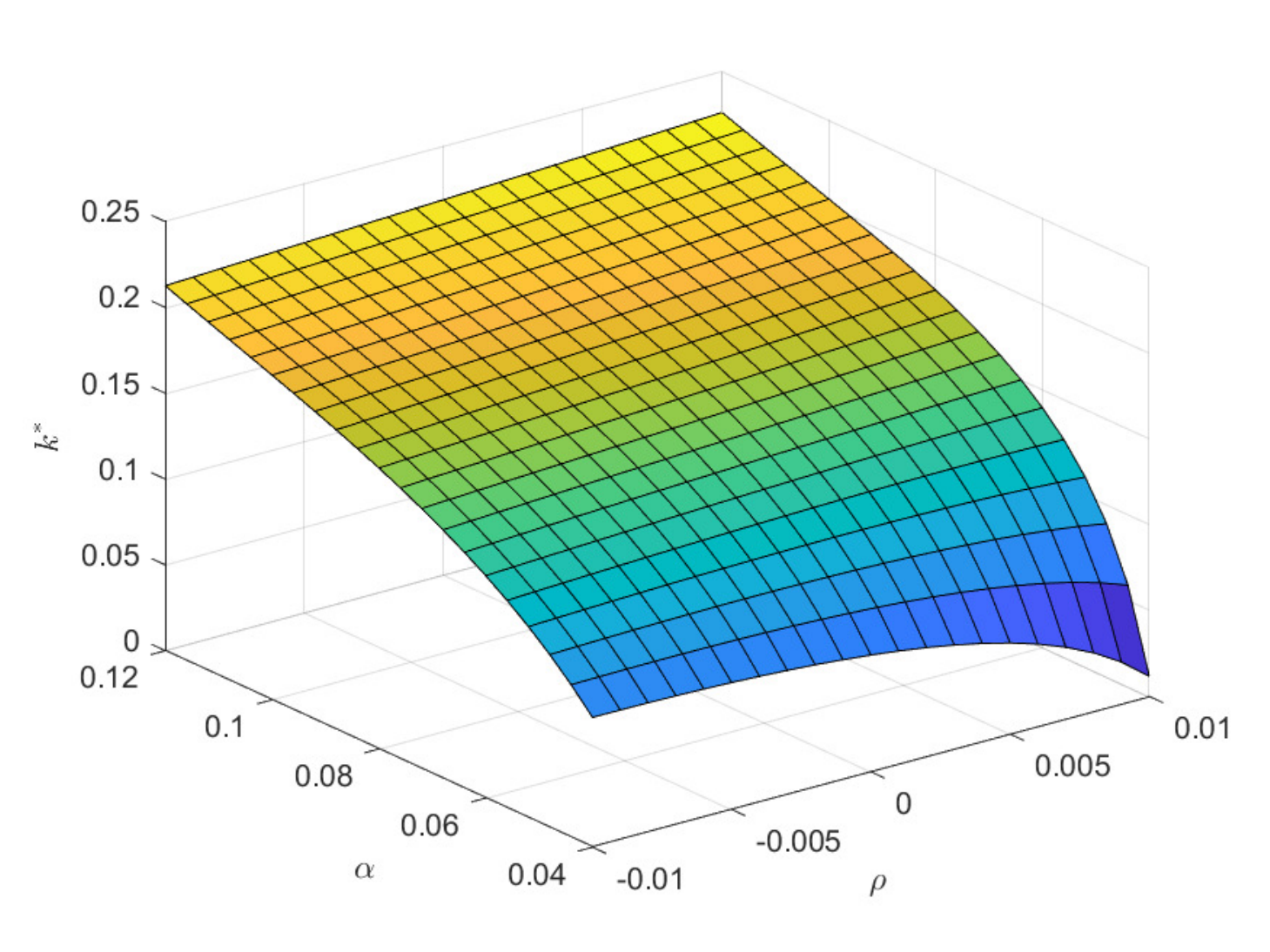}
    \end{subfigure}
    \begin{subfigure}{0.45\textwidth}
        \centering
        \includegraphics[width=0.9\textwidth]{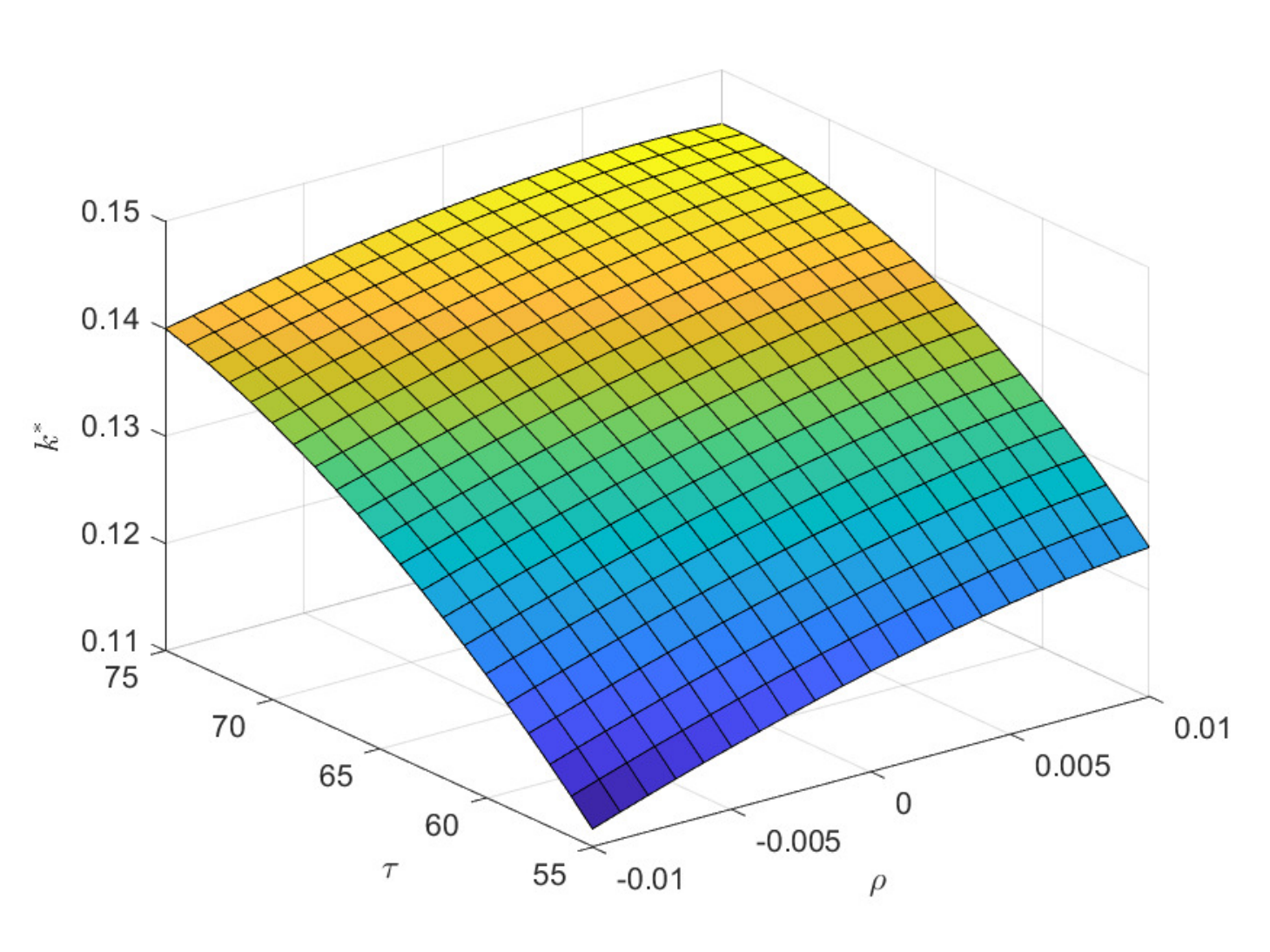}
    \end{subfigure}
    \caption{Impacts of $\rho$ and $\Delta$, $\rho$ and $\gamma$, $\rho$ and $\alpha$, $\rho$ and $\tau$ on $k^*$}
    \label{fig11}
\end{figure}
\vskip 5pt
In Fig. \ref{fig10} and Fig. \ref{fig11}, we study the impacts of the exogenous
parameters on the optimal contribution rates $\theta^{*}$ and $k^*$.
Particularly, the population growth rate $\rho$ plays two important
roles in deciding $\theta^{*}$ and $k^*$. One is that it affects the
preference orderings of different cohorts among the three pensions.
The other is that it changes the weight of the government
decision-making by changing the population of different cohorts.
\vskip 5pt
In the first case of Fig. \ref{fig10}, lower population growth rate $\rho$ decreases the attractiveness of PAYGO pension. Meanwhile, it increases the weight of the older cohorts in government's decision-making. And the latter effect dominates the former one. Thus, the optimal PAYGO contribution rate increases with respect to the decrease of the population growth rate. Moreover, higher longevity risk (smaller $\Delta$) also decreases the attractiveness of PAYGO pension. Particularly, we suppose that the annuity management company is fully aware of the longevity risk and revise the benefit rules of EET pension accordingly. Thus, higher longevity risk also decreases the attractiveness of EET pension. The results show that in the scenario of lower population growth rate, there is larger older population. In this circumstance,
the increase of longevity risk has a greater impact in reducing the attractiveness of PAYGO pension. Thus, we observe a slightly decrease in the optimal PAYGO contribution rate.  The opposite is valid in the scenario of higher population growth rate.

In the second case of Fig. \ref{fig10}, larger salary growth rate $\gamma$ increases the attractiveness of PAYGO pension
dramatically and results in a large area of $\theta^*=25\%$. This is
why PAYGO pension has been working effectively in many developing
countries over the past decades. In the third case of Fig. \ref{fig10},
larger EET return $\alpha$ decreases the attractiveness of PAYGO
pension. Thus, in the extreme case of abnormal high EET return, the
optimal PAYGO contribution rate $\theta^*$ decays to zero. This is
the case in most of the developed countries with professional
management experience of the EET fund. However, in the scenario
combined with low EET return and low population growth rate,
$\theta^*$ is relatively high because of that the older cohorts play
a decisive role in the government decision-making. In the last case
of Fig. \ref{fig10}, we study the impacts of $\tau$ and $\rho$ on $\theta^*$.
Interestingly, postponing retirement decreases the optimal PAYGO contribution rate.
Although postponing retirement leads to higher PAYGO benefit,
the longer contribution period weakens this beneficial effect. Besides, postponing retirement leads to a longer accumulation period of EET pension.
Because of the superior return efficiency of EET pension, this enhances the attractiveness of EET.
%Thus, we observe a slightly decrease in optimal PAYGO contribution rate and a dramatic increase in optimal EET
%contribution rate (cf. the last case in
%Fig. \ref{fig11} ).
\vskip5pt
In the four cases of Fig. \ref{fig11}, we observe the opposite trend
of $k^*$ with respect to the trend of $\theta^*$ in Fig. \ref{fig10}. Because
the two pensions are substitutes for each other, this opposite trend
is naturally valid. In all the scenarios, the
comparative efficiencies of the two obligatory pensions are superior, and $\theta^*+k^*=25\%$ holds.
%In the lase case of Fig. \ref{fig11}, postponing
%retirement increases the attractiveness of EET pension and $k^*$
%rises accordingly. Particularly, in the scenario of extremely early
%retirement and low population growth rate, both of PAYGO and EET
%pensions are not preferable and thus a complete individual saving
%will be welfare improving. Thus, we observe $k^*=\theta^*=0$ in this
%circumstance.

\begin{figure}[h]
	\centering
	\begin{subfigure}{0.45\textwidth}
		\centering
		\includegraphics[width=0.9\textwidth]{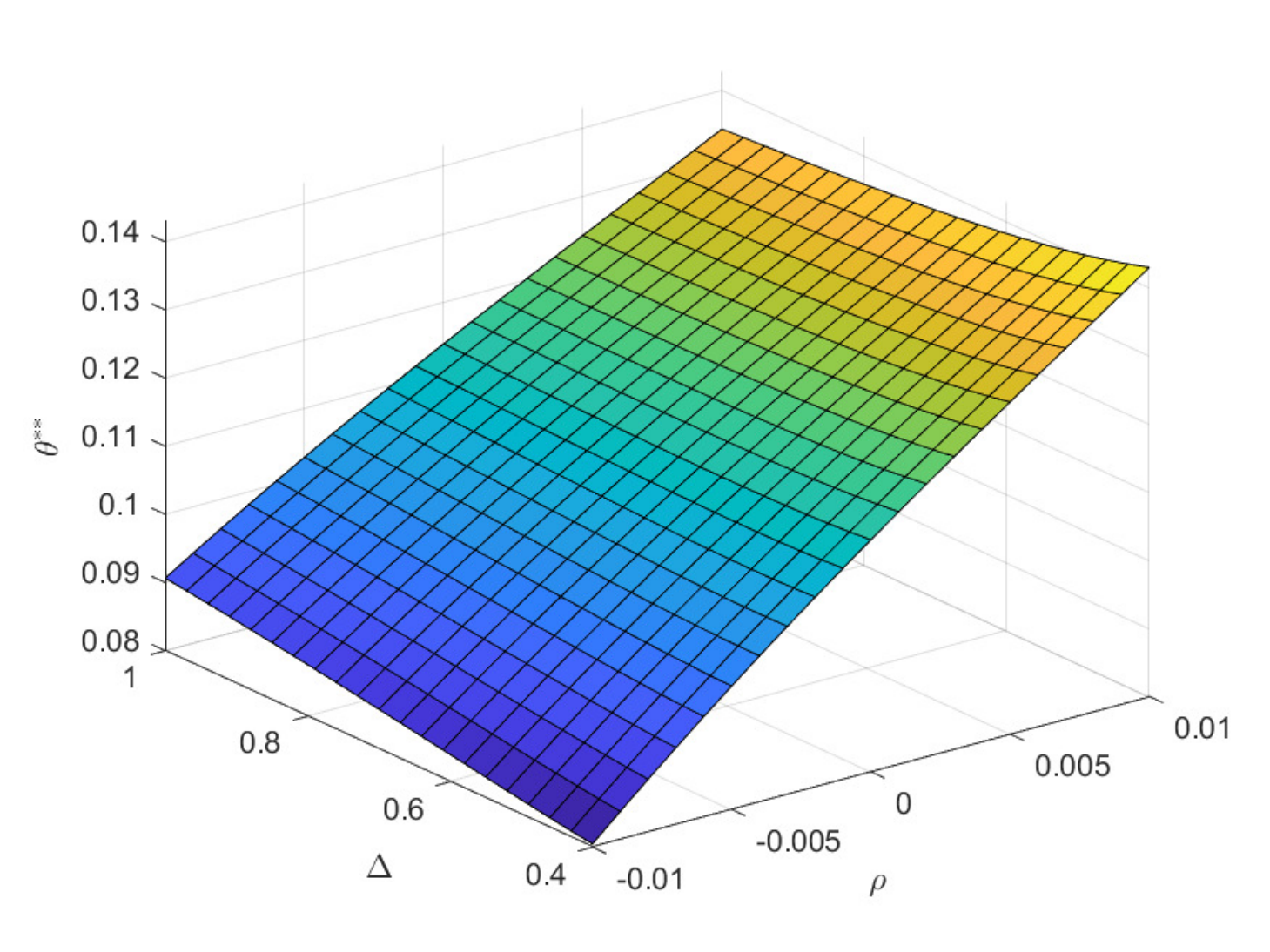}
	\end{subfigure}
	\begin{subfigure}{0.45\textwidth}
		\centering
		\includegraphics[width=0.9\textwidth]{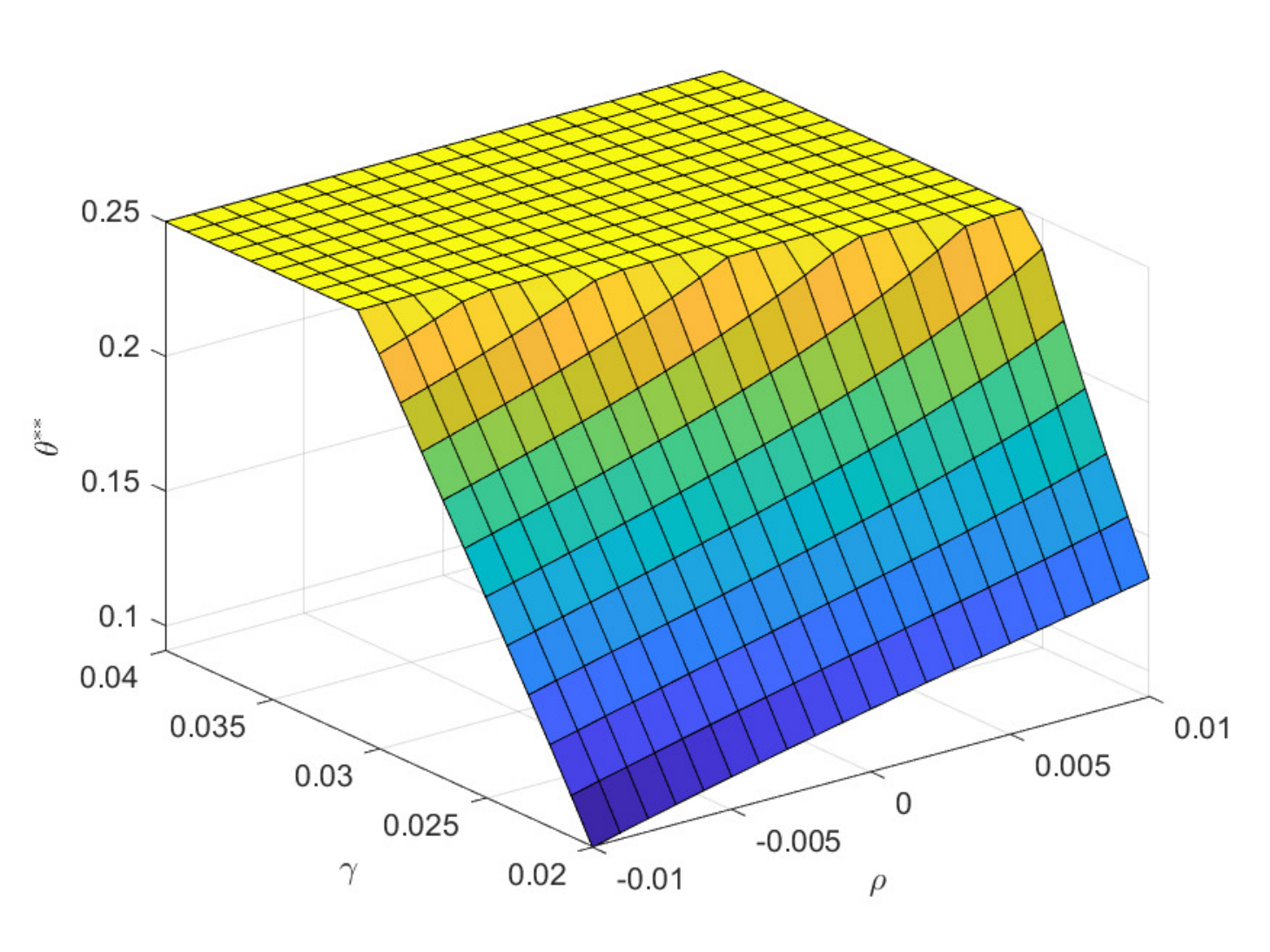}
	\end{subfigure}
	\begin{subfigure}{0.45\textwidth}
		\centering
		\includegraphics[width=0.9\textwidth]{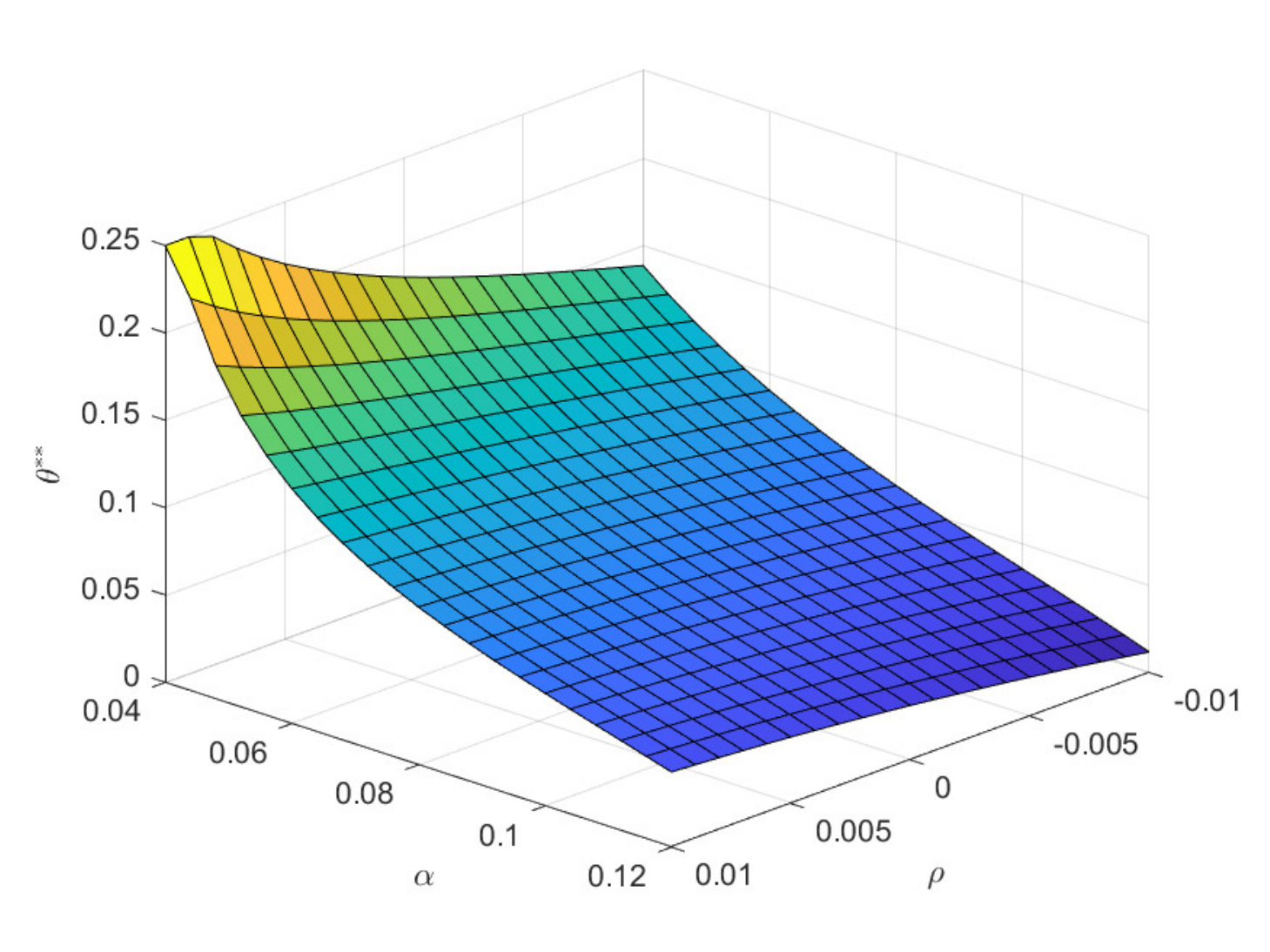}
	\end{subfigure}
	\begin{subfigure}{0.45\textwidth}
		\centering
		\includegraphics[width=\textwidth]{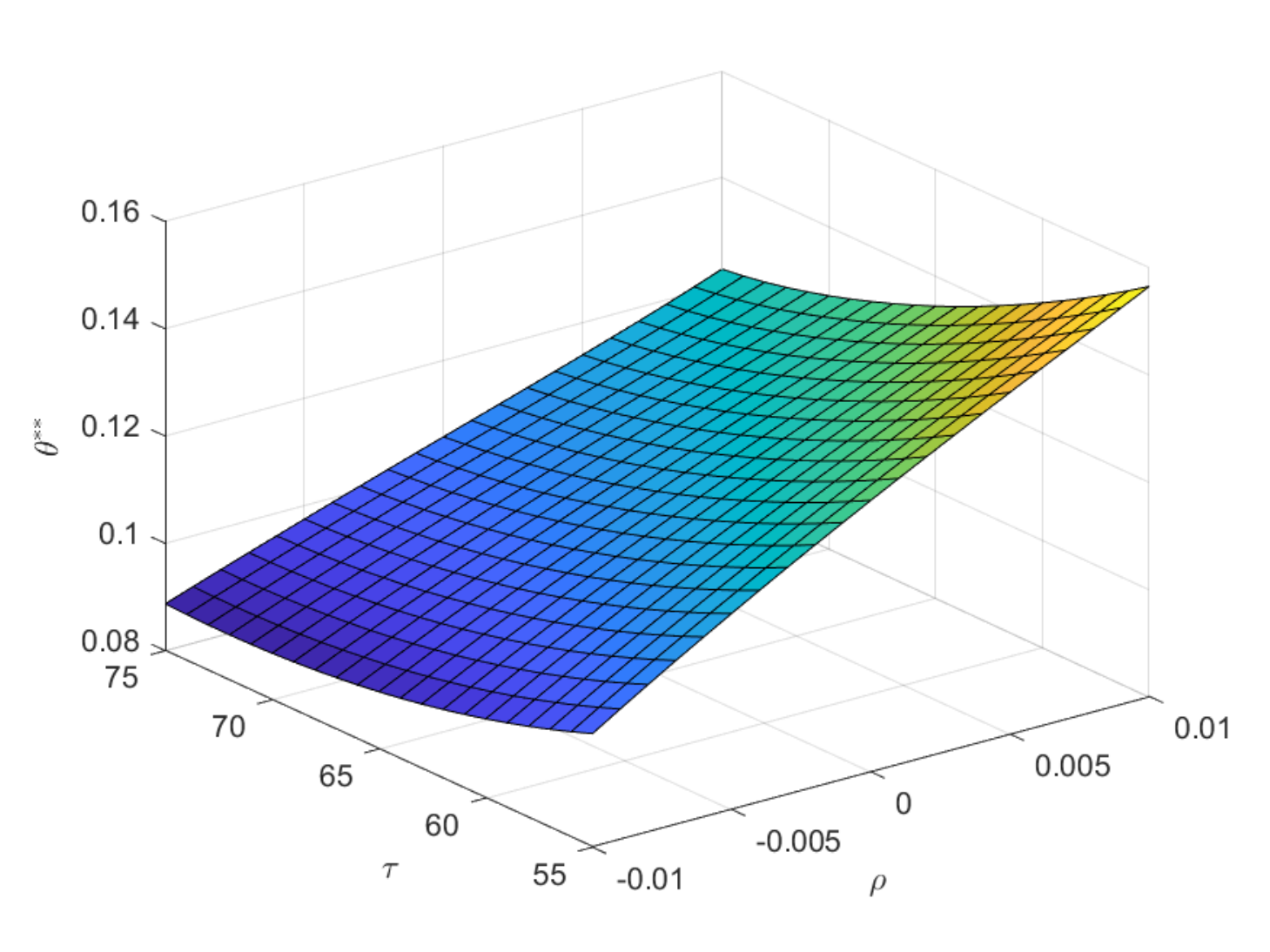}
	\end{subfigure}
	\caption{Impacts of $\rho$ and $\Delta$, $\rho$ and $\gamma$, $\rho$ and $\alpha$, $\rho$ and $\tau$ on $\theta^{**}$}
	\label{figf10}
\end{figure}

\begin{figure}
	\centering
	\begin{subfigure}{0.45\textwidth}
		\centering
		\includegraphics[width=0.9\textwidth]{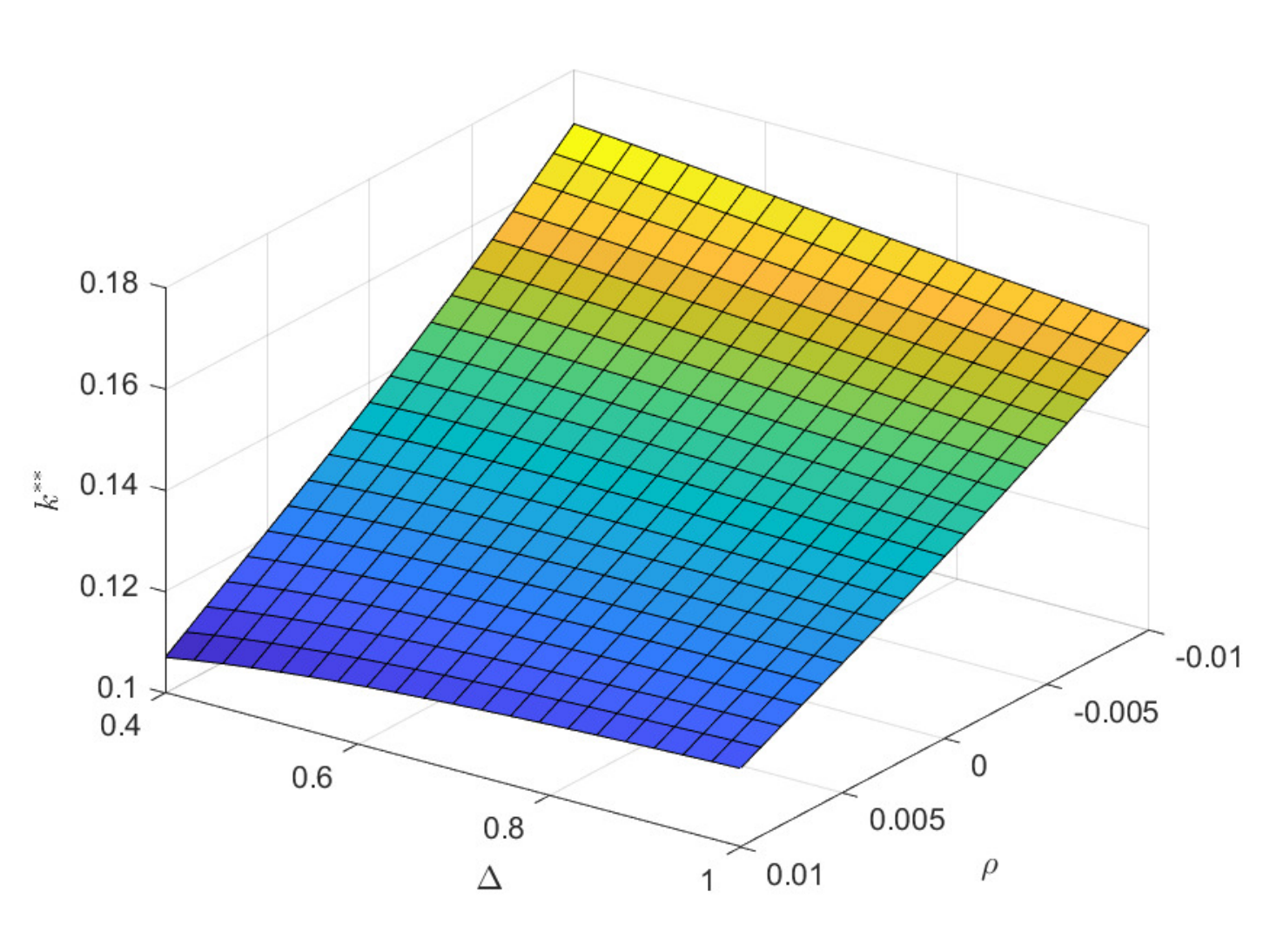}
	\end{subfigure}
	\begin{subfigure}{0.45\textwidth}
		\centering
		\includegraphics[width=0.9\textwidth]{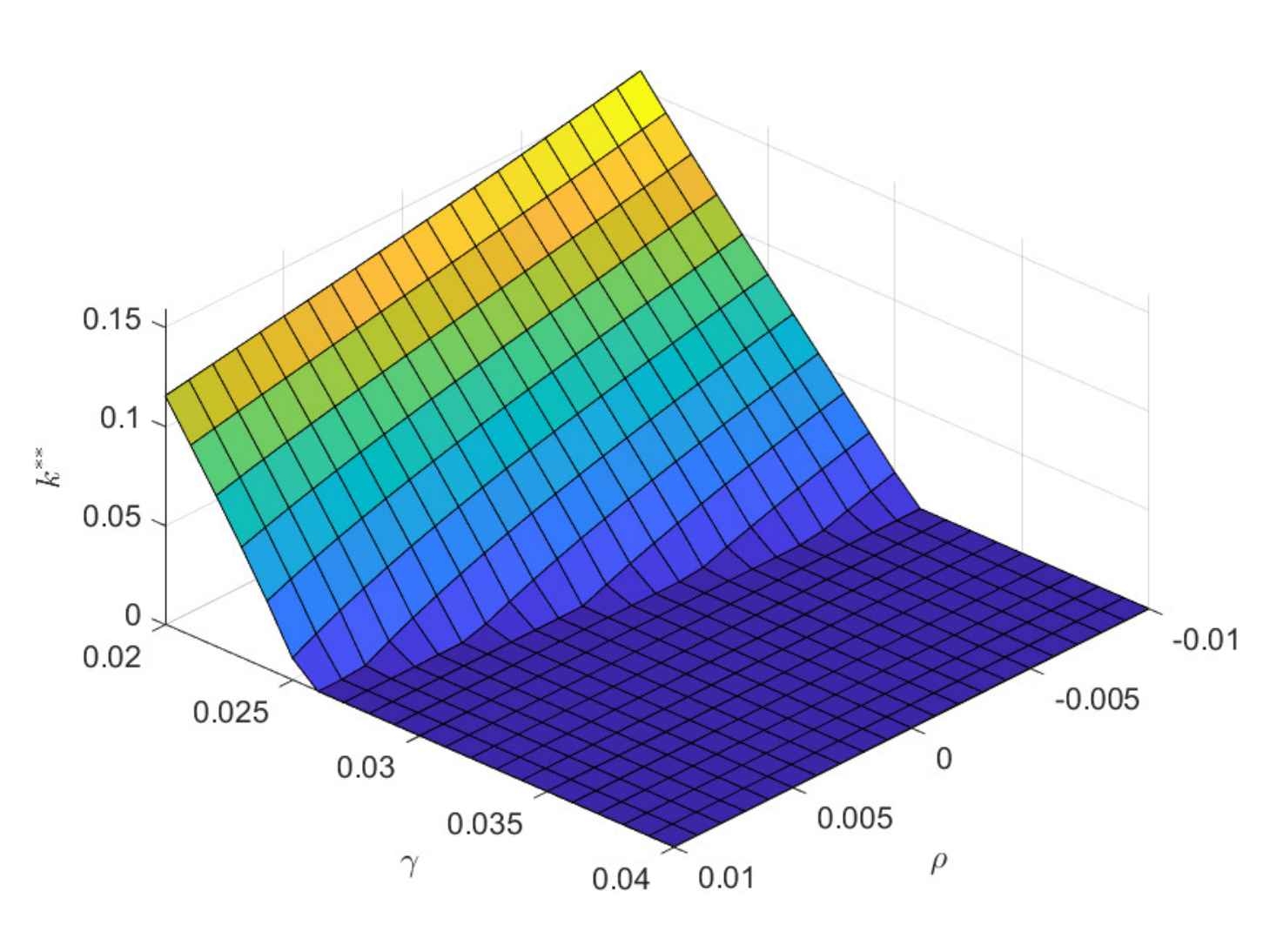}
	\end{subfigure}
	\begin{subfigure}{0.45\textwidth}
		\centering
		\includegraphics[width=0.9\textwidth]{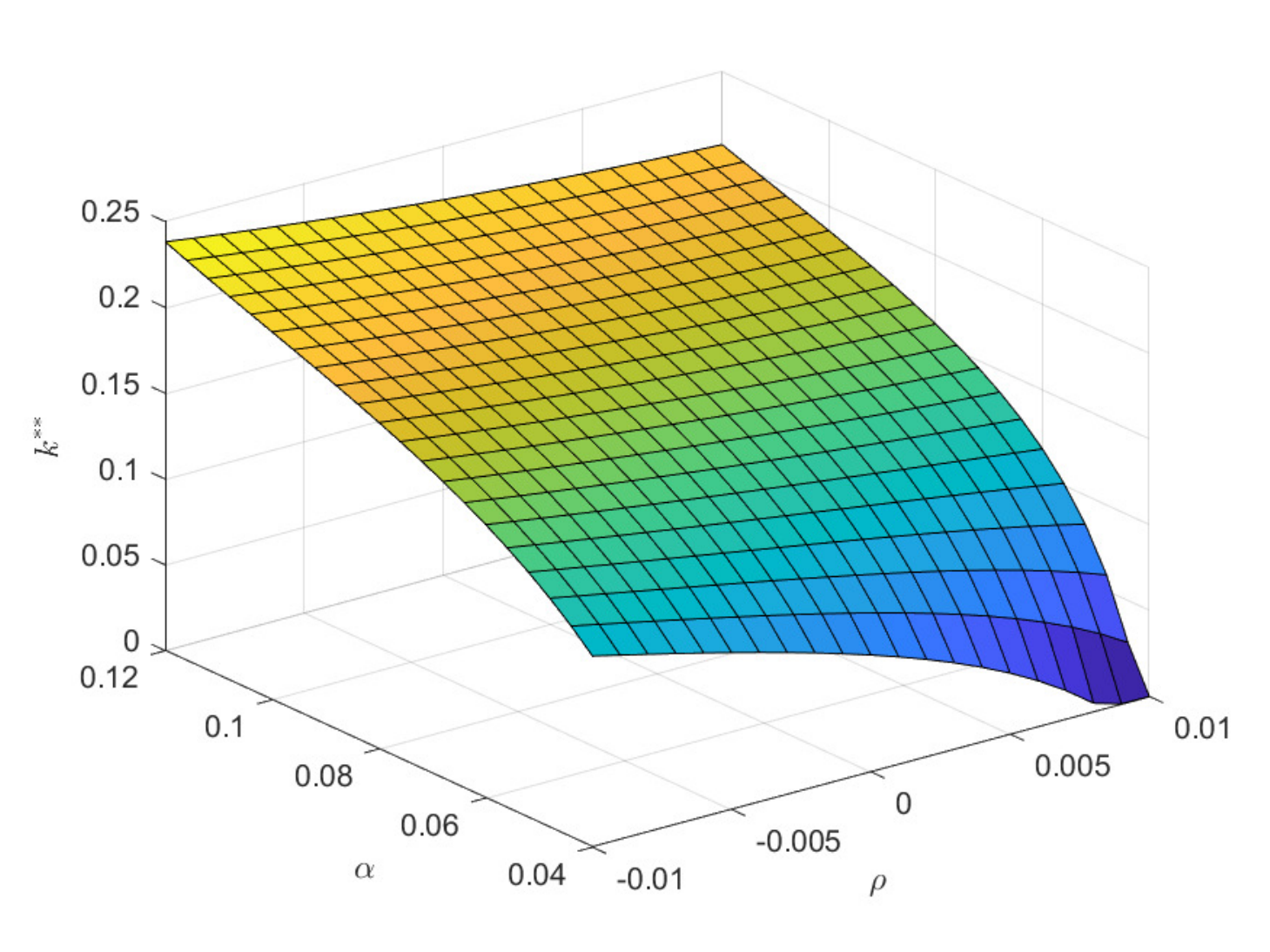}
	\end{subfigure}
	\begin{subfigure}{0.45\textwidth}
		\centering
		\includegraphics[width=0.9\textwidth]{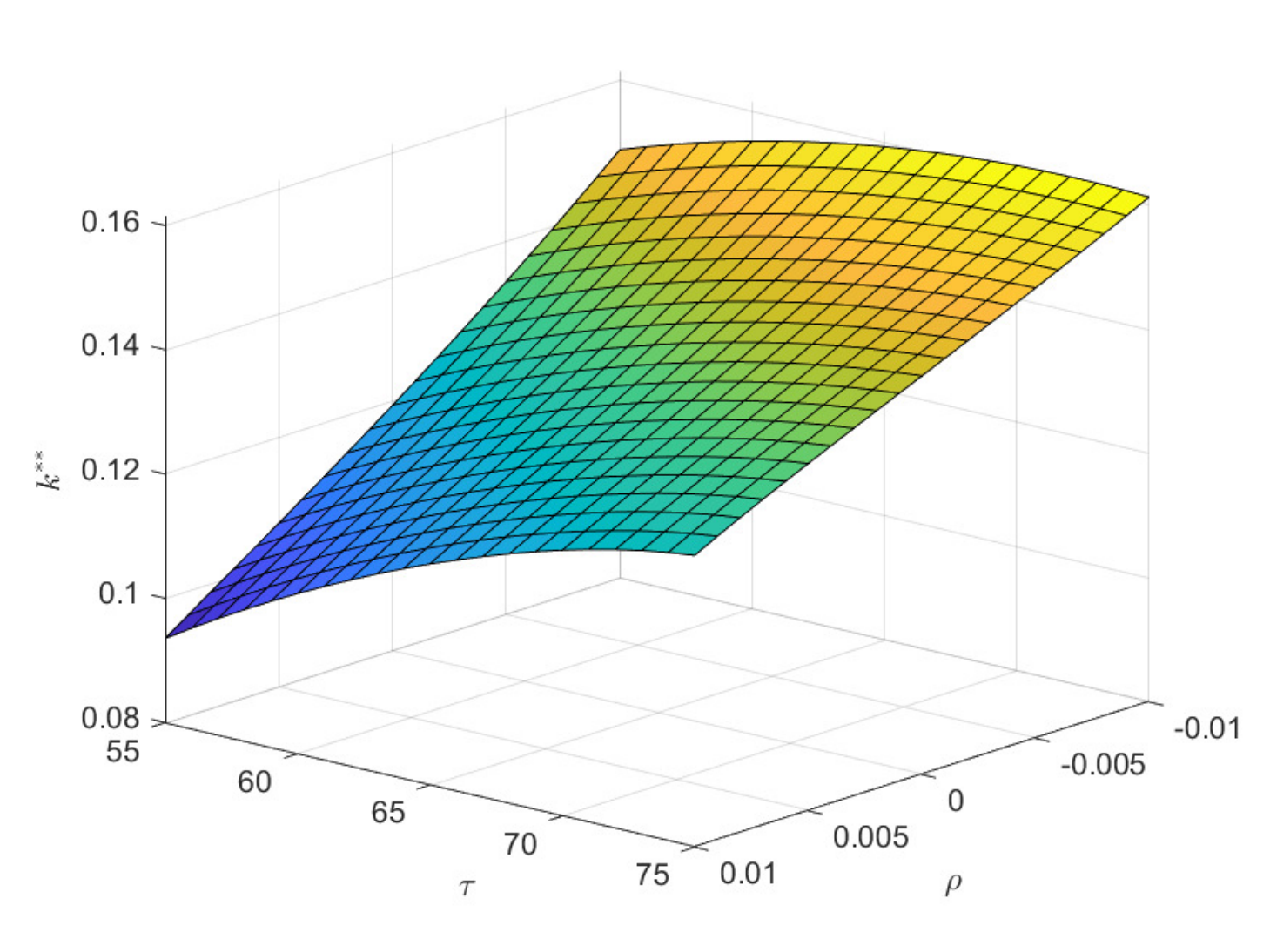}
	\end{subfigure}
	\caption{Impacts of $\rho$ and $\Delta$, $\rho$ and $\gamma$, $\rho$ and $\alpha$, $\rho$ and $\tau$ on $k^{**}$}
	\label{figf11}
\end{figure}
Particularly, we exhibit the results of $\theta^{**}$ and $k^{**}$ under the equal weighted objective in Fig. \ref{figf10} and Fig. \ref{figf11}.
We observed that the impacts of population growth rate $\rho$ on $\theta^{**}$ and $k^{**}$  are %completely
opposite to the results in Fig. \ref{fig10} and Fig. \ref{fig11}. Under equal weighted objective, lower population growth rate only
decreases the attractiveness of PAYGO pension. Thus, the optimal PAYGO contribution rate decreases and the optimal EET contribution
rate increases accordingly. As such, how the government balances the heterogeneous preference orderings of different cohorts has a decisive impact on the choice of the optimal contribution rates. Therefore, the government needs to carefully consider the issue of intergenerational  differences and fairness, and choose reasonable weight parameters.
\vskip 25pt
\setcounter{equation}{0}
\section{{\bf{Conclusions}}}\label{s7}
In this paper, we study the optimal mix among PAYGO, EET and individual savings. We establish the stochastic differential equations
to depict their heterogeneous characteristics. Particularly, EET pension is included, which has special properties in preferential taxation, return efficiency and longevity risk sharing. According to the participants' value function, each cohort has an exclusive age-dependent preference ordering among the
pensions, which is determined by the comparative efficiencies of the pensions. Accordingly, we establish three critical ages as the
boundaries to distinguish the preference, and eventually obtain the multiple outcomes of the
preference ordering based on the heterogeneous characteristic parameters.  The age-dependent preference orderings are the extension of the ``Samuelson-Aaron'' criterion.
\vskip 5pt
The optimal PAYGO and EET contribution rates are obtained by solving a Nash equilibrium between the participants and the government. Particularly, the objective of the government is to maximize the overall utility of all the participants weighted by the population of each cohort. Naturally, the negative population growth rate reduces the
attractiveness of PAYGO pension. However, it also leads to the rise
of the older cohorts' weight parameter in the government
decision-making. As such, the optimal mix is the comprehensive result of the above two effects. In the countries that suffer from shrinking population and aging problem, the U.S. can further increase the share of PAYGO pension to better improve the welfare of the older cohorts, while China can add a small obligatory share of EET pension to increase the overall utility of the participants.
\vskip 20pt
{\bf Acknowledgements.} The authors acknowledge the support from the National Natural Science Foundation of China (Nos. 12271290, 11871036). The authors are particularly  grateful to the two anonymous referees and the Editor whose suggestions greatly improve  the manuscript's quality. The authors also thank the members of the group of Actuarial Sciences and Mathematical Finance  at the Department of Mathematical Sciences, Tsinghua University for their feedbacks and useful conversations.
\vskip 15pt
%{\bf Acknowledgements.}
%The authors acknowledge the support from the National Natural Science Foundation of China (Grant No.11871036, and No.11471183). The authors also thank the members of the group of Actuarial Sciences and Mathematical Finance at the Department of Mathematical Sciences, Tsinghua University for their feedbacks and useful conversations. The authors gratefully appreciate
%Kristoffer Lindensj from Stockholm University and Guohui Guan from Renmin University of China for their useful discussions and suggestions.

\appendix
\renewcommand\thefigure{\Alph{section}\arabic{figure}}
\renewcommand{\theequation}{\thesection .\arabic{equation}}
\section{Proof of Proposition \ref{objective1}}\label{pp4.1}
	Using Eq.(\ref{equ1}) and Eq.(\ref{value}), we have
	\begin{eqnarray*}
		&&\mathrm{E}[V(z,0,W(z),0;z,\theta,k)|W(t_0)=w_0]\\
		&&=\mathrm{E}\left[\frac{1}{\delta}L(z;z)(M_1(z;z)\theta+M_2(z;z)k+M_3(z;z))^{\delta}W^{\delta}(z)|W(t_0)=w_0\right]\\
		&&=\frac{1}{\delta}L(z;z)(M_1(z;z)\theta+M_2(z;z)k+M_3(z;z))^{\delta}\mathrm{E}\left[W^{\delta}(z)|W(t_0)=w_0\right]\\
		&&=\frac{1}{\delta}L_0(M_{01}\theta+M_{02}k+M_{03})^{\delta}w_0^{\delta}e^{\delta[\gamma+\frac{1}{2}(\delta-1)\xi^2](z-t_0)}.
	\end{eqnarray*}
	%where $W(t)=w_0e^{(\gamma-\frac{1}{2}\xi^2)(t-t_0)+\xi(B(t)-B(t_0))}$.
	%Thus, the deterministic objective function of the government is given by
	%\begin{align}
	%	\phi(\theta,k)=&\int_{t_0-\omega+a}^{t_0}\frac{1}{\delta}n(z)L(t_0;z)\{x_0(z)+[M_1(t_0;z)\theta+M_2(t_0;z)k+M_3(t_0;z)]w_0\nonumber\\
	%	&+N(t_0;z)y_0(z)\}^{\delta}dz
	%	+\int_{t_0}^{\infty}\frac{1}{\delta}n(z)L(z;z)w_0^{\delta}e^{\left[\delta \gamma+\frac{1}{2}\delta(\delta-1)\xi^2-r\right](z-t_0)}\nonumber\\
	%	&[M_1(z;z)\theta+M_2(z;z)k
	%	+M_3(z;z)]^{\delta}dz.\label{equ18}
	%\end{align}
	 Thus, the total utility of the participants who will join the pension after $t_0$ is
	\begin{align*}
		&\int_{t_0}^{\infty}e^{-r(z-t_0)}n(z)\mathrm{E}[V(z,0,W(z),0;z,\theta,k)|W(t_0)=w_0]dz\\
		=&\int_{t_0}^{\infty}\frac{1}{\delta}e^{-r(z-t_0)}n_0e^{\rho z}L_0(M_{01}\theta+M_{02}k+M_{03})^{\delta}w_0^{\delta}e^{\delta\left[ \gamma+\frac{1}{2}(\delta-1)\xi^2\right](z-t_0)}dz\\
		=&\frac{1}{\delta}n_0L_0(M_{01}\theta+M_{02}k+M_{03})^{\delta}w_0^{\delta}e^{\rho t_0}\int_{0}^{\infty}e^{\left[\rho-r+\delta( \gamma+\frac{1}{2}(\delta-1)\xi^2)\right]u}du\\
		=&\left\{
		\begin{array}{lr}
			\infty,~~&\rho+\delta\left[ \gamma+\frac{1}{2}(\delta-1)\xi^2\right]\geq r,\\
			\frac{n_0L_0e^{\rho t_0}w_0^{\delta}}{\delta\left\{r-\rho-\delta\left[ \gamma+\frac{1}{2}(\delta-1)\xi^2\right]\right\}}(M_{01}\theta+M_{02}k+M_{03})^{\delta},~~&\rho+\delta\left[ \gamma+\frac{1}{2}(\delta-1)\xi^2\right]<r.
		\end{array}
		\right.
	\end{align*}
	In order to guarantee that the overall utility will not explode, the optimization problem of the government is well-defined if and only if $\rho+\delta\left[\gamma+\frac{1}{2}(\delta-1)\xi^2\right]<r$ holds. Taking $\delta=\delta_0$, we obtain $\phi_1(\theta,k)$ expressed in Eq.\eqref{utility1}. Similarly, $\phi_2(\theta,k)$ is expressed as Eq.\eqref{utility11}.
\vskip 5pt	
Eventually, we analyze the admissible scope of $\theta$ and $k$. Consider the participants' total equivalent disposable wealth:
\begin{eqnarray*}
		I(z,\theta,k;\theta_0,k_0)&=x_0(z)+[M_1(t_0;z)\theta\!+\! M_2(t_0;z)k \!+\! M_3(t_0;z)]w_0+N(t_0;z)y_0(z),
\end{eqnarray*}
which is supposed to be nonnegative for arbitrary $z\in [t_0+a-\omega, t_0]$. Combining the constraints on the sum of $\theta$ and $k$, we establish the inequalities in Eq.\eqref{equ19} that $\theta$ and $k$ should satisfy.
\vskip 10pt
\section{Estimation of private information $x_{0}(z)$ and $y_{0}(z)$}\label{A.1}
We estimate $x_0(z)$ and $y_0(z)$ in the expectation way. First, according to Eq.\eqref{equ17}, for $z\in[t_0-\omega+a,t_0]$,
$$x_0(z)=\mathrm{E}\{X^*(t_0;z,\theta_0,k_0)\},\ \ y_0(z)=\mathrm{E}\{Y(t_0;z,k_0)\}.$$
%Based on Eq.\eqref{equ1}, $W(t)=W_0e^{(\gamma-\frac{1}{2}\xi^2)t+\xi B(t)}$.\\ %Thus, $w_0=\mathrm{E}\{W(t_0)\}=W_0e^{\gamma t_0}.$
Second, solving  SDEs \eqref{equ1}, (\ref{equ2}) and taking expectation, we have
%\begin{eqnarray*}
%    Y(t;z,k)=\left\{
%    \begin{array}{lr}
%        \int_{z}^{t}e^{\left(\alpha-\frac{1}{2}
%        \beta^2\right)(t-s)+\beta(B(t)-B(s))}k_0W(s)ds,~&t\in[z,z+\tau-a],\\
%        \frac{(z+\omega-a)-t}{\omega-\tau}Y(z+\tau-a),~&t\in(z+\tau-a,z+\omega-a].
%    \end{array}
%    \right.
%\end{eqnarray*}
%Therefore
\begin{equation*}
	y_0(z)%=\mathrm{E}[Y(t_0;z,k_0)]
	=\left\{
	\begin{array}{lr}
		k_0W_0e^{\alpha t_0}\int_{z}^{t_0}e^{(\gamma-\alpha)s}ds,\phantom{eeee}z\in[t_0-\tau+a,t_0],\\
		\mathrm{E}[Y(z+\tau-a;z,k_0)],
		\phantom{eeee}z\in[t_0-\omega+a,t_0-\tau+a).
	\end{array}
	\right.
\end{equation*}
Last, similar to \cite{HLSY2021}, we use  martingale method to calculate $x_0(z)$. We directly present the final estimation here,
\begin{align*}
	x_0(z)=&-M(t_0;z,\theta_0,k_0)W_0e^{\gamma t_0}-N(t_0;z)y_0(z)\\
	&+\left[\frac{L(t_0;z)}{L(z;z)}\right]^{\frac{1}{1-\delta}}M(z;z,\theta_0,k_0))W_0e^{\gamma t_0}e^{\left[-\gamma+\frac{r}{1-\delta}+\frac{2-\delta}{2(1-\delta)^2}\nu^2\right](t_0-z)}.
\end{align*}
\vskip 10pt
\section{Proof of Proposition \ref{k}}\label{pp5.1}
	We consider the admissible scope of $k(\theta,z)$ for the participants of all ages $\zeta$ given the feasible $\theta$. Taking the fixed $z$ $(z=a+t_0-\zeta)$ in Eq.\eqref{equ19}, we have
	\begin{align}
		k(\theta;z)\in&\left\{
		\begin{array}{lr}
			\left[-\frac{x_0(z)+N(t_0;z)y_0(z)+[M_1(t_0;z)\theta+M_3(t_0;z)]w_0}{M_2(t_0;z)w_0}\vee0,m-\theta\right],~~&M_2(t_0,z)>0,\\
			\left[0,m-\theta\right],&M_2(t_0,z)=0,\\
			\left[0,-\frac{x_0(z)+N(t_0;z)y_0(z)+[M_1(t_0;z)\theta+M_3(t_0;z)]w_0}{M_2(t_0;z)w_0}\wedge(m-\theta)\right],&M_2(t_0,z)<0.
		\end{array}
		\right.
		\label{scope}
	\end{align}
Then, we divide this optimization problem into the following three cases:

 If $\zeta\leq a$, i.e., $z\geq t_0$, the optimization problem of the participants is
	\begin{align*}
		\max_{k} V(z,0,W(z),0;z,\theta,k)
		%&=\max_k \frac{1}{\delta}L(z;z)M^{\delta}(z;z,k,\theta)W^{\delta}(z)\\
		=\max_k\frac{1}{\delta}L_0(M_{01}\theta+M_{02}k+M_{03})^{\delta}W^{\delta}(z).
	\end{align*}
Taking $z=t_0$ in Eq.\eqref{scope} and  considering the sign of $M_{02}$, we have Eq.\eqref{k1}.
	%Taking $z=t_0$ in Eq.\eqref{scope}, we have
	%\begin{align*}
	%	k(\theta;z)\in&\left\{
	%	\begin{array}{lr}
	%		\left[-\frac{M_{01}\theta+M_{03}}{M_{02}}\vee0,m-\theta\right],~~&M_{02}>0,\\
	%		\left[0,m-\theta\right],&M_{02}=0,\\
	%		\left[0,-\frac{M_{01}\theta+M_{03}}{M_{02}}\wedge (m-\theta)\right],&M_{02}<0.
	%	\end{array}
	%	\right.
	%\end{align*}
	%Thus, the participants choose
	%\begin{align*}
	%	k^*(\theta;z)=k^*(\theta;t_0)=&\left\{
	%	\begin{array}{lr}
	%		m-\theta,~~&M_{02}>0,\\
	%		\left[0,m-\theta\right],&M_{02}=0,\\
	%		0,~~&M_{02}<0.
	%	\end{array}
	%	\right.
	%\end{align*}
	
	If $\zeta\geq a$, then the participants' optimization problem is
	\begin{align*}
		&\max_k V(t_0,x_0(z),w_0,y_0(z);z,\theta,k)\\
		=&\max_k \frac{1}{\delta}L(t_0;z)\left\{x_0(z)+\left[M_1(t_0;z)\theta+M_2(t_0;z)k+M_3(t_0;z)\right]w_0+N(t_0;z)y_0(z)\right\}^{\delta}.
	\end{align*}

	If $\zeta\geq\tau$, i.e., $z\leq t_0-\tau+a$, ${M_2}(t_0;z)=0$, the change of $k$ does not affect the utility of the participants and the admissible scope of $k(\theta,z)$ is $[0,m-\theta]$. Therefore, the participants arbitrarily choose $k$ between $0$ and  $m-\theta$.
	
	If $a\leq\zeta<\tau$, i.e., $t_0-\tau+a<z\leq t_0$,
	considering the admissible scope in Eq.\eqref{scope} and the sign of $M_2(t_0,z)$, we have Eq.\eqref{k2}.
	%\begin{align*}
	%	k^*(\theta;z)=&\left\{
	%	\begin{array}{lr}
	%		m-\theta,~~&M_2(t_0,z)>0,\\
	%		\left[0,m-\theta\right],&M_2(t_0,z)=0,\\
	%		0,~~&M_2(t_0,z)<0.
	%	\end{array}
	%	\right.
	%\end{align*}
\vskip 10pt
\section{Proof of Corollary \ref{objective2}}\label{pc5.1}
	The proof is based on Proposition \ref{objective1}. Being fully aware of the optimal feedback $k^*(\theta;z)$ of the participants and taking $k=k^*(\theta,z)$ in Eq.\eqref{utility1}, we derive the government's optimization objective functions $\widetilde{\phi}_{1}(\theta)\triangleq\phi_{1}(\theta,k^*(\theta;z))$ and $\widetilde{\phi}_{2}(\theta)\triangleq\phi_{2}(\theta,k^*(\theta;z))$ expressed in Eqs.\eqref{utility2}-\eqref{utility21}.
	%\begin{align*}
	%	\widetilde{\phi}(\theta)\!&=\!\int_{t_0-\omega+a}^{t_0}\!\frac{1}{\delta}n(z)L(t_0;z)\{x_0(z)\!+\![M_1(t_0;z)\theta\!+\! M_2(t_0;z)k^*(\theta;z) \!+\! M_3(t_0;z)]w_0\\
	%	&\phantom{ee}\!+\! N(t_0;z)y_0(z)\}^{\delta}dz \!+\!\frac{n_0L_0e^{\rho t_0}w_0^{\delta}}{\delta\!\left\{\! r \!-\!\rho\!-\!\delta\!\left[\! \gamma\!+\!\frac{1}{2}(\delta\!-\!1)\xi^2\!\right]\!\right\}\!}(M_{01}\theta\!+\! M_{02}k^*(\theta;t_0)\!+\! M_{03})^{\delta}\\
	%	&=\!\int_{t_0\!-\!\omega+a}^{t_0}\!\frac{1}{\delta}n(z)L(t_0;z)\{x_0(z)\!+\![(M_1(t_0;z)\!-\! M_2(t_0;z)^+)\theta\!+\! M_2(t_0;z)^+m\\
	%	&\phantom{ee}\!+\! M_3(t_0;z)]w_0\!+\! N(t_0;z)y_0(z)\}^{\delta}dz+\!\frac{n_0L_0e^{\rho t_0}w_0^{\delta}}{\delta\!\left\{\! r \!-\!\rho\!-\!\delta\!\left[\! \gamma\!+\!\frac{1}{2}(\delta\!-\!1)\xi^2\!\right]\!\right\}\!}[(M_{01}-M_{02}^+)\theta\\
	%	&\phantom{ee}\!+\! M_{02}^+m\!+\! M_{03}]^{\delta}.
	%\end{align*}
	
	Taking $k=k^*(\theta,z)$ in Eq.\eqref{equ19}, we establish the admissible scope of $\theta$ as follows:
	\begin{align*}
		\left\{
		\begin{array}{lr}
			x_0(z)\!+\![(M_1(t_0;z)\!-\! M_2(t_0;z)^+)\theta\!+\! M_2(t_0;z)^+m \!+\! M_3(t_0;z)]w_0\!+\! N(t_0;z)y_0(z)\geq 0,\\
			\phantom{eeeeeeeeeeeeeeeeeeeeeeeeeeeeeeeeeeeeeeeeeeeeeeeeeeeeeeeee}\forall ~z\in[t_0-\omega+a,t_0],\\
			\theta\in[0,m].
		\end{array}
		\right.
	\end{align*}
	
	Denote
	\begin{align*}
		&\underline{\theta}\triangleq\sup_{\{z:M_1(t_0;z)\!-\! M_2(t_0;z)^+>0\}}-\frac{x_0(z)+N(t_0;z)y_0(z)+(M_2(t_0;z)^+m+M_3(t_0;z))w_0}{(M_1(t_0;z)-M_2(t_0;z)^+)w_0},\\
		&\bar{\theta}\triangleq\inf_{\{z:M_1(t_0;z)\!-\! M_2(t_0;z)^+<0\}}-\frac{x_0(z)+N(t_0;z)y_0(z)+(M_2(t_0;z)^+m+M_3(t_0;z))w_0}{(M_1(t_0;z)-M_2(t_0;z)^+)w_0}.
	\end{align*}
We have the admissible scope of $\theta$ is $[0\vee\underline{\theta},m\wedge\bar{\theta}]$.
	
	Furthermore, analyzing the preference orderings in the flowchart of Fig.\ref{flow}, we obtain the specific forms of
	\begin{align*}
		&A_1\triangleq a+t_0-\{z:M_1(t_0;z)\!-\! M_2(t_0;z)^+>0\}=\{\zeta\in[a,\omega]:\widetilde{M_{1}}(\zeta)-\widetilde{M_{2}}(\zeta)^+>0\},\\
		&A_2\triangleq a+t_0-\{z:M_1(t_0;z)\!-\! M_2(t_0;z)^+<0\}=\{\zeta\in[a,\omega]:\widetilde{M_{1}}(\zeta)-\widetilde{M_{2}}(\zeta)^+<0\}.
	\end{align*}
	
	If $\tau\leq\zeta\leq\omega$,  we have $\widetilde{M_{1}}(\zeta)>0$, $\widetilde{M_{2}}(\zeta)=0$ and thus $[\tau,\omega]\subseteq A_1$.
	
	If $a\leq\zeta<\tau$, we need to analyze the cases of the branches in Fig. \ref{flow} respectively under heterogeneous characteristic parameters. Take the last branch as an example. Considering the fact that $\bar{\zeta}<\widetilde{\zeta}<\hat{\zeta}$, we have
	\begin{align*}
		\widetilde{M_{1}}(\zeta)-\widetilde{M_{2}}(\zeta)^+&=(\widetilde{M_{1}}(\zeta)-\widetilde{M_{2}}(\zeta))1_{\{\zeta<\bar{\zeta}\}}+\widetilde{M_{1}}(\zeta)1_{\{\zeta\geq\bar{\zeta}\}}\\
		&=\left\{\begin{array}{lr}
			\widetilde{M_{1}}(\zeta)-\widetilde{M_{2}}(\zeta)<0,\phantom{eee}\zeta<\bar{\zeta},\\
			\widetilde{M_{1}}(\zeta)<0,\phantom{eeeeeeeeeee}\bar{\zeta}\leq\zeta<\hat{\zeta},\\
			\widetilde{M_{1}}(\zeta)=0,\phantom{eeeeeeeeeee}\zeta=\hat{\zeta},\\
			\widetilde{M_{1}}(\zeta)>0,\phantom{eeeeeeeeeee}\zeta>\hat{\zeta}.
		\end{array}
		\right.
	\end{align*}
	Therefore, when the last branch occurs, i.e., $\widetilde{M}_2(a)\geq0,\ \widetilde{M}_2'(\tau)>0,\ \Lambda\leq\Lambda_{FP}  ,\ \hat{\zeta}>\widetilde{\zeta}$, we have $A_1=\left(\hat{\zeta},\omega\right]$ and $A_2=\left[a,\hat{\zeta}\right)$.
	
	Using similar analysis for all the branches in Fig. \ref{flow}, we obtain Eqs.\eqref{A1}-\eqref{A2}.

\vskip 10pt
\section{Estimation of the expected optimal wealth $\E X^{*}(t)$}\label{A.2}
Similar to the method in Appendix \ref{A.1}, we  obtain the expected optimal wealth $\mathrm{E}[X^*(t;z)]$ of the participants at different ages, where $t\in[t_0\vee z,z+\omega-a]$.

 If $t_0<z$, the PAYGO and EET contribution rates are always $\theta^*$ and $k^*$. As such, applying the results of Appendix \ref{A.1} and changing $\theta_0$, $k_0$, $t_0$ into $\theta^*$, $k^*$, $t$,  we have
\begin{equation*}
	\mathrm{E}[Y(t;z,k^*)]
	=\left\{
	\begin{array}{lr}
		k^*W_0e^{\alpha t}\int_{z}^{t}e^{(\gamma-\alpha)s}ds,&z\in[t-\tau+a,t],\\
		\mathrm{E}[Y(z+\tau-a;z,k^*)],&z\in[t-\omega+a,t-\tau+a),
	\end{array}
	\right.
\end{equation*}
and
\begin{align*}
    \mathrm{E}X^*(t;z)%&=\mathrm{E}[X^*(t;z,k^*,\theta^*)]\\
    &=-M(t;z,\theta^*,k^*)W_0e^{\gamma t}-N(t;z)\mathrm{E}[Y(t;z,k^*)]\\
    %k^*W_0\left\{
    %e^{\alpha t}\int_{z}^{t}e^{(\gamma-\alpha)u}du1_{\{t\in[z,z+\tau-a]\}}\right.\\
    %&\left.+\frac{(z+\omega-a)-t}{\omega-\tau}e^{\alpha (z+\tau-a)}\int_{z}^{z+\tau-a}e^{(\gamma-\alpha)u}du1_{\{t\in[z+\tau-a,z+\omega-a]\}}\right\}\\
    &+\left[\frac{L(t;z)}{L(z;z)}\right]^{\frac{1}{1-\delta}}M(z;z,\theta^*,k^*)W_0e^{\gamma t}e^{\left[-\gamma+\frac{r}{1-\delta}+\frac{2-\delta}{2(1-\delta)^2}\nu^2\right](t-z)}.
\end{align*}

 If $z\leq t_0\leq z+\omega-a$, the government reselects the optimal PAYGO and EET contribution rates at $t_0$. Under this condition, resolving SDEs \eqref{equ1}, (\ref{equ2}) and taking expectation, we have

If $z\leq t_0\leq z+\tau-a$,
% $$Y(t_0;z,k_0,k^*)=\int_{z}^{t_0}e^{\left(\alpha-\frac{1}{2}\beta^2\right)(t_0-s)+\beta(B(t_0)-B(s))}k_0W(s)ds,$$
%and
%\begin{equation*}
%    Y(t;z,k_0,k^*)\!=\!\left\{\!
%    \begin{array}{lr}
%        \exp{\{(\alpha-\frac{1}{2}\beta^2)(t-t_0)\!+\!\beta(B(t)-B(t_0))\}}Y(t_0)\\
%        +\int_{t_0}^{t}\exp{\{(\alpha-\frac{1}{2}\beta^2)(t-s)+\beta(B(t)-B(s))\}}k^*W(s)ds,\\
%        \phantom{eeeeeeeeeeeeeeeeeeeeeeeeeeeeeeeeee}t\in[t_0,z+\tau-a],\\
%        \frac{(z+\omega-a)-t}{\omega-\tau}Y(z+\tau-a),\phantom{eeeeeeeeeeee}t\in(z+\tau-a,z+\omega-a].
%    \end{array}
%    \right.
%\end{equation*}
%Thus
\begin{equation*}
	\mathrm{E}[Y(t;z,k_0,k^*)]=\left\{
	\begin{array}{lr}
		W_0e^{\alpha t}\left[k_0\int_{z}^{t_0}e^{(\gamma-\alpha)s}ds+k^*\int_{t_0}^{t}e^{(\gamma-\alpha)s}ds\right],\phantom{e}t\in[t_0,z+\tau-a],\\
		\mathrm{E}[Y(z+\tau-a;z,k_0,k^*)],\phantom{eeeeeeeeee}t\in(z+\tau-a,z+\omega-a).
	\end{array}
	\right.
\end{equation*}

If $z+\tau-a<t_0\leq z+\omega-a$, %$$Y(z+\tau-a;z,k_0,k^*)=\int_{z}^{z+\tau-a}e^{(\alpha-\frac{1}{2}\beta^2)(t-s)+\beta(B(t)-B(s))}k_0W(s)ds,$$
%and
%\begin{align*}
%    Y(t;z,k_0,k^*)=\frac{(z+\omega-a)-t}{\omega-\tau}Y(z+\tau-a), t\in[t_0,z+\omega-a].
%\end{align*}
%Thus
\begin{eqnarray*}
	\mathrm{E}[Y(t;z,k_0,k^*)]\!=\!k_0W_0e^{\alpha(z+\tau-a)}\!\int_{z}^{z+\tau-a}\! e^{(\gamma-\alpha)s}ds,~~t\in[t_0,z+\omega-a].
\end{eqnarray*}

Using martingale method, we derive the main results which are different from the ones in \cite{HLSY2021}.
 \begin{eqnarray*}
    \mathrm{E}[X^*(t;z)]=-M(t;z,\theta^*,k^*)W_0e^{\gamma t}-N(t;z)\mathrm{E}[Y(t;z,k_0,k^*)]+\left[\frac{\kappa e^{\frac{2-\delta}{2(\delta-1)}\nu^2t}}{e^{rt}L(t;z)}\right]^{\frac{1}{\delta-1}},
\end{eqnarray*}
where
\begin{align*}
\kappa%&\!=\!\left\{\!\frac{\mathrm{E}\left[x_0(z)+\int_{t_0}^{z+
%\omega-a}e^{-r(s-t_0)}\eta(s;t_0)[a(s;z,k^*,\theta^*)W(s)
%+\psi(s;z)Y(s)(1-\tau_2)]ds\right]}{\mathrm{E}\left[\int_{t_0}^{z
%+\omega-a}e^{-r(s-t_0)}\eta(s;t_0)\left[\frac{\eta(s;t_0)}{e^{r(s-t_0)}
%b(s;z)}\right]^{\frac{1}{\delta-1}}ds\right]}\!\right\}^{\delta-1}\!\\
        &=L(t_0;z)\left[x_0(z)+N(t_0;z)y_0(z)+M(t_0;z,\theta^*,k^*)W_0e^{\gamma t_0}\right]^{\delta-1}.
    \end{align*}
Thus,
    \begin{align*}
       \mathrm{E}[X^*(t;z)] =&-M(t;z,\theta^*,k^*)W_0e^{\gamma t}-N(t;z)\mathrm{E}[Y(t;z,k_0,k^*)]\\
        +&\!\left[\!\frac{L(t_0;z)}{L(t;z)}\!\right]\!^{\frac{1}{\delta-1}}\!\left[\!x_0(z)\!+\! N(t_0;z)y_0(z)\!+\! M(t_0;z,\theta^*,k^*)W_0e^{\gamma t_0}\!\right]\! e^{\left(\frac{r}{1-\delta}+\frac{2-\delta}{2(1-\delta)^2}\nu^2\right)t}.
    \end{align*}
\vskip 10pt
\section{Demographic model with ``baby-boom''}\label{babyboom}
%For other population models, the inverse of dependency ratio $\Lambda(t)\triangleq\frac{\Theta(t)}{\Sigma(t)}$ cannot be guaranteed independent of $t$. When $\Lambda(t)$ is independent of $t$, all the analysis in the original model still holds. When $\Lambda(t)$ is related to $t$, we can still get a similar closed-form solution like in Theorem \ref{them1}. However, different from Theorem \ref{them1}, we have
%\begin{align*}
%	M_1(t;z)&=\int_{t\vee(z+\tau-a)}^{z+\omega-a}\Lambda(s)e^{\left(r+\frac{1}{\sigma}\xi(\mu-r)-\gamma\right)(t-s)}ds-\int_{t\wedge(z+\tau-a)}^{z+\tau-a}(1-\tau_1)e^{\left(r+\frac{1}{\sigma}\xi(\mu-r)-\gamma\right)(t-s)}ds\\
%	&\triangleq M_1^+(t;z)-M_1^-(t;z),
%\end{align*}
%where
%\begin{align*}
%	&M_1^+(t;z)=\int_{t\vee(z+\tau-a)}^{z+\omega-a}\Lambda(s)e^{\left(r+\frac{1}{\sigma}\xi(\mu-r)-\gamma\right)(t-s)}ds,\\
%	&M_1^-(t;z)=\int_{t\wedge(z+\tau-a)}^{z+\tau-a}(1-\tau_1)e^{\left(r+\frac{1}{\sigma}\xi(\mu-r)-\gamma\right)(t-s)}ds.
%\end{align*}
%\vskip5pt
%The cohort's total contribution $M_1^-(t;z)$ is independent of $\Lambda(\cdot)$, while the cohort's total future benefit $M_1^+(t;z)$ depends on $\Lambda(\cdot)$. Therefore, the sign of $M_1(t;z)$ and $M_1(t;z)-M_2(t;z)$ cannot be solved analytically like in Theorem \ref{them2}-\ref{them3} and we need to solve it numerically. In addition, $M_2(t;z)$ measures the difference between EET pension and individual saving, independent of the population model, and thus Theorem \ref{them4} is valid.
%\vskip5pt
To depict the demographic changes more precisely, we explore a new demographic model with a one-off shock to depict the ``baby-boom" impacts during 1946$\sim$1964. Although there are no closed-form results, we can establish the similar preference orderings  and the optimal mix through numerical methods. We assume that the ``baby-boom" cohorts enter the labor market within the time interval $[t_1,t_2]$. The population of the new entrants $n(t)$ follows the Logistic growth model in the period of $t \in [t_1,t_2]$ and the exponential growth model when $t<t_1$ and $t>t_2$. Besides, the population growth rate takes constant value $\rho_1$ and $\rho_2$ when $t<t_1$ and $t>t_2$, respectively.  We have
\begin{align*}
	dn(t)=
	\left\{
	\begin{array}{llr}
		\rho_1n(t)dt,&t\leq t_1,\\
		\kappa(1-\frac{n(t)}{n_m})n(t)dt,&t_1<t\leq t_2,\\
		\rho_2n(t)dt,&t>t_2,
	\end{array}
	\right.
	\\
	n(t_1)=n_1,\phantom{eeeeeeeeeeeeeeeeeeeeeeeeeeeeee}
\end{align*}
where $\kappa$ is the growth rate without limitation, $n_m$ is the maximal population of the new entrants. Besides, $n(t)$ is continuous at $t_1$ and $t_2$, thus $n(t)$ has the following form:
\begin{align*}
	n(t)=\left\{
	\begin{array}{llr}
		n_1e^{\rho_1(t-t_1)},&t\leq t_1\\
		\frac{n_m}{1+(\frac{n_m}{n_1}-1)e^{-\kappa(t-t_1)}},&t_1<t\leq t_2,\\
		n(t_2)e^{\rho_2(t-t_2)},&t>t_2,
	\end{array}
	\right.
\end{align*}
where $n(t_2)=\frac{n_m}{1+(\frac{n_m}{n_1}-1)e^{-\kappa(t_2-t_1)}}$.
\vskip5pt
Under this new model, the inflow and the outflow of PAYGO pension are as follows:
\begin{align*}
	\Theta(t)&=\int_{a}^{\tau}n(t-u+a)s(u)du\\
	&=\left\{
	\begin{array}{lr}
		\int_{a}^{\tau}n(t_1)e^{\rho_1(t-t_1-u+a)}e^{-A(u-a)-\frac{B}{\ln c}(c^u-c^a)}du,\phantom{eeeeeeeeeeeeeeeeeeeeeee}t\leq t_1,\\
		\int_{a}^{\tau}\big[n(t_1)e^{\rho_1(t-t_1-u+a)}1_{\{t-u+a\leq t_1\}}+\frac{n_m}{1+(\frac{n_m}{n(t_1)}-1)e^{-\kappa(t-t_1-u+a)}}1_{\{t_1<t-u+a<t_2\}}\\
		+n(t_2)e^{\rho_2(t-t_2-u+a)}1_{\{t-u+a\geq t_2\}}\big]e^{-A(u-a)-\frac{B}{\ln c}(c^u-c^a)}du,\phantom{e}t_1<t<t_2+\tau-a,\\
		\int_{a}^{\tau}n(t_2)e^{\rho_2(t-t_2-u+a)}e^{-A(u-a)-\frac{B}{\ln c}(c^u-c^a)}du,\phantom{eeeeeeeeeeeeeeee}t\geq t_2+\tau-a,
	\end{array}
	\right.
\end{align*}
\begin{align*}
	\Sigma(t)&=\int_{\tau}^{\omega}n(t-u+a)s(u)du\\
	&=\left\{
	\begin{array}{lr}
		\int_{\tau}^{\omega}n(t_1)e^{\rho_1(t-t_1-u+a)}e^{-A(u-a)-\frac{B}{\ln c}(c^u-c^a)}du,\phantom{eeeeeeeeeeeeeeee}t\leq t_1+\tau-a,\\
		\int_{\tau}^{\omega}\big[n(t_1)e^{\rho_1(t-t_1-u+a)}1_{\{t-u+a\leq t_1\}}+\frac{n_m}{1+(\frac{n_m}{n(t_1)}-1)e^{-\kappa(t-t_1-u+a)}}1_{\{t_1<t-u+a<t_2\}}\\
		+n(t_2)e^{\rho_2(t-t_2-u+a)}1_{\{t-u+a\geq t_2\}}\big]e^{-A(u-a)-\frac{B}{\ln c}(c^u-c^a)}du,\phantom{e}\\
		\phantom{eeeeeeeeeeeeeeeeeeeeeeeeeeeeeeeeeeeeeeeeeeeee}t_1+\tau-a<t<t_2+\omega-a,\\
		\int_{\tau}^{\omega}n(t_2)e^{\rho_2(t-t_2-u+a)}e^{-A(u-a)-\frac{B}{\ln c}(c^u-c^a)}du,\phantom{eeeeeeeeeeeeeeee}t\geq t_2+\omega-a.
	\end{array}
	\right.
\end{align*}
The influence of the ``baby-boom'' starts at time $t_1$ and ends at time $t_2+\omega-a$. The inverse of dependency ratio $\Lambda(t)\triangleq\frac{\Theta(t)}{\Sigma(t)}$ is related to time $t$.
\vskip 10pt
When $\Lambda(t)$ is related to $t$, the main results in the baseline model are still valid. We can  obtain a similar closed-form solution as in Theorem \ref{them1}. However, different from Theorem \ref{them1}, we have
\begin{align*}
	M_1(t;z)&=\int_{t\vee(z+\tau-a)}^{z+\omega-a}\Lambda(s)e^{\left(r+\frac{1}{\sigma}\xi(\mu-r)-\gamma\right)(t-s)}ds-\int_{t\wedge(z+\tau-a)}^{z+\tau-a}(1-\tau_1)e^{\left(r+\frac{1}{\sigma}\xi(\mu-r)-\gamma\right)(t-s)}ds\\
	&\triangleq M_1^+(t;z)-M_1^-(t;z),
\end{align*}
where
\begin{align*}
	&M_1^+(t;z)=\int_{t\vee(z+\tau-a)}^{z+\omega-a}\Lambda(s)e^{\left(r+\frac{1}{\sigma}\xi(\mu-r)-\gamma\right)(t-s)}ds,\\
	&M_1^-(t;z)=\int_{t\wedge(z+\tau-a)}^{z+\tau-a}(1-\tau_1)e^{\left(r+\frac{1}{\sigma}\xi(\mu-r)-\gamma\right)(t-s)}ds.
\end{align*}
\vskip5pt
The cohort's total contribution $M_1^-(t;z)$ is independent of $\Lambda(\cdot)$, while the cohort's total benefit $M_1^+(t;z)$ depends on $\Lambda(\cdot)$. Therefore, we cannot determine the sign of $M_1(t;z)$ and $M_1(t;z)-M_2(t;z)$  analytically like in Theorems \ref{them2}-\ref{them3} and we need to solve it numerically.

\vskip5pt
The parameter settings are consistent with the ones in the previous baseline model. For the new parameters, we set $t_1=-40$, $t_2=-20$, $n_1=11.91$, $n_m=100$, $\kappa=0.05$, $\rho_1=-0.0025$, $\rho_2=-0.005$. Based on these parameters, the new dependency ratio varies from $\Lambda^{-1}=0.6738$ to $\Lambda^{-1}=0.7271$. The new critical age of PAYGO pension and individual savings is $\hat{\zeta}=36.9301$. The new critical age of PAYGO and EET pensions is $\widetilde{\zeta}=46.6149$.    All non-retired participants prefer EET pension to individual savings. The ``baby-boom'' briefly reduces the pressure of labor shortage and thus makes PAYGO pension more attractive. However, the impacts of ``baby-boom'' on the outcomes are limited.

\setcounter{figure}{0}
\begin{figure}
	\includegraphics[totalheight=5.5cm]{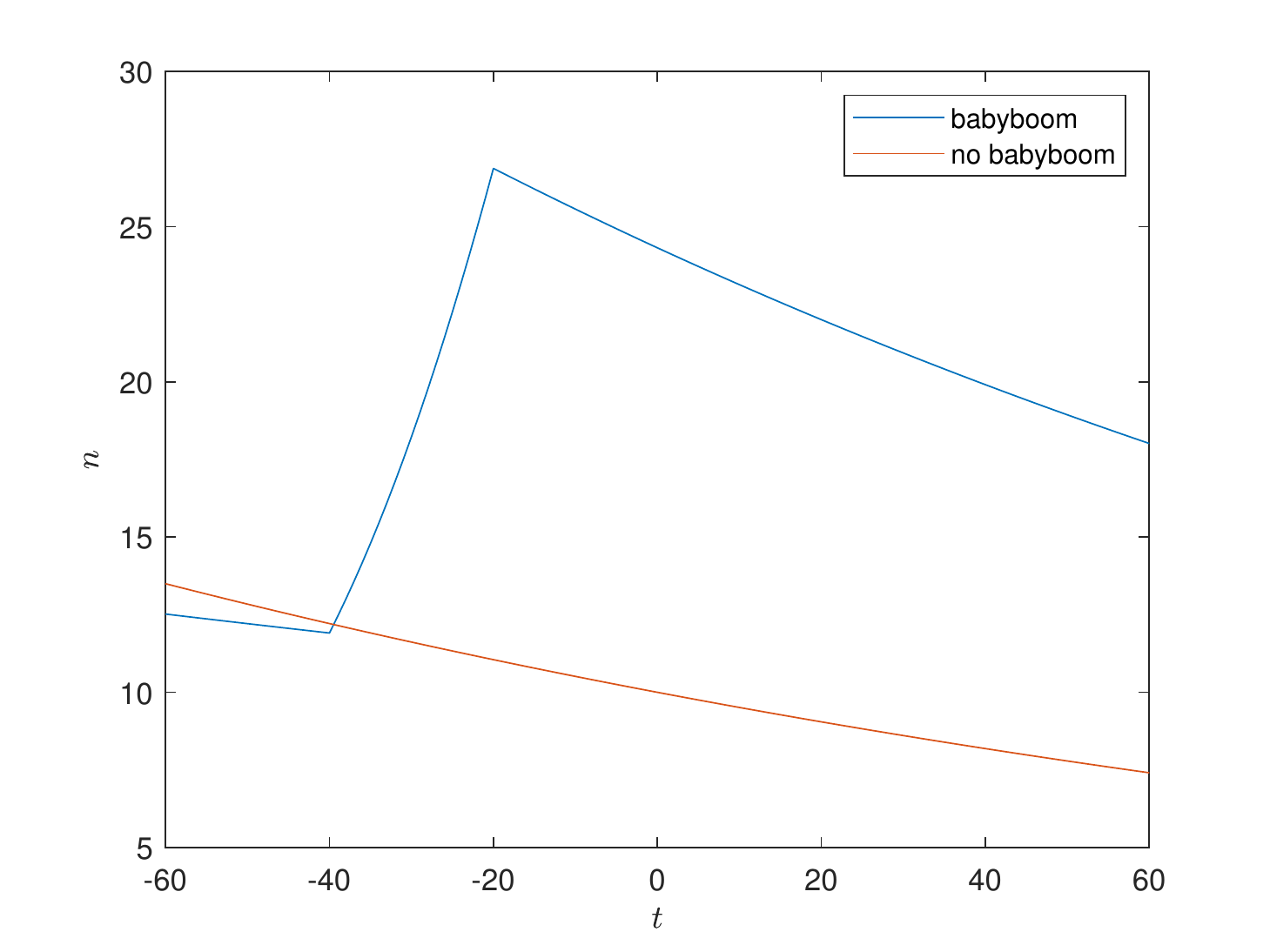}
	\caption{Population of new entrants $n(t)$ in two cases}
	\label{figb1}
\end{figure}
\vskip5pt
Fig. \ref{figb1} shows the population of new entrants $n(t)$ in two cases. In the case without ``baby-boom'', the population declines at a constant negative growth rate. In the ``baby-boom'' case, the population of new entrants goes through a phase of rapid rise, which depicts baby-boomers entering the labor market during the time interval $[t_1,t_2]$. Moreover, the new entrants' population doubles at the end of the ``baby-boom'' period $t_2$.
\begin{figure}
	\includegraphics[totalheight=5.5cm]{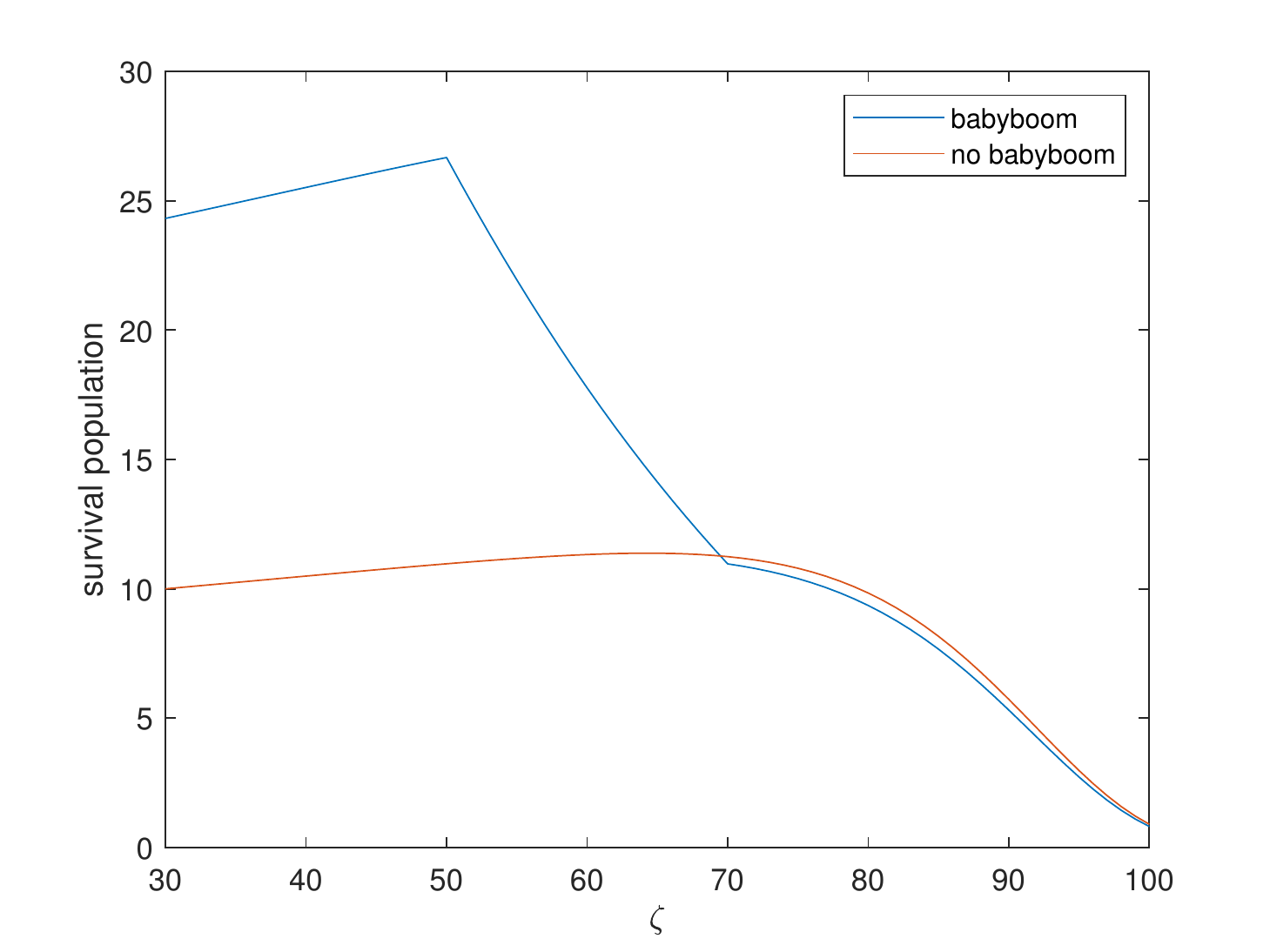}
	\caption{Survival population of different ages at time $t_0$ in two cases}
	\label{figb2}
\end{figure}
\vskip5pt
In Fig. \ref{figb2}, we exhibit the survival population of different ages at time $t_0$ in two cases. Compared to the case without ``baby-boom'',  the working population rises dramatically and the retired population is almost unchanged in the ``baby-boom'' case. Because baby-boomers are between 50 and 70 years old at time $t_0$. They are at the beginning of retirement and most are still working. It is worth noting that when baby boomers begin to retire, the ``demographic dividend" gradually vanishes.
\vskip5pt
\begin{figure}
	\includegraphics[totalheight=5.5cm]{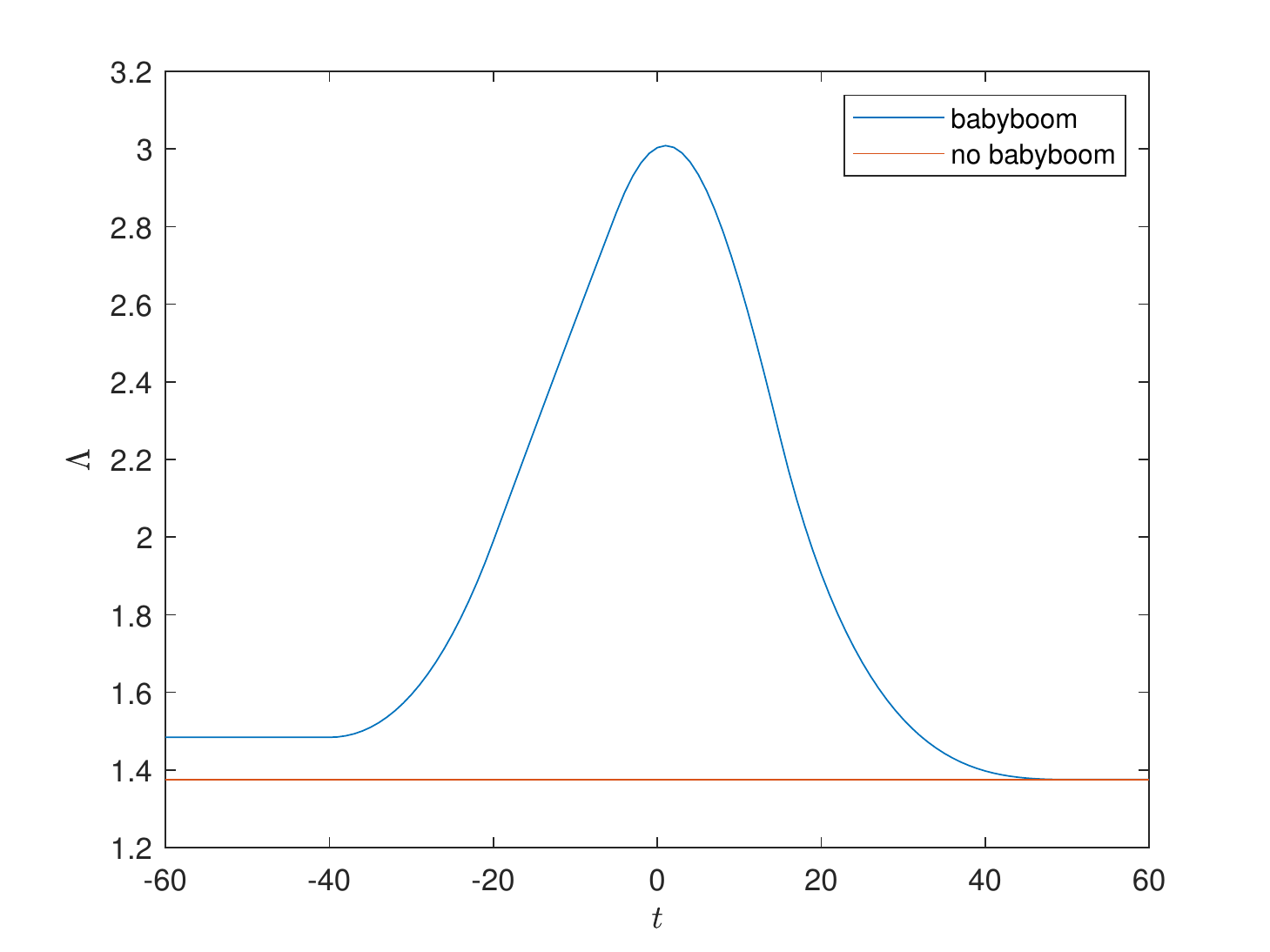}
	\caption{Inverse of dependency ratio $\Lambda(t)$ in two cases}
	\label{figb3}
\end{figure}
In Fig. \ref{figb3}, we study the inverse of dependency ratio $\Lambda(t)$ in two cases. In the case without ``baby-boom'', the inverse of dependency ratio is a constant independent of time $t$. While in the ``baby-boom'' case, the inverse of dependency ratio has a hump shape and the peak is more than twice in the case without ``baby-boom''. This huge peak reflects the abundant labor supply brought by baby-boomers. Moreover, consistent with Fig. \ref{figb2}, the peak is around $t_0$. That is, the ``demographic dividend" begins to disappear at time $t_0$. Therefore, the government had better reselect the pension contribution rates at time $t_0=0$ to cope with this trend. Interestingly, the peak of $\Lambda(t)$ appears later than that of $n(t)$. Because $\Lambda(t)$ is the proportion of the working population to the retirement population, it takes time to reach the peak after the baby-boomers enter the labor market. Notably, $\Lambda(t)$ has impacts on the preference of the cohorts by affecting the total benefit $M_1^+(t;z)$. However, the impacts of ``baby-boom'' will fade away in 50 years. Thus, its impact on the preference of the young cohorts who will retire after a long time is limited.
\begin{figure}
	\centering
	\begin{subfigure}{0.49\textwidth}
		\centering
		\includegraphics[width=1\textwidth]{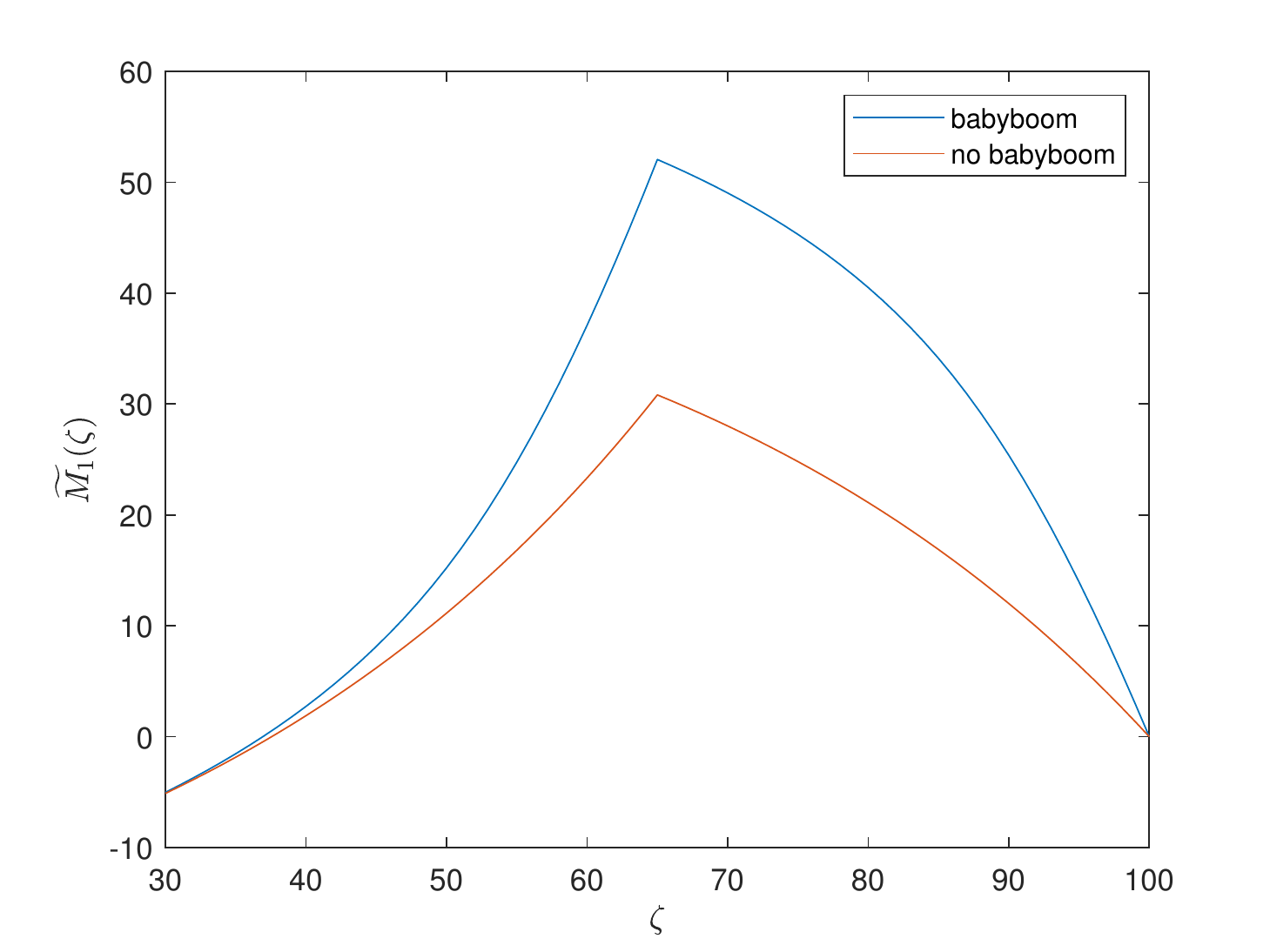}
		%\subcaption{Preference between PAYGO and individual saving}
	\end{subfigure}
	\begin{subfigure}{0.49\textwidth}
		\centering
		\includegraphics[width=1\textwidth]{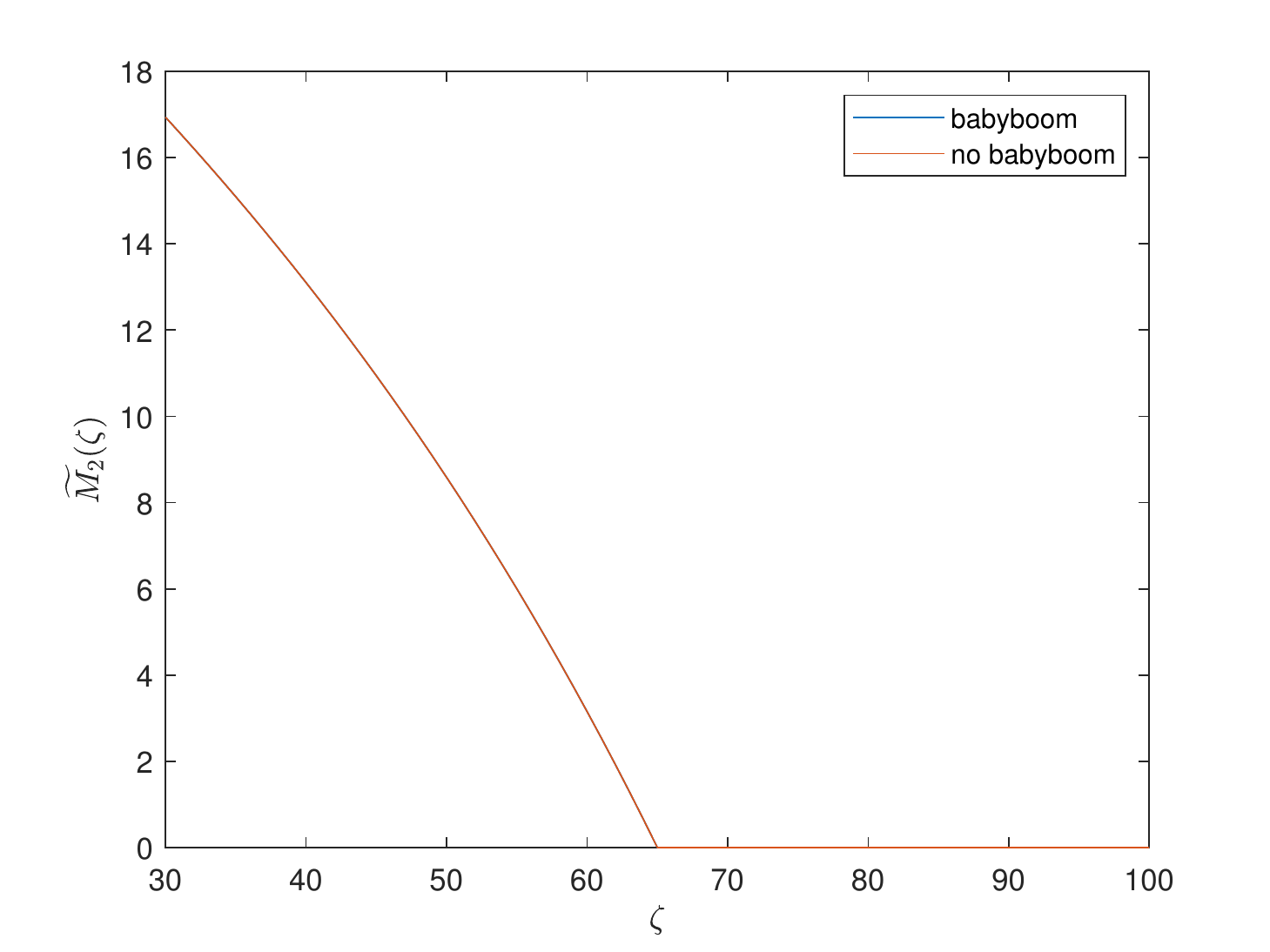}
	\end{subfigure}
	\begin{subfigure}{0.49\textwidth}
		\centering
		\includegraphics[width=1\textwidth]{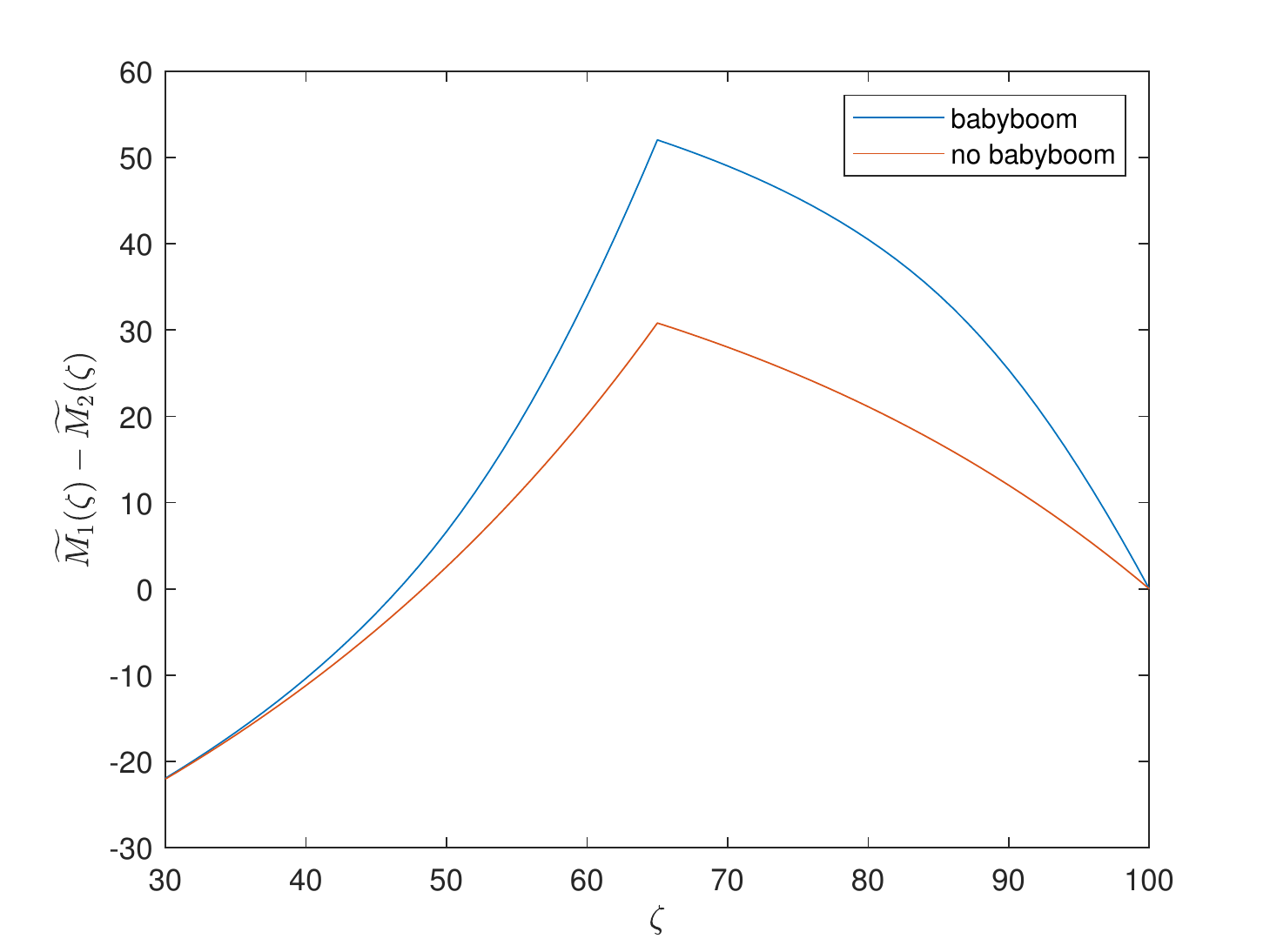}
		%\subcaption{Preference between PAYGO and EET}
	\end{subfigure}
	
	\caption{Preference among PAYGO, EET and individual savings in two cases}
	\label{figb4}
\end{figure}
\vskip5pt
Fig. \ref{figb4} shows the preference among PAYGO, EET and individual savings in two cases. In the first figure, although the ``baby-boom" case exhibits a sharper trend, $\widetilde{M_1}(\zeta)$ in two cases have similar hump shapes and both of them have unique null points $\hat{\zeta}\in[36,38]$. And the null point for the ``baby-boom'' model is smaller than the one for the baseline model. Thus, more younger cohorts prefer PAYGO pension in the ``baby-boom'' model. We also observe that ``baby-boom'' has great impacts on the utility of the older cohorts, but has less impacts on the younger cohorts. It is because that the younger cohorts will retire after the ``baby-boom'' impact vanishes. The similar results are valid in the third figure. Moreover, preference between EET pension and individual savings is independent of the demographic model. Therefore, the lines of $\widetilde{M_2}(\zeta)$  of the two cases coincide in the second figure.

%\begin{figure}
%	\centering
%	\begin{subfigure}{0.49\textwidth}
%		\centering
%		\includegraphics[width=1\textwidth]{V_fkb.pdf}
%	\end{subfigure}
%	\begin{subfigure}{0.49\textwidth}
%		\centering
%		\includegraphics[width=1\textwidth]{V_fthetab.pdf}
%	\end{subfigure}
%	\begin{subfigure}{0.49\textwidth}
%		\centering
%		\includegraphics[width=1\textwidth]{V_fkpb.pdf}
%	\end{subfigure}
%	\caption{Variation tendency of the value function with ``baby-boom"}
%	\label{figb5}
%\end{figure}
%\vskip5pt
%In Fig. \ref{figb5}, we study the variation tendency of the value function with ``baby-boom".
%%It is similar to the baseline model, except for the critical ages to change the monotonicity of the value function.
%Compared with the baseline model, the critical ages are only shifted to the left by less than 1 year. The impact of ``baby-boom'' on the critical age is relatively limited.}
%%Fixed $k=k^*$, the value function of young cohorts decreases as $\theta$ increases and the value function of old cohorts is opposite. Fixed $\theta=\theta^*$, the variation tendency is contrary to the first figure. The value function of young cohorts increases as $k$ increases and the value function of old cohorts decreases as $k$ increases. Thus, fixed $k+\theta=k^*+\theta^*$, the trend is sharper because of the simultaneous increase in $k$ and decrease in $\theta$.
\vskip 20pt
\setcounter{equation}{0}

%\begin{appendices}
%\renewcommand{\thechapter}{\Alph{chapter}.}
%    \chapter{Firstappendix}
\end{document}